\documentclass[final]{siamart0516}
\RequirePackage[numbers]{natbib}

\usepackage{graphics} 
\usepackage{epsfig} 
\usepackage{color}
\usepackage{url}
\usepackage[export]{adjustbox}
\usepackage{amsfonts}
\usepackage{amssymb}
\usepackage{amsmath}
\usepackage{comment}
\usepackage{caption}
\usepackage{subcaption}
\usepackage{graphicx}
\usepackage{float}
\usepackage{capt-of}
\usepackage[section]{placeins}
\usepackage{bbm}
\usepackage{multicol}
\usepackage [autostyle, english = american]{csquotes}
\MakeOuterQuote{"}

\def\R{\mathbb R}
\def\N{\mathbb N}

\def\P{\mathbb P}

\def\ONE{\mathbbm 1}

\def\tr{\mbox{\rm{Tr}}}
\def\Id{\mbox{Id}}

\newsiamremark{remarks}{Remarks}
\newsiamremark{remark}{Remark}
\newcommand{\beginsupplement}{%
        \setcounter{table}{0}
        \setcounter{page}{1}
        \renewcommand{\thetable}{S\arabic{table}}%
        \setcounter{figure}{0}
        \renewcommand{\thefigure}{S\arabic{figure}}%
     }

\usepackage{lipsum}
\usepackage{amsfonts}
\usepackage{graphicx}
\usepackage{epstopdf}
\usepackage{algorithmic}
\ifpdf
  \DeclareGraphicsExtensions{.eps,.pdf,.png,.jpg}
\else
  \DeclareGraphicsExtensions{.eps}
\fi

\newcommand{\TheTitle}{Eigenvalues of random  matrices with isotropic Gaussian noise and the design of  Diffusion Tensor Imaging experiments.} 
\newcommand{\TheAuthors}{Dario Gasbarra, Sinisa Pajevic, Peter J. Basser}

\headers{\TheTitle}{\TheAuthors}

\title{{\TheTitle}\thanks{Submitted to the editors 12.10.2016.
}}

\author{
  Dario Gasbarra\thanks{Department of Mathematics and Statistics, University of Helsinki, P.O. Box 68 FI-00014
    Finland (\email{dario.gasbarra@helsinki.fi})}
  \and
   Sinisa Pajevic \thanks{National Institutes of Health (NIH), Mathematical and Statistical Computing Lab, 12 South
     Drive Bethesda MD 20892 USA  (\email{pajevic@nih.gov})}
   \and
   Peter J. Basser\thanks{     Eunice Kennedy Shriver National Institute of Child Health and Human Development
     (NICHD), National Institutes of Health (NIH), 13 South Drive,
    MSC 5772, Bethesda, MD 20892-5772 USA (\email{pjbasser@helix.nih.gov})}
}

\usepackage{amsopn}

\begin{document}
\maketitle
\thispagestyle{plain}
\pagestyle{plain}

\begin{abstract}
  Tensor-valued and matrix-valued measurements of different physical properties are increasingly available in
  material sciences and medical imaging applications. The eigenvalues and eigenvectors of such multivariate data provide novel
  and unique information, but at the cost of requiring  a more complex statistical analysis. In this work we derive the
  distributions of eigenvalues and eigenvectors in the special but important case of $m \times m$ symmetric
  random matrices, $D$, observed with isotropic matrix-variate Gaussian noise. The properties of these
  distributions depend strongly on the symmetries of the mean tensor/matrix, $\bar D$. When $\bar D$ has repeated
  eigenvalues, the eigenvalues of $D$ are not asymptotically Gaussian, and repulsion is observed between the
  eigenvalues corresponding to the same $\bar D$ eigenspaces. We apply these results to diffusion tensor
  imaging (DTI), with $m=3$, addressing an important problem of detecting the symmetries of the diffusion tensor,
  and seeking an experimental design that could potentially yield an isotropic Gaussian distribution. In the
  3-dimensional case, when the mean tensor is spherically symmetric and the noise is Gaussian and isotropic, the
  asymptotic distribution of the first three eigenvalue central moment statistics is simple and can be used to test
  for isotropy. In order to apply such tests, we use quadrature rules of order $t \ge 4$ with constant weights on the
  unit sphere to design a DTI-experiment with the property that isotropy of the underlying true tensor implies
  isotropy of the Fisher information. We also explain the potential implications of the methods using simulated DTI
  data with a Rician noise model.
\end{abstract}



\section{INTRODUCTION} \label{section:intro}

Tensors of second and higher order are ubiquitous in the physical
sciences. Some examples include the moment of inertia tensor;
electrical, hydraulic, and thermal conductivity tensors; stress and
strain tensors, etc. One key advance in the field of tensor measurement
was the advent of diffusion tensor imaging (DTI), a magnetic resonance
based imaging technique that provides an estimate of a second order
diffusion tensor in each voxel within an imaging volume\cite{basser1994a,basser1994b}. This
effectively provides discrete estimates of a continuous or piece-wise
continuous tensor field within tissue and organs. With the possibility
of measuring tensors in millions of individual voxels within, for
example, a live human brain, there is a clear need for a statistical
framework to be developed to a) design optimal DTI experiments, b)
characterize central tendencies and variability in such data, and c)
provide a family of hypothesis tests to assess and compare tensors and
the quantities derived from them.

\subsection{TENSOR-VARIATE NORMAL DISTRIBUTION} \label{subsection:tensvar}

In DTI, a tensor $D$ is represented by a symmetric matrix $D=( D_{i,j} :
1\le i \le j \le 3)$ and it has been established that the measured 
tensor components $D_{ij}$, over multiple independent acquisitions from the same subject in the same voxel,
conform to
a multivariate normal distribution \cite{pajevic2003}. We previously proposed a normal
distribution for tensor-valued random variables that arise in DTI whose
precision and covariance structures could be written as fourth-order
tensors\cite{basser-pajevic03}:
\begin{equation*}  \label{isotropic_prior:0}
  p(D)\propto 
  \exp\biggl( - \frac 1 2 (D-\bar D): A : (D-\bar D) \biggr), 
\end{equation*}
where $A$ is a fourth-order precision tensor, $\bar D$ is the mean
tensor, and  "$:$" is a tensor contraction.

There are distinct advantages to analyzing tensor or tensor-field data in the laboratory coordinate system in which their
components are measured, and using the tensor-valued variates with a
fourth-order tensor precision tensor rather than writing the tensor as a vector and using a square covariance
matrix. For example, by retaining the tensor form it is easy to establish the conditions that the
statistical properties be coordinate independent, yielding a
isotropic fourth-order precision tensor
\begin{equation*} 
\label{isotropic_prior:1}
 A^{iso}_{ijkl}  = \lambda \delta_{ij}\delta_{kl}+\mu\left(\delta_{ik}\delta_{jl} + \delta_{il}\delta_{jk}\right),
\end{equation*}
which can be parameterized with only two constants, $\mu$ and $\lambda$. This form, if
achieved, can greatly simplify statistical analysis and is the focus of this paper. 

\medskip In the following sections, we switch from tensor to matrix notation \cite{basser-pajevic03}, as the
correspondence between the Gaussian tensor-variate and standard multivariate normal can be established using
appropriate conversion factors\cite{basser-pajevic07}. The outline of the paper is as follows. First, in this
section we state the properties for the $m$-dimensional isotropic Gaussian matrix. In section \ref{section:spectral}
we describe a spectral representation and change of variables applicable to general symmetric random matrices. In
section \ref{section:eigendistrib} we derive distributions for the eigenvalues and eigenvectors for the isotropic
Gaussian, while in section \ref{section:smallnoise} we obtain the analytical expressions in the limit of small noise
for different symmetries of the mean tensor $\bar D$. In the remaining sections, we focus on the application of
these results to DTI. In section \ref{section:sphericity} we develop a sphericity test, testing for the isotropy of
the diffusion tensor; in section \ref{section:rician} we study the isotropy of the  Fisher information and justify the use of 
spherical $t$-designs as gradient tables in DTI experimental design; and finally, in section
\ref{section:mcsim} we test many of the mathematical results and predictions using Monte Carlo simulations of DTI experiment.
The main theorems are proved in Appendix \ref{appendix:a}.

\subsection{ISOTROPIC GAUSSIAN MATRIX DISTRIBUTION} \label{subsection:isotropicGaussian}

Given a fixed symmetric matrix   $\bar D \in \R^{m\times m}$, it is shown in \cite{mallows},\cite{basser-pajevic03},
that the probability distribution of a $m\times m$ symmetric Gaussian random matrix     $D=( D_{ij} : 1\le i \le j
\le m)$ 
 is isotropic around $\bar D$ if and only if it has  density of the form
 \begin{align}  
  p( D)&= C_m(\mu,\lambda) 
 \exp\biggl( - \mu\tr( (D-\bar D)^2 )-
 \frac{ \lambda} 2  \{\tr(D-\bar D) \bigr\}^2    \biggr) ,
& \label{isotropic_prior:2}   \\   C_m( \mu,\lambda)&=2^{ (m-1)m/4}   \pi^{  -(m+1)m/4}  \mu^{(m+1)m/4}\sqrt{ 1+ \lambda m/(2\mu) },
& \end{align}
with precision parameter $\mu > 0$ and  interaction parameter $\lambda$ satisfying the constraint $\lambda m > - 2\mu$. 
To fix the ideas, when $m=3$ this corresponds to a Gaussian distribution
for the vectorized matrix
\begin{equation}
\mbox{vec}(D)=(D_{11} ,D_{22}, D_{33}, D_{12}, D_{13} ,D_{23 }),
\end{equation}
 with  mean $\mbox{vec}(\bar D)$  and  precision matrix 
 \begin{eqnarray}  \label{Omega_D:matrix:2nd} A(\mu, \lambda)= \left( \begin{matrix}\lambda + 2\mu & \lambda & \lambda & 0 & 0 & 0 \\
                     \lambda & \lambda + 2\mu & \lambda & 0 & 0 & 0 \\
\lambda & \lambda & \lambda + 2\mu  & 0 & 0 & 0 \\
0 & 0 & 0  & 4\mu & 0 & 0 \\
0 & 0 & 0  &0 & 4\mu & 0 \\
0 & 0 & 0 &0 &0 & 4\mu  
                    \end{matrix}
 \right) \; .
\end{eqnarray}
In particular  $(D_{ij}:1\le i < j\le m)$ are independent,
and $(D_{ii} :1\le i \le m)$ are negatively correlated for $\lambda>0$,
with covariance $\Sigma(\mu,\lambda)=A(\mu,\lambda)^{-1}$ where
\begin{subequations}
\begin{eqnarray*}  && \Sigma_{ij,ij}=
E\bigl( ( D_{ij}-\bar D_{ij} )^2\bigr) = (4\mu)^{-1}, \quad i\ne j,   \\ 
&&\Sigma_{ii,jj}=E\bigl(  ( D_{ii}-\bar D_{ii}) ( D_{jj}-\bar D_{jj} ) \bigr) 
=   \biggl( \delta_{ij}    -   \frac{\lambda }{ 2\mu + \lambda  m} \biggr) \frac 1 {2\mu}  \; .
\end{eqnarray*}
\end{subequations}
\begin{remark}  
          When $\bar D={\bf 0}$, $\lambda=0$ and
$\mu=1$ or,  depending on the scaling convention,  $\mu=1/2$, the random matrix distribution 
\eqref{isotropic_prior:2} is  known in the
literature as Gaussian Orthogonal Ensemble (GOE). The connection between general isotropic Gaussian matrices
and the GOE was first noticed in \cite{schwartzman}.
 The  fluctuations of the diagonal elements $\bigl((D_{ii}-\bar D_{ii}):1\le i\le m\bigr)$ are exchangeable
 and independent from the off-diagonal elements.
   \end{remark}
\section{SPECTRAL REPRESENTATION AND CHANGE OF VARIABLES} \label{section:spectral}
We summarize basic facts from the 
 random matrix literature \cite{dyson},\cite{metha},\cite{edelman},\cite{chikuse},\cite{forrester}.
A symmetric matrix $D \in \R^{m\times m}$ has spectral decomposition $D=O G O^{\top}$,
where $G$  is a diagonal matrix 
containing the $m$ eigenvalues $(\gamma_{1},\gamma_{2},\dots,\gamma_{m})\in \R^m$,  and 
$O=\bigl( O^{-1})^{\top}$ 
is an orthogonal matrix
with columns 
 corresponding to the  normalized eigenvectors. The orthogonal matrices form a compact  group ${\mathcal O}(m)$
 with respect to the matrix multiplication, which contains the special orthogonal group
${\mathcal S}{\mathcal O}(m)= \{ O \in {\mathcal O}(m) : \det(O)=1 \}$ of rotations.
 The $(m-1)m/2$ independent entries under the diagonal $(O_{ij}: 1\le j< i \le m )$ determine $O$, and 
 the eigenvalues are distinct for  symmetric matrices  
outside a set of   Lebesgue measure zero 
in $\R^{(m+1)m/2}$.
 The spectral decomposition is not unique,
 since $D= O G O^{\top}= R O P G P^{\top} O^{\top} R^{\top} $ for any permutation matrix $P$,
 and any $R=( R_{ij}=\pm \delta_{ij} )_{1\le i \le j \le m}$, which form the  subgroup ${\mathcal R}(m)$ of reflections with respect
 to the Cartesian axes, isomorphic to $\{1,-1\}^m$.
 In order to determine uniquely $O$ and $G$, we sort the eigenvalues
in descending order $\gamma_1 > \gamma_2 > \dots > \gamma_m$, and
impose, for  each column vector 
$(O_{1j}, O_{2j},\dots, O_{mj})^{\top}$,  $j=1, \dots ,m$, the condition that the
first encountered non-zero coordinate is positive, and denoted by ${\mathcal O}(m)^+$ the set of such matrices.
An $O\in {\mathcal O}(m)^+$ is a representative of the left coset $O {\mathcal R}(m)$. 
 The change of variables 
\begin{equation} \label{change:var}
D\mapsto X=( \gamma_{i}, O_{ij}:  1\le j< i \le m) ,
\end{equation}
has differential
%
\begin{eqnarray*} \label{change:variable}
  \prod_{1\le i \le j \le m} d D_{ij} = 
  |J(\gamma,O) |   \prod_{1\le i \le m}  d \gamma_i  \prod_{j<i} d O_{ij}  \; ,   
\end{eqnarray*}
%
\noindent where $J$ is the Jacobian of the inverse map $X\mapsto D$,  which is evaluated by means of  differential geometry. 
We consider a differentiable map $Y:\R^{m\times m} \to \R^{m\times m}  $.  The matrix differential can then be
written using  the chain rule
\begin{align*}
   dY_{ij}  = \sum_{k,h=1}^m \frac{\partial Y_{ij} }{ \partial X_{hk} } dX_{hk} \; ,
\end{align*}
and the wedge  product  acting on the transformed differentials is
\begin{align*}
  \bigl(  dY \bigr)^{\wedge}:= \bigwedge_{i,j=1}^m  dY_{ij}= 
  \det\biggl(   \frac{\partial Y_{ij} }{ \partial X_{hk} } \biggr)
  \bigwedge_{h,k=1}^m  dX_{hk} \;.
\end{align*}
Note that 
the wedge product is taken over the independent entries of the matrix,
for example if $X$ is symmetric
\begin{align*}
  ( dX )^{\wedge } =  \bigwedge_{1\le h \le k \le m } dX_{hk}\; ,
\end{align*}
and when  $X$ is skew-symmetric
\begin{align*}
  ( dX )^{\wedge } =  \bigwedge_{1\le h < k \le m } dX_{hk}\; .
\end{align*}
 The wedge product is also anticommutative, meaning that $dx \wedge dy = -dy \wedge dx$.
However when we compute volume elements, we always choose an ordering of the wedge product producing  a non-negative volume.
The Jacobian calculation is  based on the following result:
\begin{proposition}
(Prop.1.2 in \cite{forrester}) When $A,D$  are $m\times m$ matrices and $D$ is symmetric,
\begin{align} \label{eq:transformation}
    ( A^{\top} dD A )^{\wedge} = \det(A)^{m+1} ( dD )^{\wedge} \; . \end{align} \end{proposition}
Since
$O^{\top} O=I$,  it follows that the matrix differential $    O^{\top} dO=-     d O^{\top} O$ is  skew-symmetric.
We also have
\begin{align*}
  d D=  dO \; G O^{\top} + O dG O^{\top} + O G dO^{\top}
\end{align*}
and
\begin{align*}
  O^{\top} d D  \;  O = dG + O^{\top} dO  G  - G O^{\top} dO , 
\end{align*}
where the   differential matrix on the right hand side has  diagonal entries $ d\gamma_i$
and off diagonal entries
\begin{align*}
    ( O^{\top} dO )_{ij} (\gamma_i - \gamma_j) =  ( \gamma_i- \gamma_j ) \sum_{k=1}^m O_{ki}d O_{kj}\quad    i \ne j \; .
\end{align*}
By 
using 
property \eqref{eq:transformation}, we obtain
\begin{align} & \label{eq:jacobian:transformation} \bigl( dD \bigr)^{\wedge}=
\bigl(  O^{\top}  dD O \bigr)^{\wedge} =   V(\gamma) \bigwedge_{i=1}^m d\gamma_i
\biggl(    O^{\top} dO \biggr)^{\wedge} \;, \quad \mbox{ where }
     & \\  \nonumber &
 V(\gamma) 
 = \left \vert \begin{matrix} 1 & \gamma_1  & \gamma_1^2 & \dots & \gamma_1^{m-1} \\
                            1 & \gamma_2  & \gamma_2^2    & \dots & \gamma_2^{m-1}         \\
                          & & & \ddots & \\  
                             1 & \gamma_m  & \gamma_m^2  & \dots & \gamma_m^{m-1} 
                           \end{matrix} 
                           \right\vert = \prod_{1\le i < j \le m} (\gamma_i-\gamma_j)
\end{align} 
 is the Vandermonde determinant.
The wedge product  $\bigl( O^{\top} dO \bigr)^{\wedge}$ defines a uniform  
measure on ${\mathcal O}(m)$ which is invariant under the group action, and the Haar probability measure is given by
\begin{align*} H_m( dO) = \frac 1 { \mbox{Vol}( {\mathcal O}(m) ) }      ( O^{\top} dO )^{\wedge} 
\end{align*}
is 
obtained by normalizing with  the volume measure (Corollary 2.1.16 in \cite{muirhead})                           
\begin{align*}
\mbox{Vol}( {\mathcal O}(m) ) =2^m\pi^{(m+1)m/4 }\prod\limits_{j=1}^{m} \Gamma\bigl( j/2 \bigr)^{-1}= 2^m 
\mbox{Vol}( {\mathcal O}(m)^+ ) = 2 \mbox{Vol}( {\mathcal S} {\mathcal O}(m) )\;.
\end{align*}
 We  rewrite \eqref{eq:jacobian:transformation} as
\begin{align}  \label{measure:factorization}
  \bigwedge_{i\le j}dD_{ij}=   \mbox{Vol}( {\mathcal O}(m)) V(\gamma) d\gamma \; \times \; H_m(dO)\; , \quad  O\in {\mathcal O}(m)^+,
  \quad \gamma_1>\gamma_2> \dots > \gamma_m \; .   
\end{align}
\section{EIGENVALUE AND EIGENVECTOR DISTRIBUTION} \label{section:eigendistrib}
\subsection{Zero-Mean Isotropic Gaussian Matrix} \label{zero:mean}
We consider first a  zero-mean symmetric random matrix $D$ 
with  isotropic  Gaussian distribution  \eqref{isotropic_prior:2},
where $\bar D={\bf 0}$. This is an important special case to consider. While it does not satisfy 
the physical requirement that the eigenvalues of a diffusion (or other transport) tensor are all non-negative, 
it illustrates the mathematical machinery necessary 
to derive a closed-form expression for the resulting distribution of tensor eigenvalues.
From the spectral decomposition  $D=O GO^{\top}$, it follows 
by using the change of variables \eqref{change:var} in the density \eqref{isotropic_prior:2},
that 
$O$ is independent from  $G$  and
represents a random rotation distributed according to the constrained 
probability  
\begin{align*}
2^m {\bf 1}\bigl( O \in {\mathcal O}^+(m) \bigr ) H_m(dO),
\end{align*}
and 
the ordered  $D$-eigenvalues   have joint density on 
$\bigl\{ \gamma \in \R^m : \gamma_1  > \dots >\gamma_m \bigr \}$
 \begin{eqnarray}  \label{eigenvalue:density} &&
  q_0(\gamma)=   Z_m(\mu,\lambda)  V(\gamma)
 \exp\biggl( - \biggl( \mu + \frac{\lambda}2 \biggr) \sum_{i=1}^m\gamma_i^2
 - \lambda \sum_{1\le i < j \le m} \gamma_i \gamma_j \biggr),
 \end{eqnarray}
 with normalizing constant
\begin{subequations}
 \begin{align} & \label{eigenv:normalizing:const}  
Z_m( \mu,\lambda) =  Z_m(1,0) \mu^{m(m+1)/4 } \sqrt{ 1+\lambda m/ (2\mu) }  \;, & \\ &
  Z_m( 1,0) =2^{-m}\mbox{Vol}( {\mathcal O}(m)) C_{m}(1,0)= 2^{m(m-1)/4} \prod\limits_{l=1}^m  \Gamma( l/2)^{-1}  & \; .
\end{align}
\end{subequations}
\begin{remark}
The density  \eqref{eigenvalue:density} is not generally Gaussian,
since the Vandermonde determinant induces repulsion between the 
 eigenvalues, which are
 never independent, even in the case with $\lambda=0$
 and the diagonal elements    $D_{ii}$ are independent.
 When  $\lambda=0$, after rescaling, \eqref{eigenvalue:density} is the well known
GOE eigenvalue density, which plays a special role below 
(see Theorem \ref{spectral:clustering}). For $m=3$,  $Z_3(\mu, \lambda)= 4\pi^{-1}\mu^{5/2} \sqrt{ 2\mu +3\lambda }$.
\end{remark}
\subsection{General Case}
\begin{theorem} \label{main:thm}
 Let $\bar D\in R^{m\times m}$ be a symmetric matrix with a spectral decomposition
 $\bar D=\bar O \bar G \bar O^{\top}$, where
$\bar G=\mbox{diag}( \bar \gamma_1 , \bar \gamma_2 , \dots,\bar \gamma_m)$,
$ \bar \gamma_1\ge \bar \gamma_2\ge \dots \ge \bar \gamma_m$ are the ordered eigenvalues of  $\bar D$,
and $\bar O\in {\mathcal O}(m)^+$ (which is not uniquely determined when there are repeated eigenvalues),
and let $D$ be a symmetric $m\times m$ Gaussian  matrix   with  density  \eqref{isotropic_prior:2} isotropic around
the mean value $\bar D$. Then,
the ordered $D$-eigenvalues $\gamma_1> \gamma_2 > \dots > \gamma_m$ 
have joint density 
\begin{multline}
 q_{\bar \gamma}( \gamma)= 
Z_m(\mu,\lambda) V( \gamma ) 
  \exp\biggl( -  \sum_{i,j=1}^m  \biggl(  \delta_{i,j }\mu + \frac{\lambda} 2 \biggr)
  (\gamma_i -\bar \gamma_i)(\gamma_j-\bar \gamma_j) \biggr)  \\
\times 
\exp\bigl( - 2\mu \sum_{i=1}^m  \gamma_i\bar \gamma_i \bigr) {\mathcal I}_m( 2\mu \bar \gamma,\gamma),
   \label{alt:eigdensity}
 \end{multline}
%
and ${\mathcal I}_m$ is the spherical
integral below  known as the Harish-Chandra-Itzykson-Zuber (HCIZ) integral
\cite{tao_blog,hikami_brezin}: 
  \begin{align*} {\mathcal I}_m(  \bar \gamma,  \gamma) =
\int\limits_{{\mathcal O}(m)} \exp\biggl( \tr\bigl( O G  O^{\top} \bar G   \bigr) \biggr) H_m(dO) =    
\int\limits_{{\mathcal O}(m)} \exp\biggl(  \sum_{ij}  \{ O_{ij}\}^2 \bar  \gamma_i  \gamma_j \biggr) H_m(dO)  \; .
\end{align*}

Conditionally on the eigenvalues $(\gamma_1,\dots,\gamma_m)$, the conditional probability of 
$R= \bar O^{\top}O$ 
has
density  
\begin{align} \label{eig_vector:density} q_{\bar \gamma}( R | \gamma)= 2^m 
{\mathcal I}_m( 2\mu \bar\gamma, \gamma)^{-1}  \exp\biggl( 2 \mu \sum_{i,j=1}^m  \bar \gamma_i \gamma_j  R_{ij}^2 \biggr) 
\end{align}
with respect to the Haar probability measure $H_m(dR)$  on  $\bar O^{\top} {\mathcal O}(m)^+$.
\end{theorem}\begin{proof}
As in the zero mean case, we start from the isotropic Gaussian  matrix density \eqref{isotropic_prior:2} with mean $\bar D$,
 By using the
spectral representations $D= O G O^{\top}$ and $\bar D= \bar O \bar G \bar O^{\top}$,
after  the change of variables described in section \ref{section:spectral},
we find the joint density of $(G,O)$
with respect to the product measure 
\begin{align*}
  d\gamma_1 \times \dots \times d\gamma_m  \times H_m(dO)
\; \mbox{ on } 
   \{ \gamma\in \R^m: \gamma_1 > \gamma_2 > \dots > \gamma_m  \bigr\} \times {\mathcal O}(m)^+ \; ,
\end{align*}
 given as
 \begin{multline}  \label{joint:exact:density} 
     q_{\bar D}( G,O ) = 
    C_m(\mu,\lambda) \mbox{Vol}( {\mathcal O}(m)) V(G) 
 \times \\   \exp\biggl( -\mu \tr( G- O^{\top} \bar O \bar G \bar O^{\top} O)^2 -
    \frac \lambda  2 \bigl\{ \tr( G-O^{\top} \bar O \bar G \bar O^{\top}  O) \bigr\}^2 \biggr)   
   \\
 \nonumber 
  =  2^{m}Z_m(\mu,\lambda)  V(\gamma)  
  \exp\biggl( - \mu \tr\bigl( G^2 + \bar G^2 \bigr)
   - \frac{\lambda} 2  \bigl\{ \tr\bigl( G - \bar G\bigr) \bigr\}^2\biggr)
   \times \\  \exp\biggl( 2\mu  \tr\bigl( \bar O^{\top} O G  O^{\top} \bar O \bar G \bigr)   \biggr) \; .     
   \end{multline}
%
We change coordinates with $O\mapsto R=\bar O^{\top} O \in {\mathcal O}(m)^+$
and using the invariance property the  Haar measure
we see that
\begin{multline*} 2^m
   \int\limits_{{\mathcal O}(m)^+}  \exp\biggl( 2\mu  \tr\bigl( \bar O^{\top} O G  O^{\top} \bar O \bar G \bigr)   \biggr) H_m(dO)
  \\ =  \int\limits_{{\mathcal O}(m)}  \exp\biggl( 2\mu  \tr\bigl(  R^{\top} G  R \bar G \bigr)   \biggr) H_m(dR),
\end{multline*}
which proves \eqref{alt:eigdensity}.
 In the new coordinates
 \begin{multline}    
\label{HCIZ:firststep}
 q_{\bar \gamma}( G,R) =  
 2^m    Z_m(\mu ,\lambda)   V(G) \exp\biggl( - \mu \tr\bigl( (G -\bar G)^2 \bigr)
   - \frac{\lambda} 2  \bigl\{ \tr\bigl( G - \bar G\bigr) \bigr\}^2\biggr) 
\times \\   \exp\biggl( 2\mu \tr\bigl(  G  R^{\top} \bar G  R\bigr) -2\mu \gamma\cdot\bar\gamma   \biggr)\;  
    \end{multline}
 with respect to $d\gamma \times H_m(dR)$ on 
   $\{ \gamma\in \R^m: \gamma_1 > \gamma_2 > \dots > \gamma_m  \bigr\} \times   \bar O^{\top}{\mathcal O}(m)^+$  which proves  \eqref{eig_vector:density}.
\end{proof}
   
   \begin{remark} 
 When $\bar G=\bar \gamma Id$ we say that $\bar D$ is {\it spherical}. In such case $G$ is stochastically independent
from $O$, which follows  the Haar probability distribution.    
Equation \eqref{alt:eigdensity} shows the density of the ordered eigenvalues.
Often the random matrix literature deals with 
the density of the unordered eigenvalues on $\R^m$, which depends only on the order statistics and it differs by
a $1/m!$ factor.
The HCIZ integral admits  the series expansion  
\begin{eqnarray*}
   {\mathcal I}_m( \bar\gamma , \gamma) =  \sum_{k=0}^{\infty } \frac 1 {k!}  \sum_{ \alpha \in \Pi_m^k } 
   \frac{ C_{\alpha}( \bar \gamma ) C_{\alpha}( \gamma)  }{  C_{\alpha}( {\bf 1} ) },
\end{eqnarray*}
where the sum is over the set of partitions of $k$ into at most $m$ parts
\begin{eqnarray*}
  \Pi_m^{k} =\bigl \{  \alpha\in \N^m :   \alpha_1 \ge \alpha_2 \ge \dots \ge \alpha_m \ge 0  \mbox{ and  }
     \alpha_1 + \alpha_2 + \dots +  \alpha_m = k   \bigr \},
\end{eqnarray*}
and $C_{\alpha}(z_1,\dots,z_m)$ is  the  homogeneous zonal polynomial corresponding to the partition $\alpha$ \cite{james,muirhead,takemura,farrell}.
Theorem  \ref{HCIZ-asymptotics}  deals with the second order  asymptotics of ${\mathcal I}_m(n \gamma, \bar\gamma)
$ as $n\to\infty$.
When $m=3$
\begin{align*} 
{\mathcal I}_3(  \bar \gamma,  \gamma)
& =  \frac 1 {8\pi} \int_0^{2\pi} \int_0^{2\pi}\int_0^{\pi }
\exp\biggl(  \bar \gamma \Omega( \theta, \phi, \psi)  \gamma^{\top}\biggr) \sin( \theta)  d\theta d\phi d\psi , \quad \mbox{with }
 & \\ \Omega( \theta,\phi,\psi)&= 
\left[ { \footnotesize
\begin{matrix} 
(\cos\phi \cos \psi - \sin\phi\sin\psi \cos\theta)^2 &  (\cos\phi\sin\psi +\sin\phi\cos\psi \cos\theta)^2 & (\sin \phi\sin \theta)^2 \\
 ( \sin\phi \cos \psi - \cos\phi\sin\psi \cos\theta)^2 & (\sin\phi\sin\psi -\cos\phi\cos\psi \cos\theta)^2 &   (\cos \phi\sin \theta)^2 \\
               (\sin\psi \sin\theta)^2                                          &   (\cos\psi \cos\theta)^2   &       (\cos\theta)^2      
\end{matrix} } 
\right] &
  \end{align*}
expressed in  Euler angular coordinates. 
\end{remark}
\section{SMALL NOISE ASYMPTOTICS} \label{section:smallnoise}
\subsection{Spectral grouping}
 \begin{theorem} \label{spectral:clustering}
 Let $(D^{(n)},n \in \N)$ be
 a sequence of random $m\times m$ symmetric matrices 
such that, for some deterministic  limit
$\bar D$ and scaling sequence $a^{(n)}\to \infty$, 
\begin{eqnarray} \label{weak:convergence:assumption} 
  \sqrt{ a^{(n)}  } \bigl( D^{(n)} - 
  \bar D  \bigr)
   \stackrel{law}{\longrightarrow}  X,
\end{eqnarray} 
where $\mbox{vec}(X)$ is Gaussian with zero-mean and covariance $\Sigma(1,\lambda)$ 
for some $\lambda> - 2/m$ as in  \eqref{Omega_D:matrix:2nd}.

  Denoting by $(\gamma_j^{(n)} ,1\le j \le m)$  and $(\bar \gamma_j :1\le j \le m)$ the
  ordered eigenvalues of $D^{(n)}$ and $\bar D$, respectively,
  assume that  $\bar D$ has $ k$ distinct eigenvalues, i.e.
 \begin{eqnarray*}
  \bar \gamma_1 =  \dots =\bar \gamma_{\ell_1}> \bar \gamma_{\ell_1+1} = \dots 
 = \bar \gamma_{\ell_2} > \dots  > \bar \gamma_{\ell_{k-1}+1} = \dots = \bar\gamma_{\ell_k} 
  \end{eqnarray*}
   with $1\le k \le m$, $\ell_0=0,\ell_k=m$, 
   corresponding to eigenspaces of respective dimensions   $m_i=( \ell_i -\ell_{i-1})$.
   Consider the clusters 
  \begin{eqnarray*} 
    C_i^{(n)} = \{  \gamma_{\ell_{i-1}+1}^{(n)} >  \gamma_{\ell_{i-1}+2}^{(n)} > \dots >\gamma_{\ell_{i}}^{(n)} \bigr\}
  \quad 1\le i \le k \end{eqnarray*}
   formed by the ordered eigenvalues of $D^{(n)}$ corresponding to the  eigenspaces of $\bar D$ taken in the
    $\bar D$-eigenvalue order, and  define  the corresponding cluster  barycenters as  
 \begin{eqnarray*}  
 \widetilde\gamma_i^{(n)}= \frac{ \gamma_{\ell_{i-1}+1}^{(n)} + \dots + \gamma_{\ell_i}^{(n)} }{ m_i } 
 ,   \quad 1\le i \le k .
 \end{eqnarray*}  
We also consider
 the eigenvalue fluctuations
   \begin{eqnarray*}
   \xi_j^{(n)} = \sqrt{ a^{(n)} } ( \gamma_j^{(n)} - \bar \gamma_j  ) \quad  1\le j \le m  \;,
   \end{eqnarray*}
and the cluster barycenter fluctuations
  \begin{eqnarray*}
   \widetilde \xi_i^{(n)} =\sqrt{a^{(n)} } \bigl( \widetilde \gamma_i^{(n)} -  \bar\gamma_{\ell_i} \bigr)  
   =   \frac{1}{ m_i}  \sum_{j=\ell_{i-1}+1}^{\ell_i}  \xi^{(n)}_j
   ,\quad  1\le i \le k  \; .
   \end{eqnarray*}
As $n\to\infty$, the following limiting distribution appears:
   \begin{enumerate} \item \label{aspc:1} For the cluster barycenters, we have
   \begin{eqnarray*}
     \bigl( \widetilde \xi^{(n)}_1, \dots, \widetilde \xi^{(n)}_k \bigr)
     \stackrel{law}{\longrightarrow}  ( \widetilde  X_1,\dots, \widetilde X_k)\; ,
   \end{eqnarray*}
   where 
 \begin{eqnarray} \label{cluster:center:distribution}  
 \widetilde X_{i}= \frac{ X_{\ell_{i-1}+1 ,\ell_{i-1}+1} + \dots +  X_{\ell_i,\ell_i}}{ m_i }  ,   \quad 1\le i \le k
 \end{eqnarray}   
 have joint Gaussian density 
        \begin{align} \label{center:density:asymptotic} 
   q( \widetilde \xi_1,\dots \widetilde \xi_k)   = \sqrt{ 1+ \lambda m/2  } \; \prod_{i=1}^k  \sqrt  { \frac{ m_i }{\pi } }
  \exp\biggl(- \sum_{i=1}^k m_i {\widetilde \xi_i}^2   -\frac{\lambda}{2} \sum_{i=1}^k \sum_{j=1}^k m_i m_j \widetilde \xi_i \widetilde \xi_j
  \biggr)  \; ,    \end{align}  
 with zero-mean and covariance 
   \begin{eqnarray*}
      E\bigl(  \widetilde X_{i}\widetilde X_{j} \bigr) = \frac 1 2  \biggl( \frac{ \delta_{ij} }{ m_i} - \frac {\lambda} 
      { 2+ \lambda m} \biggr) , \; 1 \le i ,j\le k\; .
   \end{eqnarray*}
 \item  \label{aspc:2}
 For each cluster, the differences between the eigenvalues and  their barycenter 
  \begin{eqnarray*} \xi^{(n)}_{j} - \widetilde \xi^{(n)}_i =
 \sqrt{  a^{(n)} } \bigl( \gamma_j^{(n)}- \widetilde \gamma_i^{(n)} \bigr)  \;: \quad   i=1,\dots,k, \quad j= \ell_{i-1}+1, \dots, \ell_{i}    
  \end{eqnarray*}
 are asymptotically independent from their cluster barycenter and the other clusters, 
 with limiting distribution
  \begin{align} \label{eq:eigen:spreads}
   \bigl(  \xi^{(n)}_{\ell_{i-1} +1} - \widetilde \xi^{(n)}_i, \dots ,  \xi^{(n)}_{\ell_{i}} - \widetilde \xi^{(n)}_i\bigr) 
  \stackrel{law}{\longrightarrow}\bigl( \gamma_1 -\widetilde \gamma_{m_i} ,\dots,\gamma_{m_i}-\widetilde \gamma_{m_i} \bigr) \quad  1\le i \le k
  , \end{align}
where $(\gamma_1>\gamma_2 > \dots> \gamma_{m_i})$ are eigenvalues of the standard $m_i$-dimensional GOE
 of symmetric Gaussian matrices with zero mean and precision $A_{m_i}(1,0)$ 
 with barycenter
 \begin{align*}
  \widetilde \gamma_{m_i} = \frac 1 {m_i} \sum_{j=1}^{m_i} \gamma_j  \sim {\mathcal N}\bigl( 0, 1/(2d_i) \bigr)  \;.
\end{align*}
Moreover the differences
$\bigl( \gamma_1 -\widetilde \gamma_{m_i} ,\dots,\gamma_{m_i}-\widetilde \gamma_{m_i}  \bigr)$ are   independent
from  $\widetilde \gamma_{m_i}$,
with degenerate density 
 \begin{multline}\label{cluster:density:asymptotic} 
    q_{m_i}( 
     \zeta_{\ell_{i-1}+1},\dots ,\zeta_{\ell} )  = \\ Z_{m_i}( 1,0)
\sqrt{  \pi   m_i } 
  \exp\biggl(-  \sum_{j=\ell_{i-1}+1}^{\ell_i} \zeta_j^2 \biggr) 
  \delta_0( \zeta_{\ell_{i-1}+1 } + \dots + \zeta_{\ell_i}\bigr) \prod_{ \ell_{i-1}+1 \le j < h \le \ell_i}
   \big\vert \zeta_j-\zeta_h\big\vert  \;   
 \; ,
    \end{multline} where $\delta_0(z)$ denotes the Dirac distribution, which is also
    the conditional density of the GOE eigenvalues $(\gamma_1,\dots, \gamma_{m_i})$ conditioned
   on $\{\gamma_1+ \dots+ \gamma_{m_i}=0\}$. 
\item  \label{aspc:3} In particular  for each cluster,
\begin{align*}
\bigl(
\sqrt{a^{(n)}} (\gamma_j^{(n)}-\gamma_h^{(n)}) : \ell_{i-1}+1\le j < h \le \ell_{i-1} \bigr)
\stackrel{law}{\to}  \bigl( (\gamma_j-\gamma_h ) : 1\le j < h \le m_i \bigr)
\end{align*}
and these eigenvalue differences are asymptotically independent from the cluster barycenter and  the  other clusters.
\end{enumerate}
  \end{theorem}
\begin{remark}:
The weak convergence hypothesis \eqref{weak:convergence:assumption} 
implies 
\begin{eqnarray*}
D^{(n)}=    O^{(n)} \mbox{diag}\bigl(\gamma^{(n)}\bigr)  { O^{(n)}}^{\top}    
\stackrel{P}{\longrightarrow} \bar D=\bar O \;\mbox{diag}\bigl(  \bar \gamma \bigr)   { \bar O}^{\top} \; ,
\end{eqnarray*}
 which means  that
$\gamma^{(n)}\stackrel{P}{\longrightarrow}\bar\gamma$ and $\bar O^{\top} O^{(n)}
\stackrel{P}{\longrightarrow}\mbox{I}$ in probability.
The asymptotic distribution in \eqref{eq:eigen:spreads}
depends only on $m_i$ (the size of the cluster)
and not on the interaction parameter $\lambda$.
When $\sqrt{ a^{(n)}}( D^{(n)}-\bar D)$ has an isotropic Gaussian distribution with covariance
$\Sigma(1,\lambda)$, and the mean  $\bar D = \bar \gamma I$ is spherically symmetric, 
there is only one cluster and 
the distributional equalities in  Theorem \ref{spectral:clustering}  
   hold exactly without going to the limit in distribution. 
A related result is given in \cite{ibrahim}  for the joint asymptotic distribution
  of eigenvalues and eigenvectors.  
  Similar results have been  derived in the special case of  
  non-central Wishart random matrices, and sample covariance matrices which are asymptotically Gaussian \cite{anderson},\cite[Theorem. 9.5.5]{muirhead}.
  \end{remark}
 Next, we illustrate the implications of 
 Theorem \ref{spectral:clustering} in the 3-dimensional situation which is relevant for DTI:  

 \begin{corollary}\label{asymptotic:eigdensity} Let $D$ be $3\times 3$ symmetric matrix with Gaussian density 
 \eqref{isotropic_prior:2}. As $\mu \to \infty$ with $\lambda> - 2\mu /3$,
we have four asymptotic regimes depending on the symmetries of the mean  matrix $\bar D$.
\begin{enumerate}
 \item  $\bar\gamma_1 > \bar \gamma_2 > \bar \gamma_3$ (totally asymmetric tensor)
 
The joint density of  $(\gamma_1,\gamma_2,\gamma_3)$  is approximated by the Gaussian density of $( D_{11},D_{22},D_{33})$,
i.e.
\begin{align} \label{asymptotic:gaussian:anisotropic}
   q( \gamma_1,\gamma_2,\gamma_3)  \simeq\frac{\mu \sqrt{ 2\mu + 3\lambda} }{  \pi^{3/2} \sqrt 2 }
   \exp\biggl(  - \mu  
   \sum_{i} (\gamma_i-\bar \gamma_i)^2 - \frac {  \lambda } 2\sum_{ij} (\gamma_i -\bar \gamma_i)(\gamma_j-\bar\gamma_j)\biggr).
 \end{align}
\item  $\bar\gamma_1 > \bar \gamma_2 = \bar \gamma_3$ (prolate tensor). Let $ \widetilde \gamma_{23}=(\gamma_2 + \gamma_3)/2$.
 The joint distribution of $(\gamma_1, \widetilde \gamma_{23})$ is approximated by  the Gaussian distribution 
of $\bigl( D_{11}, (D_{22}+D_{33})/2  \bigr)$,
i.e.   
\begin{align} & \label{AD:RD:gaussian}   
   q( \gamma_1, \widetilde \gamma_{23} ) \simeq \pi^{-1}  \sqrt{ 2\mu^2 +  \lambda^2 3/4 + 3\mu \lambda } \;
\times & \\ & \nonumber \exp\biggl(  - \biggl( \mu + \frac{\lambda} 2  \biggr) (\gamma_1 -\bar \gamma_1)^2- 2( \mu + \lambda)
    ( \widetilde \gamma_{23}-\bar \gamma_2)^2 -2 \lambda (\gamma_1 -\bar \gamma_1)  ( \widetilde \gamma_{23} -\bar \gamma_{2} )\biggr)
  & \end{align}
Conditionally on $(\gamma_1, \widetilde \gamma_{23} )$, the asymptotic distribution of $(\gamma_2,\gamma_3)$
is degenerate, with $\gamma_3=(2 \widetilde  \gamma_{23}-\gamma_2)$ and  
\begin{align} \label{asymptotic:conditional}
 q( \gamma_2| \widetilde \gamma_{23} ) \simeq  (\gamma_2- \widetilde \gamma_{23} ) 
 \exp\bigl(  - 2\mu (\gamma_2- \widetilde \gamma_{23} )^2 \bigr) 2\mu
 {\bf 1}(\gamma_2 >  \widetilde \gamma_{23} ),
\end{align}
that is $( \gamma_2- \widetilde \gamma_{23} )=( \widetilde \gamma_{23}-\gamma_3)\simeq \sqrt{ \tau }$, with $\tau$
exponentially distributed with rate $2\mu$ and  independent from the barycenter $ \widetilde \gamma_{23}$.
\item  $\bar\gamma_1 = \bar \gamma_2 > \bar \gamma_3$ (oblate tensor). This is similar to the prolate case.
Let $ \widetilde \gamma_{12}=(\gamma_1 + \gamma_2)/2$. Asymptotically
the joint distribution of $( \widetilde\gamma_{12}, \gamma_3 )$ is approximated by the Gaussian distribution 

of $\bigl(  (D_{11}+D_{22})/2, D_{33} \bigr)$, with
  \begin{multline}
  q( \widetilde  \gamma_{12},\gamma_3 )  \simeq \pi^{-1}  \sqrt {   2\mu^2 + \lambda^2 3/4 + 3\mu \lambda }
 \times \\  \exp\biggl(  - \biggl( \mu + \frac{\lambda} 2  \biggr) (\gamma_3 -\bar \gamma_3)^2- 2( \mu + \lambda)
    ( \widetilde \gamma_{12}-\bar \gamma_1)^2 -2 \lambda (\gamma_3 -\bar \gamma_3)  ( \widetilde \gamma_{12} -\bar \gamma_{1} )\biggr)
\end{multline}
   and the asymptotic conditional distribution of   $(\gamma_1,\gamma_2)$ given $( \widetilde \gamma_{12},\gamma_3)$
is degenerate with $\gamma_2=(2 \widetilde \gamma_{12}-\gamma_1)$, and
\begin{align}\label{asymptotic:conditional2}
 q( \gamma_1| \widetilde  \gamma_{12} ) \simeq
 ( \gamma_1 -  \widetilde \gamma_{12} ) 
 \exp\bigl(  - 2\mu (\gamma_1- \widetilde \gamma_{12} )^2 \bigr) 2\mu  
 {\bf 1}( \gamma_1 > \widetilde  \gamma_{12} ),
\end{align}
i.e.  $( \gamma_1- \widetilde \gamma_{12} )=( \widetilde \gamma_{12}-\gamma_2)\simeq \sqrt{ \tau}$, with
$\tau$ exponentially distributed  with rate $2\mu$, 
independent from $ \widetilde \gamma_{12}$.

\item   $\bar \gamma_1= \bar \gamma_2 = \bar \gamma_3$ (isotropic tensor)
  
  The barycenter
  $ \widetilde \gamma_{123}= (\gamma_1+\gamma_2 +\gamma_3 )/3= \frac 1 3 \tr(D)$
 is Gaussian with mean $\bar\gamma_1$
  and variance $1/(6\mu + 9\lambda)$.
  
 Conditionally on $\widetilde \gamma_{123}$, $(\gamma_1,\gamma_2,\gamma_3)$  is degenerate,
  with $\gamma_2 = (3\widetilde \gamma_{123}- \gamma_1- \gamma_3)$, and the conditional density
   of $(\gamma_1,\gamma_3)$ given $\widetilde\gamma_{123}$  is approximated as
\begin{multline} \label{conditional:cluster} 
  q( \gamma_1,\gamma_3| \widetilde \gamma_{123} ) \simeq   (2\mu)^{5/2}\sqrt{ \frac 3 {\pi} }         
  ( \gamma_1- \gamma_3)( 2\gamma_1 + \gamma_3-3 \widetilde\gamma_{123}) (3 \widetilde \gamma_{123}-\gamma_1 -2\gamma_3) 
 \times \\ 
 \exp\bigl( - 2 \mu \bigl\{   (\gamma_1- \widetilde \gamma_{123} )^2+ (\gamma_3-\widetilde\gamma_{123})^2
 +(\gamma_1-\widetilde\gamma_{123})(\gamma_3-\widetilde \gamma_{123}) 
 \bigr\}\bigr)  { \bf  1}( \gamma_1 > \gamma_{123} > \gamma_3) 
\end{multline}
 Asymptotically, the conditional distribution of  the vector 
 \begin{align*}\sqrt{2\mu} (\gamma_1-\widetilde\gamma_{123},
 \gamma_2-\widetilde\gamma_{123}, 2 \widetilde\gamma_{123} -\gamma_1-\gamma_2 )
 \end{align*}
  coincides with the conditional distribution of the ordered eigenvalues of the 3-dimensional standard GOE, conditioned on 
having zero barycenter, and 
   are independent from $\widetilde\gamma_{123}$.
\end{enumerate}
\end{corollary}
\begin{remark} For a totally anisotropic  mean tensor $\bar D$, 
the asymptotic Gaussian density \eqref{asymptotic:gaussian:anisotropic} for the
rescaled eigenvalue fluctuations around their barycenter
  coincides with the  Gaussian eigenvalue density (18)  
  of \cite{basser-pajevic03}. However  in  \cite{basser-pajevic03} it was 
postulated erroneously that the  map $D=\bigl(O G O^{\top}\bigr)\mapsto G$ 
was linear with constant Jacobian and \eqref{asymptotic:gaussian:anisotropic} 
would be the  eigenvalue density of a random tensor with isotropic Gaussian noise,
which is not correct, in the  non-asymptotic case the eigenvalue density is given by \eqref{alt:eigdensity}. 
\end{remark}
\subsection{Axial and Radial diffusivity marginals}
Two eigenvalue statistics that are particularly relevant in DTI are: Axial Diffusivity (AD), which corresponds to the largest $D$-eigenvalue  $\gamma_1$
and it is measured along the principal axis of the
diffusion tensor and is considered a putative axonal damage marker,
and radial
diffusivity (RD), which correponds to $\widetilde \gamma_{23}=(\gamma_2+ \gamma_3)/2$ and   
is measured perpendicular to the principal axis and thought to be
sensitive to the degree of hindrance that diffusing water molecules experience due
to the axonal membrane and myelin sheath. In this sub-section we derive the 
distributions for AD and RD
in dimension $m=3$ when $D$ has the density given in \eqref{isotropic_prior:2}.
When the mean matrix $\bar D$ is prolate, we have shown in Corollary \ref{asymptotic:eigdensity} that in the small noise limit
the joint distribution of AD and RD is asymptotically Gaussian, given in Eq. \eqref{AD:RD:gaussian}.

In the case of $D$ with spherical mean $\bar D= \bar \gamma \mbox{Id}$, 
we can also derive the marginal densities of AD and RD. See also  \cite{chiani}, which contains a recursive
expressions for the distribution of the largest GOE eigenvalue  in arbitrary dimension. 
After changing variables in the joint conditional  eigenvalue density \eqref{conditional:cluster}, we see that $z_i=(\gamma_i-{\widetilde \gamma_{123}})$
are independent from the barycenter ${\widetilde \gamma_{123}}$,
$z_1= (\gamma_1- {\widetilde \gamma_{123}})$ and $(-z_3)=({\widetilde \gamma_{123}}- \gamma_3)$ are identically distributed, with
marginal density
\begin{align*} & q(z_1)=  
 (2\mu)^{5/2}\sqrt{ \frac 3 {\pi} }
 \exp( -   3\mu z_1^2 /2 ){\bf 1} ( z_1> 0)  \times & \\ &  \times 
 \int_{-2 z_1}^{-z_1/2}
 (z_3-z_1)( 2 z_1 + z_3) ( 2 z_3 +z_1)  \exp\bigl( -  (z_3 +z_1/2)^2 2\mu \bigr) dz_3
& \\
&  =
\mu^{3/2} \sqrt{ \frac 6 {\pi} }  
 \biggl(  \frac{9 z_1^2}{ 2 } + \frac{ \exp( -  9\mu z_1^2  /2)-1 }{ \mu} \biggr) \exp( -   3\mu  z_1^2 /2 ){\bf 1} ( z_1> 0)
 \end{align*}
and cumulative distribution function
\begin{align*} &  P\bigl( \gamma_1 - {\widetilde \gamma_{123}} \le t \bigr)= 1- P\bigl(  \gamma_3 -{\widetilde \gamma_{123}} \le - t \bigr) = 
1-P\bigl( \widetilde \gamma_{23}- \widetilde \gamma_{123} \le -t/2 \bigr) 
& \\ & =
  \mu^{3/2}\sqrt{ \frac 6 {\pi} }
 \int_0^t 
 \biggl( \frac{ 9 z^2 }{ 2 } + \frac{  \exp( -9 \mu  z^2 /2)-1 }{ \mu} \biggr)\exp( -   3 \mu z^2  /2 ) dz
&\\ &
=\bigl\{ \Phi(  t \sqrt{3\mu } )+ \Phi( t \sqrt{12\mu} ) -1 \bigr\} -3t \sqrt{ \frac{ 3\mu}{ 2\pi} } 
\exp( - 3\mu t^2  /2 ) & 
 \end{align*} 
 where 
\begin{align*}\phi(t) = \frac 1 {\sqrt { 2\pi} } \exp( -t^2/2 ), \quad \Phi(t)=\int_{-\infty}^t \phi(s)ds \end{align*} denote 
the standard Gaussian density and cumulative distribution function, respectively.
The cumulative distribution function of $\gamma_1$  
is obtained by taking  convolution with the barycenter $\widetilde \gamma_{123}$ distribution ${\mathcal N}\bigl(\bar \gamma, 1/(6\mu + 9 \lambda) \bigr)$,
obtaining 
\begin{align*} & P(  \gamma_1- \bar\gamma \le t)= 1-P( \gamma_3-  \bar \gamma \le -t )=& \\ &
 \int_{-\infty }^t P( \gamma_1  -{\widetilde \gamma_{123}} \le t-x) 
 \exp\biggl( - (6\mu + 9\lambda) x^2/2 \biggr) \sqrt{ \frac{ 6\mu + 9 \lambda} {2\pi} }dx & \\ & 
 =\int_{-\infty}^t \bigl[
\Phi(  (t-x)\sqrt{3\mu } )  + \Phi( (t-x) 2\sqrt{3\mu} ) \bigr]
  \exp\biggl( - (6\mu + 9\lambda) x^2/2 \biggr) \sqrt{ \frac{ 6\mu + 9 \lambda} {2\pi} }dx
 &\\&
  +
  \frac{ \mu }{ ( \mu + \lambda ) } \biggl\{
 \frac{ (2\mu+ 3\lambda) t}{  \sqrt{ \mu+ \lambda } }
 \phi\biggl(  \frac{ t  \sqrt{ 2\mu^2+ 3\lambda \mu}  }{ \sqrt{\mu + \lambda } } \biggr)
 \biggl[ 
 \Phi\biggl( \frac{ t (2\mu + 3\lambda) }{ \sqrt{\mu+\lambda} }\biggr) -1 \biggr] 
 +\frac{ \phi( t \sqrt{ 6\mu + 9 \lambda } )  }{\sqrt {2\pi} } \biggr\} & \\
& -  \Phi\bigl( t \sqrt{6\mu + 9 \lambda}\bigr) \; .
  &
\end{align*}
The joint density of $AD$ and $RD$ is given by
\begin{align*}
& q( \gamma_1,  \widetilde \gamma_{23} )=
  \frac{ 4\mu^{3/2} \sqrt{2\mu + 3 \lambda} }{\pi } 
  \biggl(  (\gamma_1- \widetilde \gamma_{23} )^2 + \frac{ \exp( - 2\mu(\gamma_1-
  \widetilde \gamma_{23} )^2  )-1 }{ 2\mu} \biggr)  \times
& \\ & \exp\biggl( - \bigl( \mu + \frac{\lambda} 2 \bigr)\bigl(\gamma_1- \bar \gamma\bigr )^2   - (2\mu +  2\lambda  )  \bigl(\widetilde \gamma_{23}- \bar\gamma\bigr)^2  
 - 2\lambda  \bigl(\gamma_1-\bar\gamma \bigr) \bigl(\widetilde \gamma_{23} -\bar\gamma \bigr)    \biggr)  {\bf 1} ( \gamma_1> \widetilde \gamma_{23} )\; .
&
\end{align*} 
\subsection{Eigenvector asymptotics}
In the settings of Theorem  \ref{spectral:clustering}, where   $D^{(n)}$ and $\bar D$
have respective spectral decompositions  $O^{(n)} G^{(n)} {O^{(n)}}^{\top}$ and $\bar O\bar G \bar O^{\top}$, we
study the asymptotics of 
  $R^{(n)} = \bar O^{\top} O^{(n)} \in {\mathcal O}(m)$.
Omitting the $n$ superscript,   we use the decomposition  
$R=\check R \hat R$,
where
   \begin{eqnarray} \label{block:diagonal structure} &&
  \check R=\begin{pmatrix} \check R_{(1,1)}  &      & &{\text{\large 0} }  \\  &\check  R_{(2,2)} &      & & \\ 
                      &        &  \ddots & \\    {\text{\large 0} } & & & \check  R_{(k,k)}   \end{pmatrix}
                     \; ,
                      \end{eqnarray}
\noindent is  block diagonal with blocks
 $\check R_{(j,j)} 
 \in{\mathcal O}(m_j)$ corresponding to the $m_j$-dimensional eigenspaces of $\bar D$.
 
 These matrices form a  subgroup 
 ${\mathcal K}_{\bar\gamma}\simeq {\mathcal O}(m_1) \times{\mathcal O}(m_2) \times \dots  \times {\mathcal O}(m_k)$,
such that $\check R \bar D \check R^{\top}=\bar D, \; \forall \check R\in {\mathcal K}_{\bar \gamma} $,
and the conditional eigenvector density \eqref{eig_vector:density} is invariant under the action of 
${\mathcal K}_{\bar\gamma}$.

$\hat R \in {\mathcal S} { \mathcal O }(m)$  is a rotation
with Lie matrix exponential representation 
 \begin{align} \label{skew:symmetric:class}\hat R=\exp( \hat S)= \sum_{k=0}^{\infty} \frac{ \hat S^k}{k!} , 
          \mbox{ where }     \hat S=   {\small  \begin{pmatrix}
   0  & \hat  S_{(1,2) }  &  \dots  & \hat S_{ (1,k-1) }  &  \hat S_{(1,k) } \\
 -\hat  S_{(1,2)}^{\top}  &  0  &  \dots     &  \hat  S_{(2,k-1)}  & \hat S_{(2,k)} \\
     &  &  \ddots   &    &  \\
 -\hat  S_{(1,k-1)}^{\top}  &  -\hat S_{ (2,k-1)}^{\top} &  \dots &      0 & \hat S_{(k-1,k)}  \\
  -\hat S_{(1,k)}^{\top} & - \hat S_{(2,k)}^{\top}  &  \dots  & - \hat S_{(k-1, k)}^{\top} &  0
  \end{pmatrix}                         }
 \end{align} 
 \noindent  is 
skew-symmetric, with blocks $\hat S_{(j,l)}=-\hat S_{(l,j)}^{\top} \in \R^{m_j\times m_l}$ for $1\le j<l\le k$, and zero $(m_j\times m_j)$-blocks
 on the diagonal, with $\bigl( m^2 - \sum_{i=1}^k m_i^2\bigr)/2$ free parameters. The subgroup
 \begin{align*}
 {\mathcal C}_{\bar\gamma}=\bigl\{  \exp( \hat S) : \hat  S\mbox{ has the skew-symmetric structure \eqref{skew:symmetric:class} } \bigr\} \;
\end{align*}
 is a complement subgroup of ${\mathcal K}_{\bar\gamma}$ in ${\mathcal O}(m)$.

In dimension $m=3$, $\hat R= \exp( \hat S)$ is a clockwise rotation 
by an angle $\theta= \sqrt{ \hat S_{23}^2+ \hat S_{13}^2 + \hat S_{12}^2 }$
around the unit vector  $ u = ( \hat S_{23}, -\hat S_{13} ,\hat S_{12} )/\theta$. The matrix exponential
$\exp(d \hat S)$ of an infinitesimal $3\times3$ skew symmetric matrix is the composition of  three
 infinitesimal rotations around the Cartesian axes $x,y,z$,
 by the Euler angles $d\hat S_{23}$ (roll), $d\hat S_{13 }$ (pitch), 
 and $d\hat S_{12}$ (yaw), respectively,
 which commute up to infinitesimals of  higher order.

 \begin{theorem}\label{eigenvector:asymptotics}
   In the settings of Theorem  \ref{spectral:clustering}, let
   \begin{align*}
 \bar O^{\top} O^{(n)}=  R^{(n)} = \check R^{(n)}  \hat R^{(n)} =\check R^{(n)}   \exp\bigl(  \hat S^{(n)} )
   \end{align*}
    with $\check R^{(n)}\in {\mathcal K}_{\bar \gamma}$ and  $\hat S^{(n)}$ skew-symmetric. 
  The blocks $\check R^{(n)}_{(i,i)}, i=1,\dots,k$ corresponding to the $\bar D$ eigenspaces
  are asymptotically 
   distributed according to the product of the Haar measures on the respective orthogonal groups ${\mathcal O}(m_i)$,
   with the constraint $\bar O \check R^{(n)} \in {\mathcal O}(m)^+$,
   and asymptotically independent from the eigenvalue fluctuations. 
 
 After rescaling, the entries  $\bigl( \sqrt{a^{(n)}}\hat S^{(n)}_{ij}: \bar \gamma_i > \bar \gamma_j)$ are
   asymptotically mutually independent and independent from $\check R^{(n)}$ and the eigenvalue fluctuations, with
   limiting Gaussian distribution
   \begin{align*}
                 {\mathcal N}\biggl( 0,  \frac 1 { 4 (\bar \gamma_i - \bar \gamma_j )^2 } \biggr) .
   \end{align*}  
 \end{theorem}
 {\bf Remark}: Theorem  \ref{eigenvector:asymptotics}
  extends  Theorem 4.1 in \cite{schwartzman} for $\bar D$ with non-negative distinct eigenvalues,
 given also in \cite{pajevicbasser2010},
  to the case with repeated eigenvalues.  
\subsection{Second order approximation of the  HCIZ-integral}
 \begin{theorem}\label{HCIZ-asymptotics} Let $\gamma,\bar\gamma\in \R^m$ ordered vectors, such that
   the coordinates $(\gamma_1 >\gamma_2 > \dots >\gamma_m)$ are distinct, while the $\bar\gamma$ coordinates may
   coincide, with multiplicities $m_i=( \ell_i -\ell_{i-1})$ and
 \begin{align*} 
  \bar \gamma_1 = \dots =\bar \gamma_{\ell_1}> \bar \gamma_{\ell_1+1} = \dots 
 = \bar \gamma_{\ell_2} > \dots  > \bar \gamma_{\ell_{k-1}+1} = \dots = \bar\gamma_{\ell_k}\; ,
  \end{align*}
 for $0=\ell_0<\ell_1< \dots<\ell_k=m,  \quad 1\le k \le m$. Then, as $n\to\infty$, 
 \begin{align} \label{HCIZ:asymptotic:eq} &
 \lim_{n\to\infty} 
 {\mathcal I}_m( n \gamma, \bar \gamma) 
   \exp( -n \gamma\cdot  \bar\gamma) n^{\bigl( m^2 -\sum\limits_{i=1}^k m_i^2 \bigr)/4}= & \\ \nonumber &
   \frac{ \prod\limits_{l=1}^{m} \Gamma( l/2) }
   { \prod\limits_{i=1}^{k}\prod\limits_{l=1}^{m_i}  \Gamma( l/2) }
  \prod\limits_{i=1}^{k-1} \prod\limits_{j=\ell_{i-1}+1}^{\ell_i} \prod_{h=\ell_i+1}^m 
  \bigl[  ( \gamma_j - \gamma_h)  
  (\bar \gamma_j - \bar \gamma_h)  \bigr]^{-1/2} \; .
 &
 \end{align}
 \end{theorem} 
 \begin{remark}
 Theorem \ref{HCIZ-asymptotics} was proven by \cite{anderson}(see also \cite[Thm. 9.5.2.]{muirhead})
 in the case of non-negative eigenvalues without multiplicities.
\end{remark}
 \section{TESTING THE SPHERICITY HYPOTHESIS} \label{section:sphericity} 
In DTI, it is often desirable to establish different symmetries of the underlying tensor field. One of the often used
tests is the test of isotropy of the underlying mean diffusion tensor \cite{basser1994b}. Here we also develop one such test and we
call it a test of sphericity, to avoid confusion with the ``isotropy'' of the
precision tensor. 
Consider a sequence of 
random symmetric matrices $D^{(n)}$ such that
  $ \sqrt{ a^{(n)}  }\bigl( D^{(n)}- \bar D\bigr)\stackrel{law}{\to} X$,
  where  the limit is a zero mean Gaussian symmetric matrix, $\bar D$ is deterministic and
 $a^{(n)}\to\infty$ is a scaling sequence. For example, in Section \ref{section:rician} the scaling sequence is given by the number of
 gradients in the DTI measurement.
 In order to test the  sphericity hypothesis
  \begin{align*}
H_0:\quad  \bar D= \bar\gamma\; \Id \mbox{ for some unknown $\bar \gamma \in \R$, }
\end{align*}
we introduce the sampled eigenvalue central moments  
\begin{subequations}
\begin{align*}
  \kappa_1(D)&= \frac 1 m \sum_{i=1}^m \gamma_i = \frac 1 m \tr(D) , 
&\\ \kappa_r(D) & =\frac 1 m \sum_{i=1}^m (\gamma_i -\kappa_1(D) )^r   
  = \frac 1 m \tr\biggl( \biggr( D- \frac{\tr(D)} m \mbox{Id} \biggr)^r \biggr) & \\ &= 
   \sum_{q=0}^r \binom{r}{q} \frac{(-1)^q}{m^{q+1}}\tr ( D^{r-q}) \tr(D)^q\; ,
   2\le r \in \N. &
\end{align*} 
\end{subequations}
where $ \gamma_i$ are the eigenvalues of $D$. 
 \begin{lemma}  \label{homogenous:polynomial}  $\kappa_r(D)$ is a homogenous polynomial of degree $r$ in the matrix entries, satisfying 
 $\forall  c \in \R$
\begin{align}\label{invariance:property}\kappa_1(D+c\;\Id)=\kappa_1(D)+c, \quad \kappa_r(D+c\;\Id)= \kappa_r(D) \quad ,\; r \ge 2 \; . 
\end{align}   This implies that 
  the derivatives  satisfy
  $\nabla^{\ell} \kappa_r(Id )=0$ $\forall  0 \le \ell < r$, 
 
 \noindent while $\nabla^{r} \kappa_r(D)= \nabla^{r} \kappa_r(0 )$
  are constant tensors such that
  \begin{align*}   &
        \kappa_r(D)  = \frac 1 {r!} \nabla^{r} \kappa_r(0) \underbrace{ D \otimes \dots \otimes D}_{ \mbox{ $r$-times } }  
  \;   , \quad   \tr\bigl(  \nabla^{\ell} \kappa_r(D) \bigr)= 0, \quad\forall r \ge 2,\;  1\le \ell \le r \; . 
       & \end{align*}
  \end{lemma}
  \begin{corollary} Let $D^{(n)}$ be a sequence of $m\times m$ symmetric random matrices and
  $X$  a zero mean symmetric Gaussian matrix 
  such that, for some $\bar\gamma \in \R$ and scaling
  sequence $a^{(n)}\to\infty$,
  \begin{align*}
   \sqrt{a^{(n)}} \bigl( D^{(n)}- \bar\gamma \; \Id \bigr) \stackrel{law}{\longrightarrow} X \; .  \end{align*}
Then
  \begin{align*}
   \bigl( \sqrt{ a^{(n)} }\bigl (D^{(n)}-\bar\gamma \Id\bigr), \bigl(a^{(n)} \bigr)^{r/2} \kappa_r\bigl(     D^{(n)} \bigr) \; : \;
   2\le r \le m \bigr) \stackrel{law}{\longrightarrow }  \bigl( \kappa_r(X)   \; :\; 1\le r \le m \bigr). 
  \end{align*} 
  When the covariance of  $X$  is   isotropic, $(\kappa_r(X): 2\le r \le m )$ are stochastically independent from $\kappa_1(X)$.
  \end{corollary}  
  \begin{proof} For the first statement we apply the continuous mapping theorem together with \eqref{invariance:property}.
  If $X$ has zero mean isotropic Gaussian distribution, 
  the conditional distribution of  $(X-\kappa_1(X) \Id)$  given $\kappa_1(X)$  is also 
  zero-mean isotropic Gaussian and does not depend on the value of $\kappa_1(X)$.
 \end{proof}
   To test the sphericity hypothesis with $\bar\gamma \ne 0$ it is natural to
   use  statistics of the form
   \begin{align*}
   \tau^{(n)}=\tau\bigl(\kappa_1\bigl(D^{(n)}\bigr), \bigl( a^{(n)}\bigr)^{r/2} \kappa_r( D^{(n)}): 2\le r \le m\bigr)\; ,
   \end{align*}
   and calibrate the test against the 
    distribution of 
    \begin{align}\label{asymptotic:calibration}\tau^{(\infty)}=\tau\bigl(c  ,\kappa_r(X): 2\le r \le m\bigr)\; ,\end{align}
     evaluated at  $c=\kappa_1\bigl(D^{(n)}\bigr)$.  However,
without additional assumptions on the  covariance structure of $X$
the probability density functions of $\kappa_r(X)$ for $r\ge 2$ do not have closed form expressions
and can be only computed  numerically, for example by  Monte Carlo simulations.
Note also that,  
since  $\nabla^{\ell} \kappa_r( \Id )=0$  $\forall r\ge 2, 0 \le \ell < r$,
we are dealing 
with a singular hypothesis testing problem  \cite{drton2009,drton,watanabe},
where the constraints $\{ \kappa_r( \bar D )=0, r\ge 2 \}$ which we are 
testing for  are singular at the true parameter $\bar D=\bar\gamma\;\Id$, consequently 
any smooth sphericity  statistics $\tau^{(n)}$ 
will follow non-Gaussian higher order asymptotics.
 We proceed now  in dimension $m=3$, 
assuming that the Gaussian matrix limit  $X$ has  
zero mean and isotropic precision matrix $A(1,\lambda)$ with $\lambda> -2/3$,
to compute explicitly the asymptotic density of  some commonly used sphericity statistics based
on eigenvalues sample mean, variance and skewness. 
\begin{lemma}\label{sphericity:test:thm}In the settings of Theorem \ref{spectral:clustering}, under the sphericity hypothesis $H_0$,
the test statistics 
\begin{multline} \label{test:statistics:T}
\tau_1^{(n)}= \sqrt{ a^{(n)} }\bigl(\kappa_1( \gamma^{(n)}) -\kappa_1(\bar \gamma) \bigr ),
\tau_2^{(n)}=6 a^{(n)}\kappa_2( \gamma^{(n)} )    , 
\tau_3^{(n)} =  \sqrt 2 \kappa_3( \gamma^{(n)} )  \kappa_2( \gamma^{(n)} )^{-3/2} , 
\end{multline}
are asymptotically independent, with limiting distributions
\begin{align} \label{limiting:test:distribution}
 \tau_1^{(n)} \stackrel{law}{\longrightarrow} {\mathcal N}\bigl( 0, 1/(6+9\lambda) \bigr),
 \quad \tau_2^{(n)}  \stackrel{law}{\longrightarrow} \chi^2_{5}, \quad \tau^{(n)}_3  
 \stackrel{law}{\longrightarrow} \mbox{Uniform}([-1,1])
 \; .
\end{align}
In dimension   $m$
\begin{eqnarray*} 
\biggl(( 2 m + \lambda m^2 )
  \bigl\{\kappa_1( \gamma^{(n)} ) -  \kappa_1(\bar \gamma)  
\bigr\}^2   \; , \;    
2m \kappa_2(\gamma^{(n)})   \biggr) a^{(n)} 
\stackrel{law}{\longrightarrow} \bigl(  \chi^2_1 \;, \; \chi^2_{  (m+1)m/2-1  }  \bigr) 
 \end{eqnarray*} 
with asymptotically independent components.
\end{lemma}
\begin{proof} 
We start from the asymptotic eigenvalue density \eqref{alt:eigdensity}, which under  $H_0$ is given by
 \begin{multline*}  
 q_{\bar \gamma}( \gamma_1,\gamma_2, \gamma_3)=  \\
 \frac{ 4 \mu^{5/2} \sqrt{ 2\mu +3\lambda }}{  \pi} \;  V( \gamma_1,\gamma_2,\gamma_3 ) 
  \exp\biggl( - \frac{ (6\mu +9\lambda) }2  \bigl(\kappa_1(\gamma)-  \kappa_1(\bar \gamma) \bigr)^2
   - \mu \sum_{i=1}^3 \bigl(\gamma_i - \kappa_1(\gamma) \bigr)^2 \biggr)
  \end{multline*}
 and apply the Continuous Mapping Theorem \cite{van_der_vaart} to the smooth bijection 
\begin{align*} & (\gamma_1,\gamma_2,\gamma_3)\mapsto (\kappa_1,\kappa_2,\kappa_3)  \mbox{
with Jacobian }
& \\ &
  \biggl[ \frac{ \partial \kappa_i }{\partial \gamma_j} \biggr]=  \left( {\small 
  \begin{matrix} 1/3 &   (\gamma_1- \kappa_1) 2/3  &   (\gamma_1- \kappa_1)^2 -\kappa_2 \\
  1/3 &    (\gamma_2- \kappa_1) 2/3  &   (\gamma_2- \kappa_1)^2 -\kappa_2  \\
1/3 &    (\gamma_3- \kappa_1)  2/3 & (\gamma_3- \kappa_1)^2 -\kappa_2 \end{matrix} }\right)
\mbox{ satisfying }\det\biggl(
\biggl[ \frac{ \partial \kappa_i }{\partial \gamma_j} \biggr]\biggr)= V(\gamma)2/9.
&\end{align*}
By changing variables the Vandermonde determinant cancels out, and the resulting joint central moments density is given by
\begin{align*}
  q( \kappa_1,\kappa_2,\kappa_3)= 
  \frac{ 18 \mu^{5/2} \sqrt{ 2\mu +3\lambda }}{  \pi} 
  \exp\biggl( - \frac{ (6\mu +9\lambda) }2  \bigl(\kappa_1-   \kappa_1(\bar\gamma) \bigr)^2
   - 3\mu \kappa_2 \biggr).
\end{align*}
It follows by an optimization argument that the support of the $\kappa_3$ conditional 
distribution given $\kappa_2$ is the interval $[ - \kappa_2^{3/2} /\sqrt 2  ,  \kappa_2^{3/2} /\sqrt  2]$. 
We do a further change of variables setting $\kappa'=(\kappa_1,\kappa_2, \tau_3)$ with  $\tau_3= \kappa_3 \kappa_2^{-3/2}\sqrt{2}$, 
obtaining
\begin{align}
 & q( \kappa_1,\kappa_2,\tau_3)= 
  \label{test:statistics:distribution}
  \sqrt{\frac{6 \mu + 9 \lambda }{2\pi} }  \exp\biggl( -\frac{ 6\mu + 9 \lambda}
  2 (\kappa_1-  \kappa_1(\bar \gamma) )^2 \biggr) d\kappa_1 
\times   & \\ \nonumber &
  {\bf 1}( \kappa_2 \ge 0)
  \frac{ (3\mu)^{5/2} }{ \Gamma(5/2) }\exp\bigl( - 3 \mu \kappa_2 \bigr) 
  \kappa_2^{3/2} 
  d\kappa_2  \times   {\bf 1}(  |  \tau_3 |\le 1 ) \frac{ d \tau_3 } 2   , 
&\end{align}
 which factorizes as the distribution of independent random variables 
 
 \noindent $\kappa_1\sim{\mathcal N}\bigl(\kappa_1(\bar\gamma),1/(6\mu+9\lambda)\bigr)$, 
 $\kappa_2 \sim \bigl( \chi_{5}^2 /(6\mu) \bigr) $
 and $\tau_3$  uniformly distributed on $[-1,1]$
 \end{proof}
 Related ellipticity and sphericity 
 measures are {\it Fractional Anisotropy} \cite{basser1995}
 \begin{align*} \mbox{FA}=
  \sqrt{ \frac { 3 \tr(D^2) - \tr(D)^2 } { 2 \tr(D^2 ) }}=
  \sqrt{  \frac{ 3\kappa_2 }{ 2 (\kappa_1^2+ \kappa_2) }  }  \;,
 \end{align*}
 {\it Relative Anisotropy} \cite{basser-pajevic03}
 \begin{align*}
   \mbox{RA}=\sqrt{\frac { 3 \tr(D^2) - \tr(D)^2  }{  \tr(D)^2     } }= \frac{\sqrt{ \kappa_2} } { | \kappa_1| } \; ,
 \end{align*}
 and   
 {\it Volume Ratio} \cite{pierpaoli_et_al:1994}
 \begin{align*} 
 \mbox{VR}= 27 \;\frac{ \det\bigl(D \bigr) }
 {\tr\bigl( D \bigr)^3 }=\frac{ \gamma_1 \gamma_2 \gamma_3 }{  \kappa_1(\gamma)^3 }\;, \quad \mbox{ where }
\gamma_1\gamma_2 \gamma_3 =\kappa_3(\gamma) + \kappa_1(\gamma) ^3-  \frac 3 2 \kappa_1(\gamma) \kappa_2(\gamma)\; .
 &\end{align*} 
\begin{corollary} \label{volume:ratio:corollary} 
In the settings of  Theorem \ref{spectral:clustering} with dimension $m=3$, under the sphericity hypothesis $H_0$,
there are two  possible asymptotic regimes: 
\begin{enumerate} 
\item  \label{regime:1} when $\bar D=0$ the sequence of statistics
\begin{align} \label{asymptotic:regime:1}&
\biggl(
\mbox{FA}(\gamma^{(n)}) \; ,\; \mbox{RA}( \gamma^{(n)}) , \;  \bigl( 1-\mbox{VR} ( \gamma^{(n)}  ) \bigr) ,  \bigl(
\tau_1^{(n)} \bigr)^2, \tau_2^{(n)}, \tau_3^{(n)}
\biggr)
\end{align}
converges jointly in distribution to the random vector
\begin{multline}\label{zero:regime}
 \biggl( \sqrt{    \frac{  3\chi_2^5 }{ 2\chi_5^2+  \chi_1^2  12/(9\lambda+6)    } }
, \sqrt{\frac {( 3\lambda+ 2)} 2 \frac{\chi_5^2}{\chi_1^2}}, \\   \nonumber
\frac{ (9\lambda+6) } 4   \frac{ \chi^2_5}{\chi_1^2 }
 -  \biggl\{ 
 \biggl (  {3\lambda} +2\biggr)  \frac{ \chi^2_5}{\chi_1^2 } \biggr\}^{3/2} \frac U 4
 , \frac{ \chi_1^2}{ 9\lambda+6}  ,\chi_5^2 ,  U\biggr) 
\end{multline}
with independent $\chi^2_1,\chi_5^2$ and $U\sim\mbox{Uniform}[-1,1]$.
\item  \label{regime:2} Otherwise,  the rescaled statistics
\begin{align}\label{FA:2 }\tau^{(n)}_4=&
2\sqrt{ a^{(n)} }\;|\kappa_1(\gamma^{(n)} )|\;\mbox{FA}(\gamma^{(n)})\simeq 
2\sqrt{ a^{(n)} }\;|\kappa_1(\bar\gamma)|\;\mbox{FA}(\gamma^{(n)})  
& \\  \label{asymptotic:regime:2} 
\tau_5^{(n)}= & 4 a^{(n)}  k_1( \gamma^{(n)})^2 \bigl( 1- \mbox{VR} ( \gamma^{(n)} )
  \bigr) \simeq 4 a^{(n)}  k_1( \bar \gamma  )^2 \bigl( 1- \mbox{VR}( \gamma^{(n)} ) \bigr) 
  & \\  \tau^{(n)}_6= &
 - 4 a^{(n)}  k_1(  \gamma^{(n)} )^2 \log \big\vert \mbox{VR} ( \gamma^{(n)}  ) \big\vert\simeq
 - 4 a^{(n)}  k_1(  \bar\gamma )^2 \log \big\vert \mbox{VR} ( \gamma^{(n)}  ) \big\vert 
\end{align}
are asymptotically equivalent
with 
\begin{align*}
\big \vert \bigl(\tau_4^{(n)} \bigr)^2 - \tau_2^{(n)} \big\vert \stackrel{ P}{\longrightarrow}0 , \;
\big\vert\tau_5^{(n)} - \tau_6^{(n)}  \big\vert \stackrel{ P}{\longrightarrow} 0,\mbox{ and }
\big\vert\tau_5^{(n)}- \tau_2^{(n) } \big\vert \stackrel{ P}{\longrightarrow} 0 
\end{align*}
in probability, and $\tau_2^{(n)} \stackrel{law}{\to} \chi^2_5$.
\end{enumerate}
\end{corollary}
 \begin{remark} 
Corollary \ref{volume:ratio:corollary} generalizes Thm.8.3.7 in \cite{muirhead}
 on VR asymptotics without positivity assumptions.
 In order to use the VR statistics to test the isotropy of the mean $\bar D$,
one should first test the hypothesis $\kappa_1(\bar\gamma)=0$, under which
\begin{align*}
  (9\lambda+6) a^{(n)}  \kappa_1(\gamma^{(n)} )^2 \stackrel{law}{\longrightarrow} \chi_1^2 \; .
\end{align*}
If this  hypothesis is accepted, we assume that we are
in the asymptotic regime \eqref{regime:1}
and construct a conditional sphericity test by using
the conditional distribution of $\mbox{VR}(\gamma^{(n)})$ 
given  $\bigl\{ a^{(n)}\kappa_1(\gamma^{(n)})^2=t\bigr\}$, which converges in distribution to the law of
\begin{align*} 1 +\biggl( 
 \frac{  \chi_5^2 } {3 t} \biggr)^{3/2} \frac U 4 - \frac{ \chi^2_5}{ 4 t}\; ,
 \end{align*}
with $\chi_5^2$ independent from $U\sim\mbox{Uniform}[-1,1]$.
If the hypothesis $\kappa_1(\bar\gamma)=0$ is rejected we use the 
rescaled volume ratio statistics $\tau_5^{(n)}$ in \eqref{asymptotic:regime:2}.

 Eigenvalue central moment statistics have been considered earlier in the DTI literature,
the distribution of $\tr(D)$ for $D$ isotropic Gaussian is derived in \cite{basser:pajevic2003},
the variance is discussed in
\cite{basser1995},\cite{ibrahim},\cite{schwartzman}, and skewness in \cite{basser1997}.
 Note that under $H_0$  
the limit laws of $\tau_2^{(n)}, \tau_3^{(n)}$ are parameter free. However 
evaluating $\tau_2^{(n)}$ requires knowledge of the scaling sequence normalization,
    while  $\tau_3^{(n)}$ does not.
   $\vert \tau_3^{(n)}\vert$ can be used as  two-sided test statistics,
      accepting the sphericity hypothesis with confidence level $\alpha$ 
       when $\vert\tau_3^{(n)} \vert \in \bigl ( (1-\alpha)/2, (1+\alpha)/2 \bigr)$. 
       The left-tail rejection region corresponds to
       the anomalous situation with $(\gamma_1^{(n)}-\gamma_2^{(n)} ) \simeq (\gamma_2^{(n)}-\gamma_3^{(n)})$, and
         the right tail corresponds to 
         $\gamma_1^{(n)} \simeq \gamma_2^{(n)} \gg \gamma_3$ or $\gamma_1^{(n)} \gg \gamma_2^{(n)} \simeq \gamma_3^{(n)}$.
We can test for symmetries with a 
sequence of confidence levels $p^{(n)}=\P\bigl( \chi^2_5 < c^{(n)}\bigr)$, with $c^{(n)}\to\infty$
and $c^{(n)}/a^{(n)} \to 0$,
   and construct an asymptotically superefficient eigenvalue estimator $\hat \gamma^{(n)}$:
\begin{enumerate}   \item
If  $\kappa_2\bigl(  \gamma^{(n)} \bigr) < c^{(n)} /\bigl( 6 a^{(n)} \bigr) $,
 accept the isotropy hypothesis and set 
 $\hat \gamma_1^{(n)} = \hat \gamma_2^{(n)} = \hat \gamma_3^{(n)}= \kappa_1\bigl( \gamma^{(n)} \bigr)$

 \item
else if 
  \begin{align*} \bigl(\gamma_1^{(n)} -\gamma_2^{(n)} \bigr)^2  a^{(n)} < - 2 \log( 1- p^{(n)} )  \; , \end{align*}
 accept the oblate  tensor hypothesis and
set $ \hat \gamma_1^{(n)} = \hat \gamma_2^{(n)}= (\gamma_1^{(n)} + \gamma_2^{(n)})/2 > \hat \gamma_3^{(n)}= \gamma_3^{(n)}$, 

\item 
else if 
 \begin{align*}\bigl(\gamma_2^{(n)} -\gamma_3^{(n)} \bigr)^2  a^{(n)}  < -2 \log( 1- p^{(n)} ) \; ,
 \end{align*}
accept the prolate diffusion tensor hypothesis and
set $ \hat \gamma_1^{(n)}= \gamma_1^{(n)}<  \hat \gamma_2^{(n)} = \hat \gamma_3^{(n)}= (\gamma_2^{(n)} + \gamma_3^{(n)})/2$, 

\item 
otherwise reject  the hypothesis that the tensor has symmetries
and use  the unmodified estimator $\hat\gamma^{(n)}=\gamma^{(n)}$.  
\end{enumerate}
The situation with mean matrix $\bar D=0$ arises in two-sample problems. 
Consider two $m\times m$ symmetric random matrices $D',D''$, which  are measured with independent and  isotropic Gaussian noises, with 
precision matrices $A(\mu',\lambda')$ and $A(\mu'',\lambda'')$,  and means $\bar D',\bar D''$,
respectively. 
Their difference $D=( D'-D'') $ is again symmetric Gaussian with mean 
$\bar D=( \bar  D'-\bar D'') $ and isotropic precision matrix $A(\mu, \lambda )$, with parameters
\begin{align*}&\mu=\frac{\mu'\mu''}{\mu'+\mu''},\;\lambda=\frac{ 2\alpha \mu  }{ \mu'+\mu''- m \alpha}, 
\; \alpha= \frac{ \lambda' \mu''} { 2\mu'+ m \lambda'} +   \frac{ \lambda'' \mu'} { 2\mu''+ m \lambda''}\; .
&\end{align*}
In order to test the hypothesis  $\bar D'=\bar D''$, one could use the statistics
\begin{align}\label{chi26:statistics}
 \bigl\{ 2 m\mu \kappa_2( D ) + \bigl( 2m \mu+ \lambda m^2 \bigr)\kappa_1( D )^2   \bigr\}\sim\chi^2_{(m+1)m/2}.
\end{align}
Testing  equality in distribution of two sample matrix eigenvalues and eigenvectors separately 
has been discussed in  \cite{schwartzman}, under the hypothesis of asymptotically  Gaussian and isotropic error,
generalized in \cite{schwartzman2010} to non-isotropic error covariances.
\end{remark}
\section{ASYMPTOTIC STATISTICS IN DTI UNDER RICIAN NOISE}  \label{section:rician}
We consider an ideal DTI experiment with measurements following the Rician likelihood
\begin{align} \label{observation:density}
 p_{S,\eta^2}( Y )= \frac{ Y}{\eta^2}
 \exp\biggl( - \frac{ Y^2+ S^2}{ 2\eta^2} \biggr)   I_0\bigl(  Y S / \eta^2 \bigr ), 
\end{align}
where  $S$ is the signal,
 $Y$ the observation, $\eta^2$ the  noise parameter,
and  $I_{\ell}(z)$  is the modified Bessel function of first kind of order $\ell$.
 The signal is determined by the 2nd-order tensor model
 \begin{align} \label{signal:eq}
 S=S(g,D)=\rho \exp( - g D g^{\top} ) \;,  \rho > 0, \; g\in \R^3, \; D\in \R^{3\times 3}\; ,
 \end{align}
where $D$ is the (symmetric)  diffusion tensor,  $\rho$ is the unweighted reference signal,
and $g$ is the applied magnetic field gradient.
The function $g\mapsto S(g,D)/\rho$ is interpreted as the Fourier transform
of the displacement distribution of a water molecule  undergoing Gaussian diffusion in
an unit time interval, and the problem is to estimate the diffusion tensor $D$ from the
noisy spectral measurements $Y$. 
 For fixed $\rho$ and $\eta^2$ we denote the loglikelihood of $D$ as
 \begin{align*}
    L(D)= \log\bigl(  p_{S,\eta^2}( Y ) \bigr) \; .
 \end{align*}
The observed information with respect to the tensor parameter $D$ is given by
\begin{align*} & J_o(D)=
 -  \biggl[ \frac{\partial^2 L(D) }
             { \partial D_{ij}\partial D_{l r }}  \biggr]_{i\le j, l\le r} & \\ &=\frac{  S^2 }{ \eta^2}
 \biggl( 2  +  \frac {Y^2 }{\eta^2}  \biggl\{
  \frac{  I_1\bigl(   S Y / \eta^2   \bigr  )^2} { 
  I_0\bigl(  S Y / \eta^2  \bigr )^2} -1   \biggr\}   \biggr)\biggl [( 2-\delta_{ij})(2-\delta_{lr})
  g(i) g(j) g(l) g(r)  \biggr]_{i\le j,l\le r} \;. &
\end{align*}
and the Fisher information is obtained by integrating out
the data $Y$ with respect to  \eqref{observation:density} under the signal model \eqref{signal:eq} with
  tensor parameter $D$, obtaining 
\begin{align} \label{fisher:info}
 J(D)&= E_{D}\bigl( J_o(D) \bigr) =E_{D}\biggl(
 \biggl[ 
 \frac{\partial L(D) }{\partial D_{ij} } \frac{\partial L(D) }{\partial D_{lr} }  \biggr]\biggr)_{i\le j,l\le r}
& \\ \nonumber & =w(S/\eta )   \biggl[ 
 (2- \delta_{ij})(2-\delta_{lr})  g(i) g(j) g(l) g(m)  \biggr]_{i\le j,l\le r}\;,
&\end{align}
 depending on the signal to noise ratio (SNR) $S/\eta$  of the complex Gaussian error model through the weight function 
\begin{align*}
  w(z)=
  \frac{ \exp(-z^2/2)} {z^2}  \int_0^{\infty}  x^3\exp\biggl( - \frac{x^2}{2z^2}\biggr) \frac{ I_1(x)^2 } {  I_0(x) } dx -z^4   \ge 0\;,
\end{align*} 
see \cite{idier}.
Note that necessarily  $J_{ij,ij}(D)=4 J_{ii,jj}(D)$  $\forall\;1\le j<i \le 3$.
By replacing the Rician density \eqref{observation:density} with another likelihood which is function of the SNR,
we always obtain  a  Fisher information of the form  \eqref{fisher:info},
with a different weight function.

We now consider a sequence of DTI-experiments, with  measurements $\bigl(Y_k^{(n)}: k=1,\dots,M^{(n)}  \bigr)$ from respective
 signals $\bigl(S^{(n)}_k : k=1,\dots,M^{(n)}  \bigr)$,
 corresponding to the gradients  $\bigl(g^{(n)}_k : k=1,\dots,M^{(n)}  \bigr) \subset \R^3$,
and  denote the scaled Fisher Information as
  \begin{align} \label{fisher:n}
 J^{(n)}(D)= \frac 1 { M^{(n)} } \biggl[ 
 (2- \delta_{ij})(2-\delta_{lr})                
 \sum_{k=1}^{M^{(n)}}w\bigl(   S^{(n)}_k /\eta \bigr)   
  g_k^{(n)}(i) g_k^{(n)} (j) g_{j}^{(n)}(l) g_k^{(n)}(r)  \biggr]_{i\le j,l\le r}.
\end{align}
Assume  that $M^{(n)}\to\infty$  and the sequence of  discrete gradient distributions
 \begin{align*}
\pi^{(n)}(dg)=\frac 1 {M^{(n)} } \sum_{k=1}^{M^{(n)} } {\bf 1}\bigl(  g_k^{(n)}  \in dg \bigl)  
 \end{align*}
converges weakly to a  probability $\pi$ on $\R^3$, which implies
 \begin{align}\label{fisher:infty:noniso}
  & \lim_{n\to\infty}  J^{(n)}(D) = J^{(\infty)}(D)
   & \\ \nonumber &= \biggl[ 
 (2- \delta_{ij})(2-\delta_{lr})   
 \int_{\R^3} w\bigl( 
  \exp\bigl( -  g D g^{\top} \bigr)\rho/\eta \bigr)   g(i) g(j) g(l) g(r) \pi(dg)    \biggr]_{i\le j,l\le r}.   
 &\end{align}
Let   $D^{(n)}$ be a regular statistical estimator of the tensor parameter, as for example the Maximum Likelihood Estimator (MLE),
the penalized MLE, the Bayesian Maximum a Posteriori Estimator (MAP), or the posterior mean, 
based on  the data $(Y_k^{(n)}: 1\le k \le M^{(n)} )$ with gradients  $(g_k^{(n)}: 1\le k \le M^{(n)} )$.
 When $0<\det( J^{(\infty)} )<\infty$,
 under the tensor model 
with true parameter $\bar D$,  all these regular  estimators are consistent with asymptotically Gaussian error, such that 
\begin{align} \label{gaussian:limit:estimator}
 \sqrt{M^{(n)} } ( D^{(n)} - \bar D ) \stackrel{law}{\to} X \sim {\mathcal N}\bigl( 0,  (J^{(\infty)}(\bar D) )^{-1} \bigr)\;.  
\end{align}
\subsection{Isotropic Gaussian limit error distribution}
When  $J^{(\infty)}(\bar D)=A(\bar\mu,\bar\mu)$ as in   \eqref{Omega_D:matrix:2nd} for some $\bar\mu>0$,
the Gaussian limit distribution  \eqref{gaussian:limit:estimator} is isotropic. In such case
Theorem \ref{spectral:clustering}, Corollary 
\ref{asymptotic:eigdensity} and Lemma \ref{sphericity:test:thm}
apply with $a^{(n)}=\bar\mu  M^{(n)} $ and $\lambda=1$.
When the true tensor  $\bar D= \bar \gamma \mbox{I}$ is isotropic, 
and  the asymptotic gradient design distribution $\pi(dg)$ is radially symmetric,  asymptotic isotropy is  achieved with
\begin{align}\label{fisher:infty:iso}  
 & J^{(\infty)}(\bar D)  = & \\ &\nonumber 
   \biggl( \int_0^{\infty} w\bigl( 
  \exp\bigl( - \bar \gamma b \bigr)\rho/\eta \bigr) \nu(db)
    \biggr)\biggl[ 
 (2- \delta_{ij})(2-\delta_{lr})    
 \int_{ {\mathcal S}^2 }   u(i) u(j) u(l) u(r) \sigma(du)    \biggr]_{i\le j,l\le r} 
& \\ \nonumber  &
 =  A(1,1)  \bar \mu, \quad\quad 
\bar \mu  =  \biggl( \int_0^{\infty} w\bigl( 
  \exp\bigl( - \bar \gamma b \bigr)\rho/\eta \bigr)  \nu(db)
    \biggr)\bigg/15 ,   & 
\end{align}
where $b=\parallel g\parallel^2$, referred as  $b$-value, is integrated with respect to 
$$\nu(db)=\pi( \{g:\parallel g \parallel^2 \in db \})\;,$$
and $u=g/\parallel g\parallel$ has  uniform  distribution  $\sigma(du)$  on the 
surface of the unit sphere ${\mathcal S}^2= \{u \in \R^3: \parallel u \parallel = 1 \}$.
A more general  condition  implying \eqref{fisher:infty:iso} is the following: the asymptotic gradient design distribution decomposes as
\begin{align}\label{gradient:distribution}
\pi( dg) = \nu( db ) s(du |b ) \;,
\end{align}
where for $\nu$-almost all $b$-values,   the conditional probability on ${{\mathcal S}^2}$ is such that
 \begin{align}  \label{4:moment:condition}
    \int_{{\mathcal S}^2} f(u)s(du|b) = \int_{{\mathcal S}^2} f(u)\sigma(du)
  \end{align}
for all homogeneous polynomials $f(u_1,u_2,u_3)$ of  degree $t=4$.
\begin{proposition} \label{asymptotic:optimality:1} When the true diffusion tensor $\bar D$ is isotropic, the uniform gradient distribution $\sigma(du)$ maximizes $\det(  J )$
among all probability distributions on the unit sphere.  
 \end{proposition}
\begin{proof} 
When $J$ is invertible 
 we have \cite[Theorem 8.1]{magnus}
\begin{align}\label{eq:logdetdiff}
  d\log \det(J)=
  \tr( J^{-1} dJ )  , \quad \; d^2 \log \det(J) =  - \tr(  J^{-1} dJ J^{-1} dJ ) \le  0 \; ,
\end{align}
 which implies that the  function $J\mapsto \log \det(J)\in \R\cup\{ -\infty\}$ is concave, and a local maximum is also a global maximum. 
Let $\nu(du)$ be   probability measure on ${\mathcal S}^2$,
and consider a small perturbation of the uniform measure $\sigma$ in the direction $\nu$.
By taking the  differential using \eqref{eq:logdetdiff}, we obtain
\begin{align} \label{det:proof}
\lim_{\varepsilon \to 0+} &
\frac{  \log \det J( (1-\varepsilon) \sigma + \varepsilon \nu) - \log \det J(\sigma) }{ \varepsilon }  & \\ \nonumber &
= \int_{{\mathcal S}^2 } \biggl(  \sum_{i\le j, l\le r}  J^{-1}_{ij,lr} (\sigma)( 2-\delta_{ij}) (2- \delta_{lr} )    u_i u_j u_l u_r \biggr)(\nu-\sigma)(du) =0
\;, &\end{align}
where  since $J^{-1}(\sigma)$ is also isotropic,   for every $u,v\in {\mathcal S}^2$
\begin{align*}
\sum_{i\le j, l\le r}  J^{-1}_{ij,lr} (\sigma)( 2-\delta_{ij}) (2- \delta_{lr} )    u_i u_j u_l u_r 
=\sum_{i\le j, l\le r}  J^{-1}_{ij,lr}(\sigma) ( 2-\delta_{ij}) (2- \delta_{lr} )    v_i v_j v_l v_r
\end{align*}
   and the integrand in $\eqref{det:proof}$ is constant, which means that  $ \det\bigl(J(\sigma)\bigr)$ is a global maximum. 
\end{proof} 
This shows that, when the true tensor $\bar D$ is isotropic,  asymptotically uniform 
gradient designs are most informative, minimizing the Gaussian entropy of the asymptotic estimation error
\begin{align*}
  H(J^{(\infty)})=\mbox{const.} - \log\bigl( \det(J^{(\infty)} )\bigr)/2 \; .
\end{align*}
  In the next section we introduce
discrete gradient distributions which attain the same bound.   
\subsection{Spherical $t$-designs in Diffusion Tensor Imaging} \label{section:t-spherical}
A spherical $t$-design $\Upsilon \subset {\mathcal S}^{m-1}$ is a finite subset of $m$-dimensional unit vectors with the property
\begin{align}  \label{4:moment:condition:2}
   \int_{{\mathcal S^{m-1}}  }  f( u ) \sigma(du) = \frac 1 { \# \Upsilon } \sum_{  \upsilon \in  \Upsilon }  f( \upsilon)
\end{align}
for all polynomials $f(u_1,\dots,u_m)$ of degree $r\le t$, where  $\sigma$ is the uniform probability measure on ${\mathcal S}^{m-1}$,
and $\#\Upsilon$ is the number of points in $\Upsilon$.
In other words, a spherical $t$-design is a quadrature rule on $S^{m-1}$
with constant weights.  
The algebraic theory behind such designs  is deep and beautiful \cite{delsarte}, for a recent survey see \cite{bannai,an}.
In  particular, in dimension $m=3$,  spherical $t$-designs of order $t\ge 4 $ satisfy   \eqref{4:moment:condition}.
A  database of spherical $t$-designs on ${\mathcal S}^2$ computed by Rob Womersley is available at his webpage
\url{http://web.maths.unsw.edu.au/~rsw/Sphere/EffSphDes/}. Table \ref{table:design} displays the sizes
of these designs and Fig. \ref{fig:gradients:4th_order} shows a spherical $t$-design of order 4 with 14 gradients from
Womersley's database.

When $\Upsilon= - \Upsilon$,  
we say that the spherical design 
is {\it  antipodal}.
Two well known examples (see \cite{basser-pajevic03},\cite{batchelor})
are   the regular icosahedron and  its dual, the regular dodecahedron, 
whose vertices form antipodal spherical $t$-designs of order $5$ with sizes $12$ and $20$, respectively.
Note that any two antipodal gradients produce the same DTI-signal. 
Starting from  an antipodal spherical $t$-design $\Upsilon$ 
and selecting one gradient from  each antipodal pair  $\{u, -u \}\subset \Upsilon$, 
we obtain a  design $\Upsilon'$ of size $\#\Upsilon'=\#\Upsilon/2$ 
which satisfies \eqref{4:moment:condition:2}   for all homogeneous polynomials $f$ of even degree $\le t$.
Figures  \ref{fig:gradients:icosahedron}-\ref{fig:gradients:dodecahedron}
show respectively the intersection of the northern hemisphere with the regular icosahedron and 
dodecahedron,  forming gradient designs of size $6$ and $10$ which satisfy \eqref{4:moment:condition:2} 
 for all  homogeneous polynomials $f$ of degrees $2$ and $4$.

   In the DTI experiment, for a finite subset of $b$-values
$0< b_1^{(n)}\le \dots \le b_n^{(n)}$ and respective spherical $t$-designs  $\Upsilon_{\ell}^{(n)}$  of order   $t^{(n)}_{\ell}\ge 4$, 
we construct the gradient set as the union of  shells
\begin{align*}
 G^{(n)}= \bigcup_{\ell=1}^n 
\Upsilon_{\ell}^{(n)}\sqrt{b_{\ell}^{(n)}}  \subset \R^3\; .
\end{align*}
The resulting gradient distribution 
\begin{align*}
 \pi^{(n)}( B ) = \frac{ \#\bigl(  G^{(n)} \cap B \bigr) }
{ \# G^{(n)} } , \quad   B \subseteq \R^3 \; .
\end{align*}
satisfies \eqref{gradient:distribution}, and  when the true tensor $\bar D=\bar\gamma \mbox{I}$ is totally symmetric,  we  have
\begin{align*}
  J^{(n)}(\bar D)=\bar \mu^{(n)} A( 1 , 1)
\end{align*}
with
\begin{align*}
 \bar \mu^{(n)}= \frac 1 { 15 \#G^{(n)}  } \sum_{\ell =1}^n 
  w\bigl( 
  \exp\bigl( - \bar \gamma b^{(n)}_{\ell}  \bigr)\rho/\eta \bigr)  \# \Upsilon^{(n)}_{\ell}  \;,
\end{align*}
i.e. the Fisher information  coincides with the precision matrix of an Isotropic Gaussian matrix distribution.
When $\Upsilon \subset S^{2}$ is a spherical $t$-design and $O \in SO(3)$ is a rotation matrix,
the rotated design $ O\Upsilon $ is  a spherical $t$-design as well.
Since the true tensor   $\bar D $ is unknown, and possibly it is not isotropic,  
in practice it is advisable to choose the gradient directions covering  ${\mathcal S}^2$ as  uniformly as possible.
To achieve that,
different $t$-designs can be rotated with respect to each other in order maximize the  spread between gradient directions.
Namely, starting from a collection of spherical $t$-designs $\Upsilon_1^0,\dots, \Upsilon_n^0$ of respective orders $t_k, \; 1\le k \le n$
 we find the optimized design $\Upsilon_k^{(n)} = O^*_{k} \Upsilon_k^0$, $1\le k\le n$, where 
 $O^{*}_1,\dots ,O^*_n$ are rotation matrices maximizing
\begin{align} \label{geodesic:dist:critera}
 \max_{ O_1,\dots, O_n \in SO(3)} \min_{1\le k < l \le n}  \bigl\{  \mbox{dist}(  O_k \Upsilon^0_k , O_l \Upsilon_l^0 )  \bigr\}\; ,
 \end{align}
 with  $\mbox{dist}( U,V)= \sup_{ u  \in U, v \in V } \mbox{dist}(u,v) $ and  $\mbox{dist}(u,v) $ is the  geodesic
  distance on  ${\mathcal S}^2$.
This can be achieved by a  greedy  iterative algorithm, where in turn \eqref{geodesic:dist:critera} is optimized
with respect to each single $O_k$ keeping fixed the other rotations until convergence to
a fixed point. 
Fig.  \ref{fig:gradients} shows a gradient sequence obtained in such a way, with colors
correponding to spherical $t$-designs on different shells.
The benefits of these gradient designs are illustrated  in the next paragraph.
 \section{ILLUSTRATION OF THE METHODS} \label{section:mcsim}
\subsection{Monte Carlo study with  isotropic Gaussian noise}\label{MC:study:I}
Fig. \ref{fig1}   shows the results from a Monte Carlo study with a sample of $N=10000$ i.i.d. $3\times 3$ symmetric random matrices
with isotropic Gaussian density \eqref{isotropic_prior:2}  with precision parameters $\mu=1/2$, $\lambda=0$,
for
various choices of the diagonal  mean matrix:
\begin{enumerate}

\item[(a)]  $\bar D=0$, correponding to the $3\times3$ GOE,

\item[(b)] $\bar D$ isotropic,  with $\bar\gamma_1=\bar\gamma_2=\bar\gamma_3=15$, 

\item[(c)] $\bar D$ prolate, with $\bar\gamma_1=15  > \bar \gamma_2 = \bar\gamma_3 =3$, 

 \item[(d)] $\bar D$ oblate, with $\bar\gamma_1=\bar\gamma_2 =15   > \bar\gamma_3 =3$. 

 \end{enumerate}For comparison we show in Fig. \ref{fig1}e  
i.i.d. eigenvalue pairs from the $2\times 2$-Gaussian orthogonal ensemble, and in  \ref{fig1}f
 i.i.d. pairs of independent  standard Gaussian random variables. The empirical joint  eigenvalue distribution
avoids the diagonal, in agreement with \eqref{eigenvalue:density}. 
We see that the fluctuations
of the eigenvalues corresponding to the same $\bar D$ eigenspaces around their mean 
are distributed like the GOE corresponding to the dimension of the eigenspace. 
 One can see also some
differences between the GOE eigenvalue distribution in dimension 2 (in Fig. \ref{fig1}e, 
sampled with precision parameters $\mu=1/2,\lambda=0$,
which agrees with
\ref{fig1:positive:prolate} and
\ref{fig1:positive:oblate}), and dimension 3 (in Fig. \ref{fig1:GOE3}, which
agrees with \ref{fig1:positive:totally:symmetric}).

Fig. \ref{fig3} shows that, in the case with prolate mean matrix, 
the empirical distribution of the  cluster barycenter $(\gamma_2+ \gamma_3)/2$ 
fits very well
the Gaussian distribution.

Fig. \ref{fig:power:comparison}
shows the behaviour of the sphericity test statistics $\tau_2,\tau_4,\tau_5$
under  Gaussian matrix distributions
with the same  isotropic  precision matrix $ A(2,2)$, and different  means: namely
a spherical mean tensor, and  $15$ prolate mean tensors, all  with the same mean diffusivity $\kappa_1(\bar D)=15$,
and  FA in  $(0.01,0.15]$. We can see that at this noise level, under the null hypothesis, the
distributions of these three test statistics fit very well the asymptotic $\chi_5^2$ distribution,
while under prolate alternatives the corresponding sphericity tests have approximately the same power
 at all significance levels.
 
Fig. \ref{fig:eigenvector:oblate} displays on the unit sphere the orthonormal eigenvector
triples from the Gaussian model with isotropic noise parameters $\mu=1/2,\lambda=0$, with $N=200$ i.i.d. replications.
On the left side figure the mean tensor diagonal and totally anisotropic  with $\bar\gamma_1=15,\bar\gamma_2=7.5, \bar\gamma_3=3$.
On the right the mean tensor is diagonal and oblate, with 
$\bar\gamma_1=\bar\gamma_2=15, \bar\gamma_3=3$, and the
eigenvectors corresponding to the first two eigenvalues
 are uniformly  distributed   around  the equator.
 
\subsection{Monte Carlo study of sphericity test statistics based on DTI data with Rician noise}\label{MC:study:II}
In order to validate the asymptotic results of Lemma
\ref{sphericity:test:thm} and Corollary \ref{volume:ratio:corollary}, 
 we conducted  another large Monte Carlo study,
 with DTI data  simulated under the Rician
 noise model 
 with   ground truth parameters
 $\eta^2=64.056$, $\rho=110.046$, and
 isotropic  diffusion tensor $\bar D=6.622\times 10^{-4}\times\mbox{\rm Id}$ ${\rm mm^2/s}$.
 For each of  the experimental designs 1-5 below, which have increasing number of acquisitions, we simulated $N=50000$ replications
 of the dataset, and for each replication 
  we computed the MLE  $D^{(n)}$ based on the simulated data by using the EM-algorithm from \cite{liu}.
 The empirical distribution of the sphericity statistics
 $\tau_2^{(n)},\tau_3^{(n)}$ \eqref{test:statistics:T} and $\tau_5^{(n)}$ \eqref{asymptotic:regime:2} 
 with their theoretical limit distributions are displayed correspondingly 
 in Figures  \ref{fig:sphericity:test:14}-\ref{fig:sphericity:test:32x3}.
 \begin{enumerate}
 \item[Design 1:]  Spherical $t$-design of order  $4$ with $14$ gradients computed by R. Womersley, shown
  in Fig. \ref{fig:gradients:4th_order}, with  $b$-value $996$  ${\rm s/mm^2}$, 
  and one acquisition at zero $b$-value, for a total of $15$ acquisitions.
  The corresponding Fisher information is given by 
  \begin{align*}
  J^{(n)}(\bar D)=      \bar \mu^{(n)} A_{3}(1,1), \quad\bar\mu^{(n)}=  4.63  \times 10^7  {\rm s^2/{mm^4} },
\end{align*}
  and the ML estimator $\mbox{vec}(D^{(n)})$ has a Gaussian approximation
                                           with mean         $\mbox{vec}(\bar D)$ and isotropic covariance
\begin{align*}\Sigma^{(n)}=J^{(n)}(\bar D)^{-1}= 10^{-9} \times\left(
\begin{matrix}
    8.64  &   -2.16   &  -2.16   &        0   &        0    &       0    \\
   -2.16  &    8.64   &  -2.16   &        0   &        0    &       0      \\
   -2.16  &   -2.16   &   8.64   &        0   &        0    &       0        \\
         0  &         0   &        0   &   5.4   &        0    &       0         \\
         0  &         0   &        0   &        0   &   5.4    &       0          \\
         0  &         0   &        0   &        0   &        0    &  5.4          \\
\end{matrix} \right)                                          {\rm \frac{ mm^4}{s^2} }\; .
\end{align*}
   \item[Design 2:] It is based on the icosahedron with  the $6$  gradients shown in Fig. \ref{fig:gradients:icosahedron} for each  $b$-value 
 in the set $\{   560,
         778,
         996,
        1276,
        1556,
        1898,
        2240\}$
 ${\rm s/mm^2}$,
 and one acquisition at zero $b$-value, for a total of $43$ acquisitions.  
 The corresponding Fisher information is given by 
 \begin{align*}
 J^{(n)}(\bar D)=      \bar \mu^{(n)} A_{3}(1,1), \; \bar\mu^{(n)}=  1.323  \times 10^8  {\rm s^2/{mm^4} }
 \end{align*}
                                           and the ML estimator $\mbox{vec}(D^{(n)})$ has a  Gaussian approximation
                                           with mean         $\mbox{vec}(\bar D)$ and isotropic covariance
\begin{align*}\Sigma^{(n)}=J^{(n)}(\bar D)^{-1}= 10^{-9} \times\left( {\small
\begin{matrix}
    3.02   &   -0.76   &   -0.76   &         0   &         0  &          0   \\
   -0.76   &    3.02   &   -0.76   &         0   &         0  &          0     \\
   -0.76   &   -0.76   &    3.02   &         0   &         0  &          0       \\
         0   &         0   &         0   &    1.89   &         0  &          0        \\
         0   &         0   &         0   &         0   &    1.89  &          0         \\
         0   &         0   &         0   &         0   &         0  &     1.89          \\
\end{matrix} } \right)                                          {\rm \frac{ mm^4}{s^2} } \; .
\end{align*}
 
  \item[Design 3:]  It is based on the dodecahedron with the  $10$  gradients shown in  Fig. \ref{fig:gradients:dodecahedron}
 for each  $b$-value 
 in the set 
 \begin{align*}\bigl\{   560,
         778,
         996,
        1276,
        1556,
        1898,
        2240\bigr\} \; {\rm s/mm^2} ,\end{align*}
 and one acquisition at zero $b$-value, for a total of $71$ acquisitions.
 The corresponding Fisher information is given by 
   \begin{align*}J^{(n)}(\bar D)=      \bar \mu^{(n)} A_{3}(1,1), \quad \bar\mu^{(n)}=  2.205   \times 10^8  {\rm s^2/{mm^4} },
   \end{align*}
and the  ML estimator $\mbox{vec}(D^{(n)})$ has a  Gaussian approximation
                                           with mean         $\mbox{vec}(\bar D)$ and isotropic covariance
\begin{align*}\Sigma^{(n)}=J^{(n)}(\bar D)^{-1}= 10^{-9} \times\left( {\small
\begin{matrix}
            1.81  &  -0.45   & -0.45   &       0  &        0   &       0  \\
           -0.45  &   1.81   & -0.45   &       0  &        0   &       0    \\
           -0.45  &  -0.45   &  1.81   &       0  &        0   &       0      \\
                 0  &        0   &       0   &  1.13  &        0   &       0       \\
                 0  &        0   &       0   &       0  &   1.13   &       0        \\
                 0  &        0   &       0   &       0  &        0   &  1.13         \\
\end{matrix} } \right)                                          {\rm \frac{ mm^4}{s^2} } \; .
\end{align*}
  \item[Design 4:]
  Combination of  spherical $t$-designs of orders 5,7,9,11, shown in  Fig. \ref{fig:gradients}
 on  shells corresponding to the $b$-values $\{  560,
         996,
        1556,
        2240\}$, respectively,   with one acquisition at zero $b$-value, for a total of $163$ acquisitions.
 
    The corresponding Fisher information is given by 
    \begin{align*}
J(D^{(n)})=      \bar \mu^{(n)} A_{3}(1,1),  \quad\bar\mu^{(n)}=  5.263   \times 10^8  {\rm s^2/{mm^4} },
    \end{align*}
and the ML estimator $\mbox{vec}(D^{(n)})$ has a  Gaussian approximation
                                           with mean         $\mbox{vec}(\bar D)$ and isotropic covariance
\begin{align*}\Sigma^{(n)}=J^{(n)}(\bar D)^{-1}= 10^{-10} \times\left(  {\small
\begin{matrix}
    7.6 &    -1.9   &  -1.9    &       0   &        0    &       0     \\
   -1.9 &     7.6   &  -1.9    &       0   &        0    &       0       \\
   -1.9 &    -1.9   &   7.6    &       0   &        0    &       0         \\
      0 &       0   &     0    &    4.75   &        0    &       0          \\
      0 &       0   &     0    &       0   &     4.75    &       0           \\
      0 &       0   &     0    &       0   &        0    &    4.75            \\
\end{matrix} }\right)                                          {\rm \frac{ mm^4}{s^2} }\; .
\end{align*}
  \item[Design 5:] with  $3$ repetitions of the  $32$  gradients in  Fig. \ref{gradient:table}
 for each $b$-value in  
 \begin{align*}\bigl\{ \mbox{\footnotesize 62,     249,       560,       996,    1556,        2240, 
 3049,  3982,        5040 ,6222,        7529,    8960,     10516,      12196,      14000} \bigr \}  {\rm s/mm^2}, \end{align*}
 and $3$ acquisitions at zero $b$-value, for a total of $1443$ acquisitions.
                                           The ML estimator $\mbox{vec}(D^{(n)})$ has a Gaussian approximation
                                           with mean         $\mbox{vec}(\bar D)$ and  non-isotropic covariance
\begin{align}\label{non:iso}\Sigma^{(n)}=J^{(n)}(\bar D )^{-1}= 10^{-10} \times\left( \small{ 
\begin{matrix}
    6.77  &  -3.24 &  -2.31 &  -0.07 &  -0.08  &  0.21     \\
   -3.24  &  7.04  & -2.53  & -0.11  &  0.15   &-0.05        \\
   -2.31  & -2.53  &  6.70  &  0.10  & -0.10   &-0.59          \\
   -0.07  & -0.11  &  0.10  &  1.17  & -0.14   & 0.01           \\
   -0.08  &  0.15  & -0.10  & -0.14  &  1.3   &-0.01            \\
    0.21  & -0.05  & -0.59  &  0.01  & -0.01   & 1.33             \\
\end{matrix} } \right)                                          {\rm \frac{ mm^4}{s^2} } \; .
\end{align}
\end{enumerate}
All scatterplots in Figures  \ref{fig:sphericity:test:14}-\ref{fig:sphericity:test:32x3} are consistent with 
the asymptotic independence of
the sphericity statistics $\tau_2^{(n)}$ and $\tau_5^{(n)}$ from  $\tau_3^{(n)}$.
When the experimental design is based on spherical $t$-designs of order $t\ge 4$ (Designs 1-4), with isotropic
Fisher information,
the empirical distributions of  $\tau_2^{(n)}$ and $\tau_5^{(n)}$ fit well the theoretical limit distribution $\chi^2_5$ (Figures \ref{fig:sphericity:test:14}-\ref{fig:sphericity:test:162}).
The  5th design has the largest number of acquisitions and it is the most informative of all, however
the Fisher information is not isotropic and Fig.
\ref{fig:sphericity:test:32x3} shows that the empirical distributions of $\tau^{(n)}_2$ and $\tau^{(n)}_5$
do not fit the $\chi^2_5$ distribution, with the consequence of underestimating
the Type I error probability of rejecting an isotropic true tensor. We conclude that the distribution of 
these sphericity statistics is sensitive to anisotropies of the estimation error distribution.
As it was shown in section  \ref{section:sphericity},
these sphericity test statistics should be calibrated against the law of
 $\tau(c +\kappa_1(X),\kappa_2(X),\kappa_3(X) )$, evaluated at $c=\kappa_1(D^{(n)})$,
 where $X$ is the zero mean symmetric Gaussian matrix with covariance \eqref{non:iso}.

We also remark that in lower part of Fig. \ref{fig:sphericity:test:14}-\ref{fig:sphericity:test:10}, compared
with the uniform density, the histogram estimator of the  $\tau_3^{(n)}$ density shows an increasing
linear trend. This linear trend is less evident in \ref{fig:sphericity:test:162}, which is based on a larger number of
acquisitions, and the distribution of the MLE $D^{(n)}$ is presumably 
better approximated by a Gaussian than
in the previous cases.
By taking absolute value  $\vert\tau_3^{(n)}\vert$ the linear trend cancels out,
and the  histogram of $\vert\tau_3^{(n)}\vert$ in the upper part of 
Figures \ref{fig:sphericity:test:14}-\ref{fig:sphericity:test:32x3} fits robustly the
uniform distribution in all the situations we have considered.

\section{CONCLUSION} \label{conclusion}
We have considered the problem of estimating   the spectrum $\bar \gamma_1\ge \bar \gamma_2 \ge \dots \ge \bar \gamma_m$  and the eigenvectors      of a real symmetric  $m\times m$
matrix $\bar D$, possibly non-positive, by the spectrum and the eigenvectors 
of a consistent  and asymptotically Gaussian matrix estimator $D^{(n)}$,
 assuming that  the covariance of  the rescaled limit is isotropic.
 When $\bar D$ has repeated eigenvalues, the delta method does not apply and
the spectrum  of the matrix estimator has  a non-Gaussian limit distribution.
In the limit, 
the random  eigenvalues $ \gamma_1^{(n)}>\gamma_2^{(n)}>\dots >  \gamma_m^{(n)}$ of $D^{(n)}$
form clusters corresponding to the $\bar D$
eigenspaces, with jointly Gaussian barycenters. 
Within each cluster, the differences between eigenvalues and  barycenter
are independent from the barycenter  and the other clusters, 
and follow the conditional law of GOE eigenvalues conditioned on having zero barycenter.

In many applications  it is important to detect the symmetries of the true matrix parameter $\bar D$, in particular to test whether $\bar D$ is  spherical, 
which leads to singular hypothesis testing problems.
A statistical test against $\bar D$-symmetries needs to be calibrated  
taking into account the repulsion between
the random eigenvalues of $D^{(n)}$  corresponding to the same $\bar D$-eigenspace.
In dimension $m=3$, we derived
the    asymptotic joint distribution of some commonly used
sphericity statistics as Fractional  Anisotropy, Relative Anisotropy and Volume Ratio under isotropy assumptions.
We have also discussed the implications of these general results  for the design and analysis of DTI measurements,
and we showed that gradient designs based on spherical $t$-designs have isotropic Fisher information and are asympotically most informative
when the true tensor is 
spherical.
A direct application would be in  denoising the  FA maps derived from  diffusion tensor estimates.
Testing for sphericity at each volume element with a fixed confidence level, corresponds to
a FA cut-off threshold which is not constant over the voxels  but depends locally on the  estimated
noise and mean diffusivity parameters. We have seen 
in the Monte Carlo study  that the simulated sphericity statistics fit well their theoretical limit
distribution when the Fisher  information of the experiment was isotropic.
However, 
there was a significant discrepancy  under  experimental design 5, with non-isotropic Fisher information.
We conclude that these findings  give a  strong theoretical argument in favour of using spherical $t$-designs in DTI, and we plan to conduct 
similar experiments with real DTI data in the near future.
 Finally, our work in progress is to  generalize this theory to situations in which the covariance of the Gaussian 
limit matrix has symmetries without being fully isotropic.

\section*{ACKNOWLEDGMENTS} We thank Konstantin Izyurov, Sangita Kulathinal Antti Kupiainen and Juha Railavo for insightful discussions.

\begin{table}[]
\centering
\caption*{Size of Spherical $t$-Designs }
\begin{tabular}{|l|r|r|r|r|r|r|r|r|r|r|r|r|r|r|r|r|r|r|r|r|r|r|}\hline
$t$ &   4 & 
 5 & 
 6 & 
  7  & 
  8 & 
  9  & 
 10  &  
 11  & 
 12 &  
 13 & 
 14 & 
 15 &
 16 & 17 
 \\  \hline 
$n_a$ & -&12 & - & 32 & -& 48 &- & 70 & -& 94  &- & 120 &- & 156 
\\ \hline$n$ &
   14  &
 18   &
   26  &
   32  &
 42 &   50
  & 62  
  & 72 
  &  86
 & 98
  & 114
  &128
& 146
& 163
\\ \hline
  \end{tabular}
\caption{ Number $n_a$ of points  in  some known antipodal spherical $t$-designs of order $4\le t \le 17$ 
in ${\mathcal S}^2$, computed by
 Rob Womersley, while $n$ is for his non-antipodal spherical $t$-designs. }\label{table:design}
\end{table}

 \begin{figure}[]
\centering
\begin{minipage}[r]{0.51\textwidth}
\centering
\vspace{0pt}
\begin{tabular}{ |c c c |} 
   \hline  
   $u_x$ &$u_y$ & $u_z$  \\ \hline 
0.0000   &  0.0000  &  1.0000 \\
0.9473 &  0.0000  & 0.3202 \\
  -0.9035  &  0.1944 & -0.3821 \\
   0.2693 &  0.7379 &  0.6189 \\
   0.4465 & -0.6627 &  0.6012 \\
  -0.8205 & -0.0749 &  0.5668  \\
   -0.1166 &   0.8072 & -0.5787 \\
     0.6831 &   0.6942 & -0.2269 \\
   0.0897 &  0.0476  & -0.9948 \\
   0.7740 & -0.2872 & -0.5642 \\
   0.2389 & -0.9284 & -0.2846 \\
  -0.5595 & -0.5216 & -0.6441 \\
  -0.5094 & -0.8054 &  0.3029 \\
  -0.5394 &  0.7991 &  0.2655 \\  \hline
\end{tabular}
\end{minipage}%
\hfill
\begin{minipage}[l]{0.49\textwidth}
\centering
\vspace{0pt}
\begin{figure}[H]
\centering
\includegraphics[width=\textwidth]{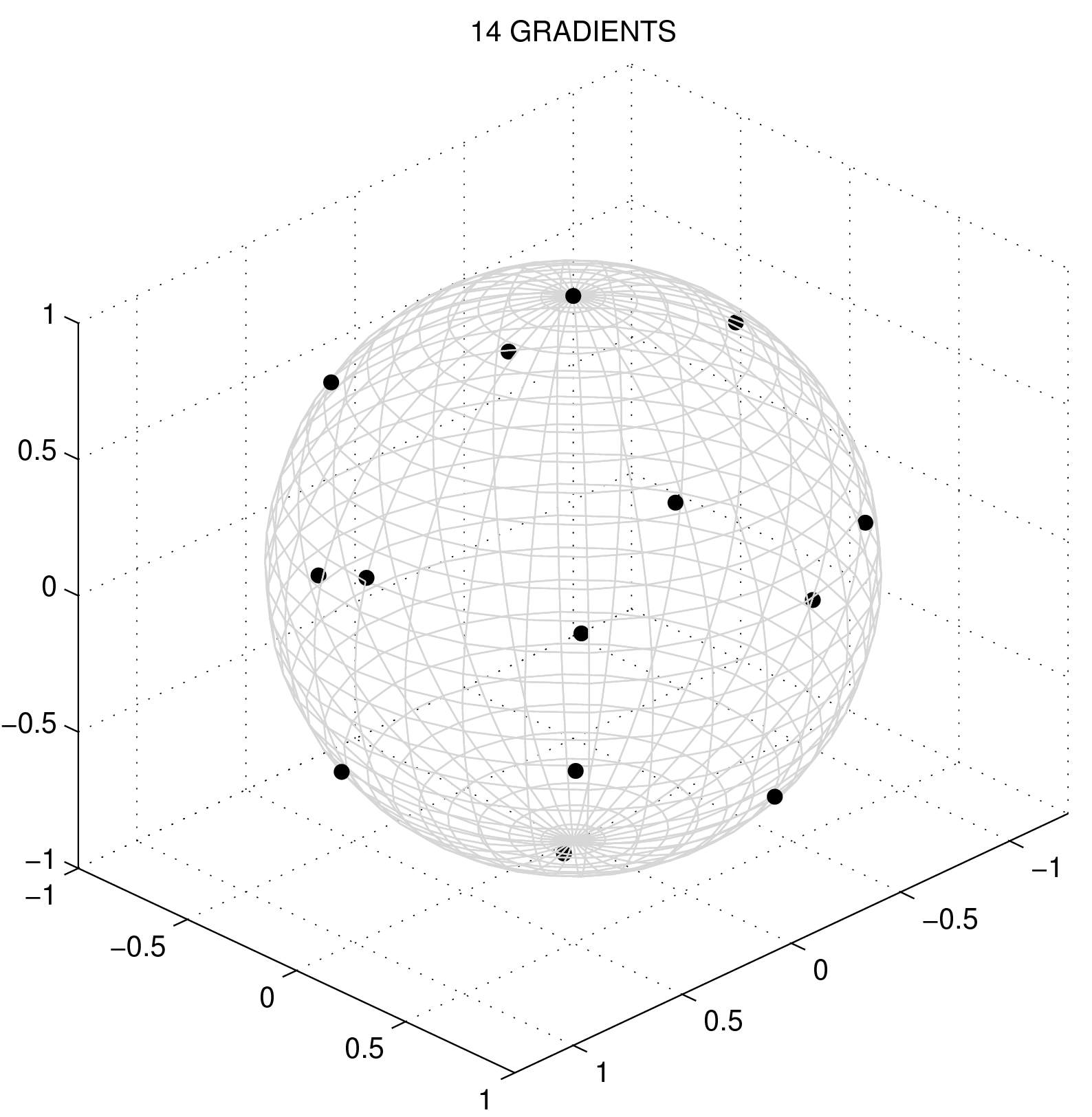}
\end{figure}
 \end{minipage}
 \caption{A non-antipodal spherical $t$-design of order 4, with 14 gradients, by Rob Womersley }
 \label{fig:gradients:4th_order}
\end{figure}

\begin{figure}[]
\centering
\begin{minipage}[r]{0.51\textwidth}
\centering
\vspace{0pt}
\begin{tabular}{ |c c c| } 
   \hline  $u_x$ &$u_y$ & $u_z$  \\ \hline 
   0.0000  &  0.0000  &  1.0000 \\
   0.8944   & 0.0000  &      0.4472 \\   
   0.2764   &-0.8507  &  0.4472 \\ 
  -0.7236&   0.5257 &  0.4472 \\
  -0.7236&  -0.5257 &  0.4472 \\
   0.2764&   0.8507 &  0.4472  \\
\hline
\end{tabular}
\end{minipage}%
\hfill
\begin{minipage}[l]{0.49\textwidth}
\centering
\vspace{0pt}
\begin{figure}[H]
\centering
\includegraphics[width=\textwidth]{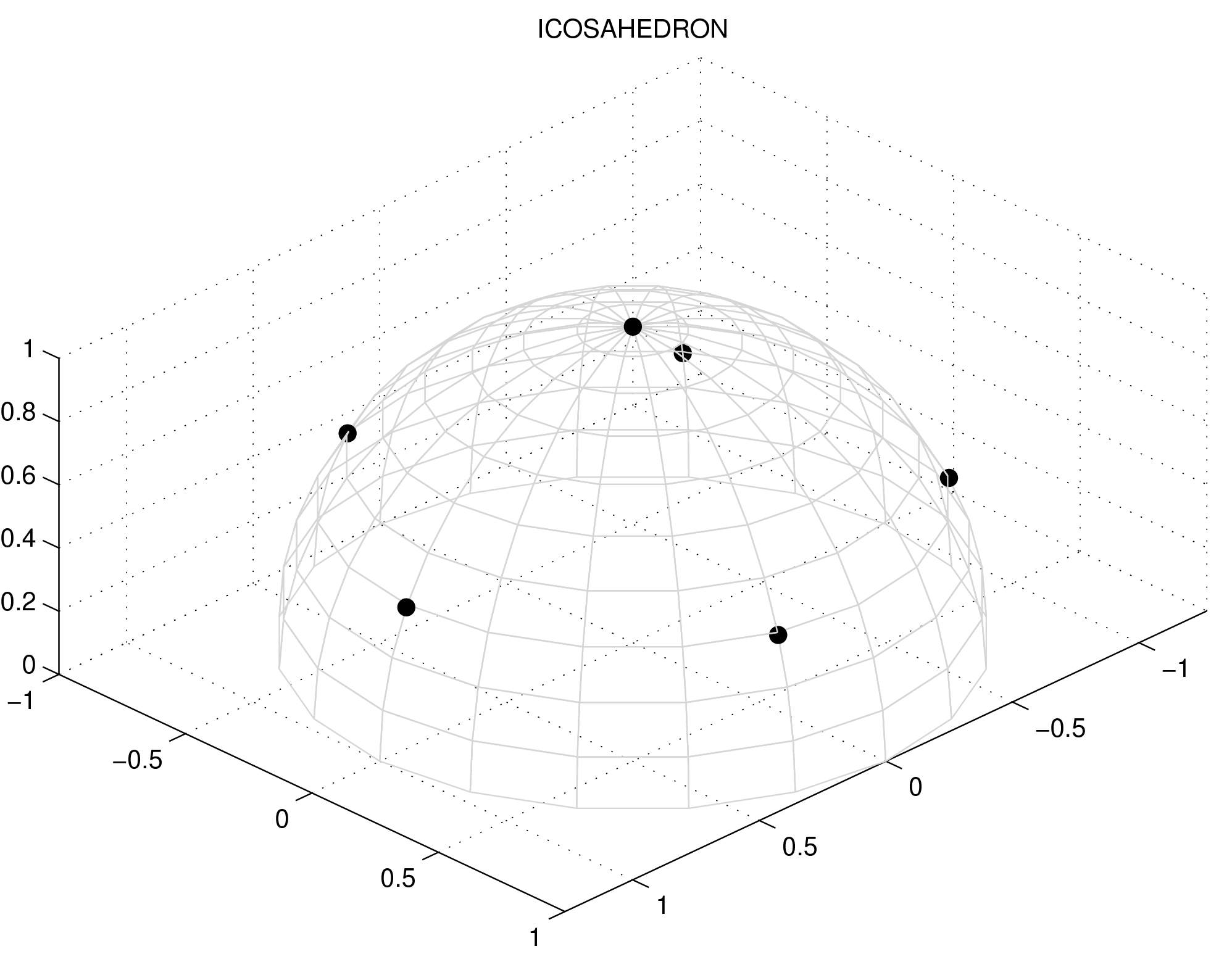}
\end{figure}
 \end{minipage}
 
\caption{Gradient design based on the icosahedron with 6 gradients on the northern hemisphere. }
 \label{fig:gradients:icosahedron}
 \end{figure}

\begin{figure}[]
\centering
\begin{minipage}[r]{0.51\textwidth}
\centering
\vspace{0pt}
\begin{tabular}{ |c c c | } 
   \hline  $u_x$ &$u_y$ & $u_z$  \\ \hline 
    -0.9342  & 0.3568  & 0.0000   \\
  -0.5774 & -0.5774&  0.5774  \\
  -0.5774 &  0.5774 &  0.5774  \\
  -0.3568 &  0.0000 &  0.9342  \\
   0.0000 & -0.9342  & 0.3568 \\
   0.0000 &  0.9342  & 0.3568  \\
   0.3568 &  0.0000  & 0.9342   \\
   0.5774 & -0.5774  & 0.5774   \\
   0.5774 &  0.5774  & 0.5774  \\
   0.9342 &  0.3568  & 0.0000  \\
%
\hline
\end{tabular}
\end{minipage}%
\hfill
\begin{minipage}[l]{0.49\textwidth}
\centering
\vspace{0pt}
\begin{figure}[H]
\centering
\includegraphics[width=\textwidth]{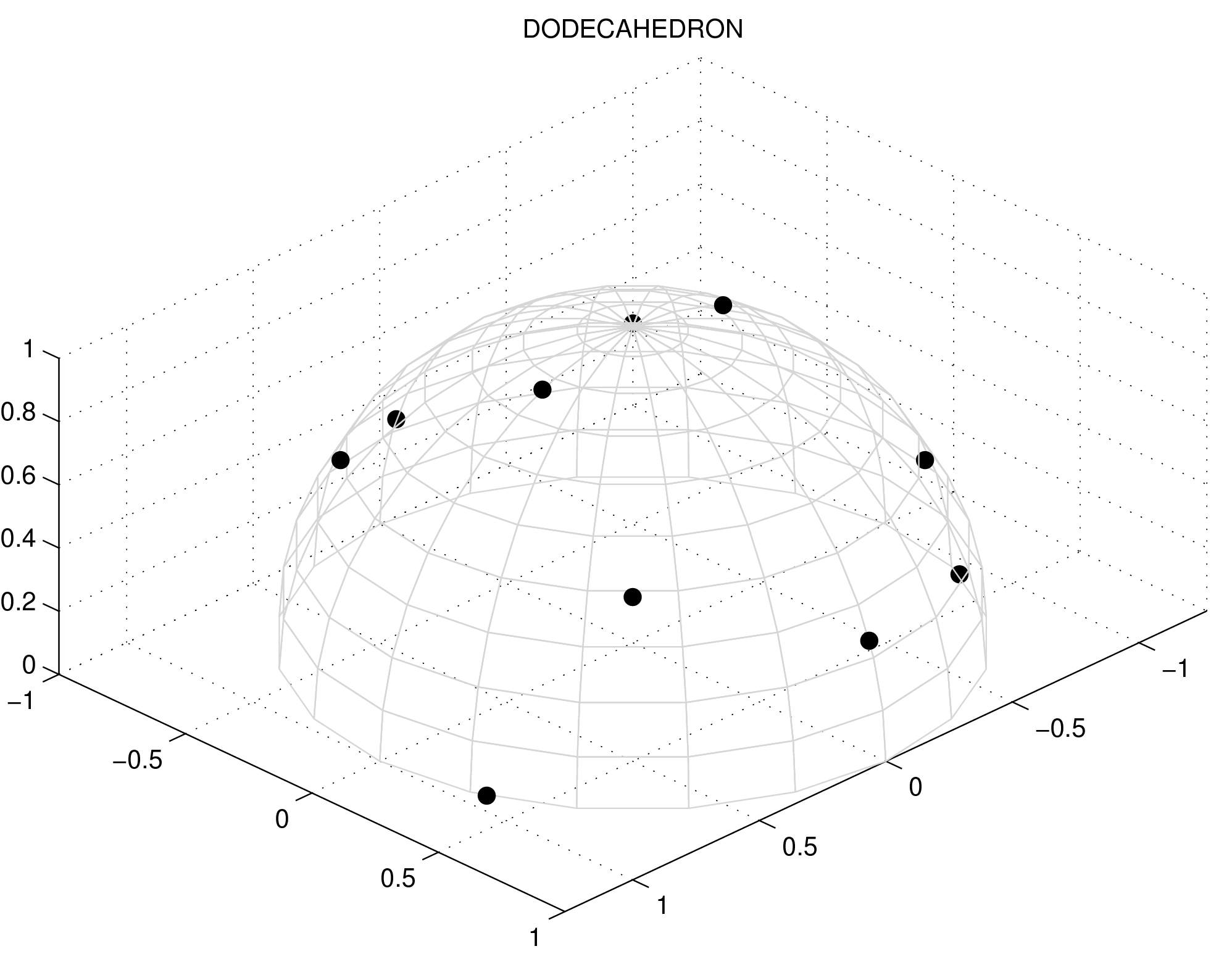}
\end{figure}
 \end{minipage}
\caption{Gradient design based on the dodecahedron with 10  gradients.
}
 \label{fig:gradients:dodecahedron}
 \end{figure}


\begin{figure}[]
\centering
\includegraphics[width=2.5in]{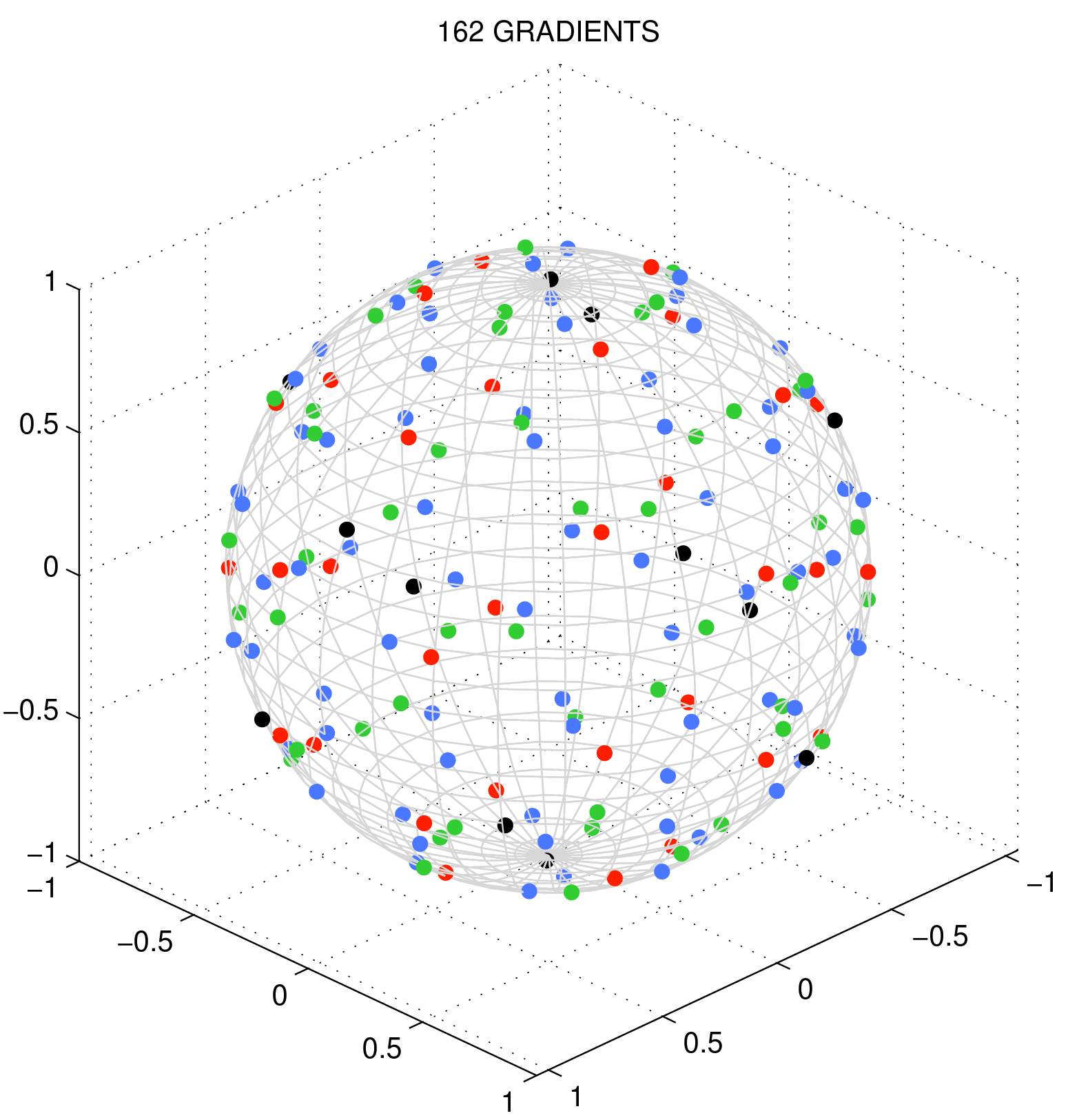}
\caption{Gradient sequence based on combined antipodal spherical $t$-designs of orders 5 (black),7 (red),9 (green)  and 11 (blue),
of respective sizes 12,32,48 and 70. The spherical $t$-designs on different shells were
 rotated in order to maximize the minimal geodesic distance \eqref{geodesic:dist:critera} between gradients. 
  }
 \label{fig:gradients}
\end{figure}

\begin{figure}[!tbp] 
~
\begin{subfigure}[b]{0.4\textwidth}
\centering
\includegraphics[width=2.0in]{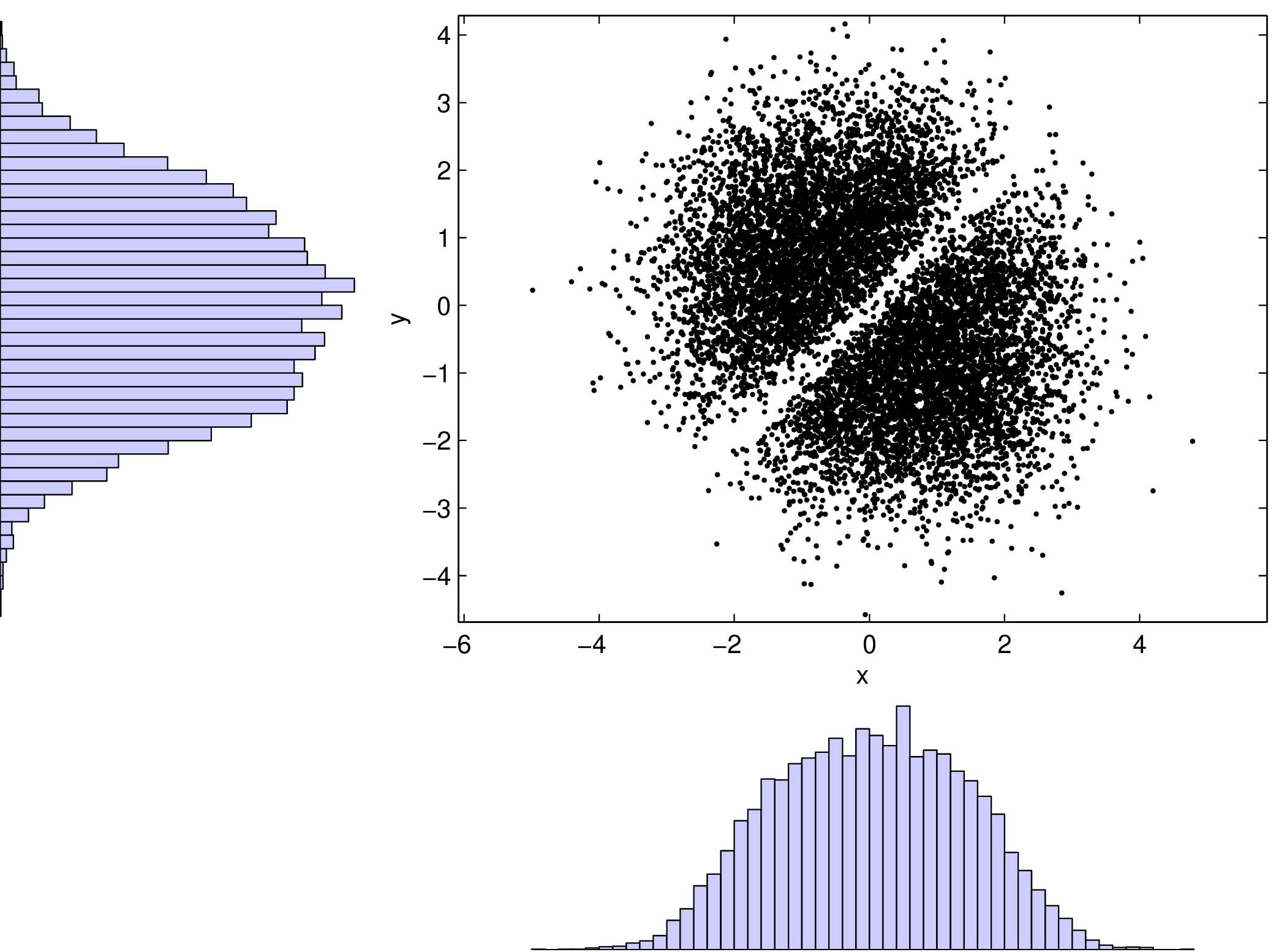}
\caption{ \footnotesize  $\{ \gamma_i,\gamma_j\},1\le i \ne j\le 3$, 
$\bar\gamma_i=0$.}\label{fig1:GOE3}
\end{subfigure}
~
\begin{subfigure}[b]{0.4\textwidth}
\centering
\includegraphics[width=2.0in]{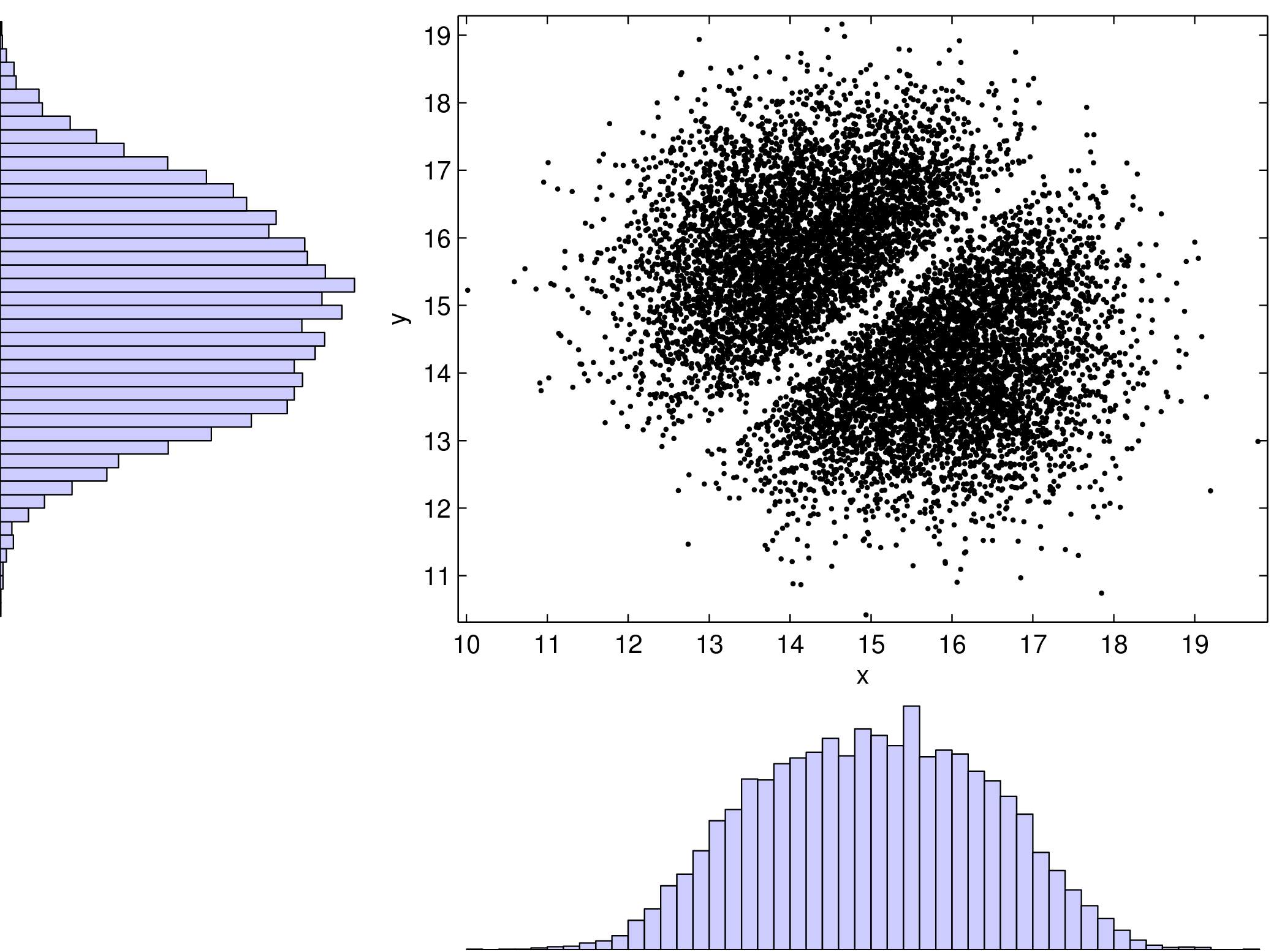}
\caption{
 \footnotesize
$\{ \gamma_i,\gamma_j\},1\le i \ne j\le 3$, 
$\bar\gamma_1=\bar\gamma_2=\bar\gamma_3=15$.
}\label{fig1:positive:totally:symmetric}
\end{subfigure}
~

\begin{subfigure}[b]{0.4\textwidth}
\centering
\includegraphics[width=2.0in]{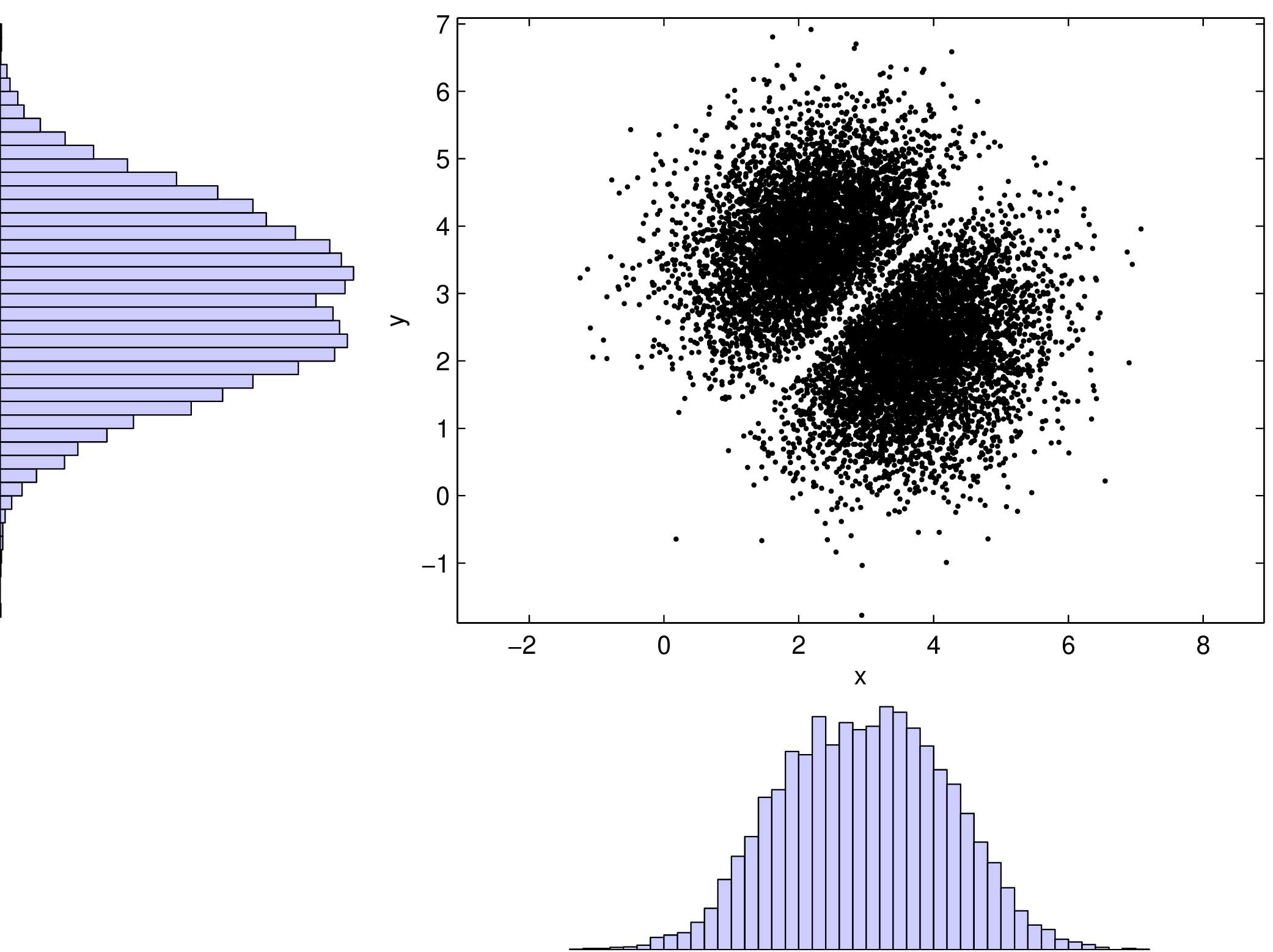}
\caption{
\footnotesize
$\{ \gamma_2,\gamma_3\}$,
 $\bar\gamma_1=15 > \bar \gamma_2 = \bar\gamma_3 =3$ 
.}\label{fig1:positive:prolate}
\end{subfigure}
~
\begin{subfigure}[b]{0.4\textwidth}
\centering
\includegraphics[width=2.0in]{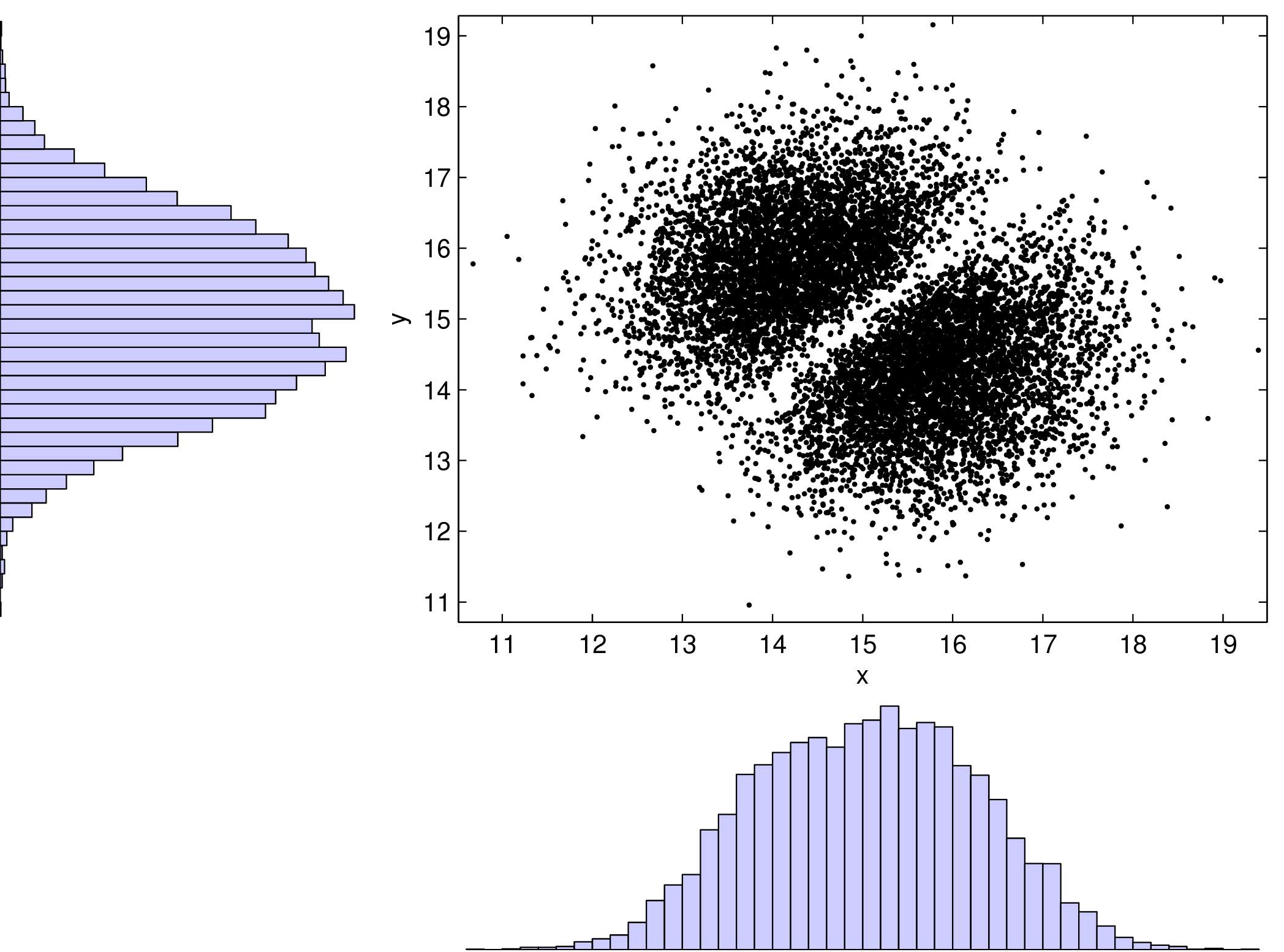}
\caption{
\footnotesize
$\{ \gamma_1,\gamma_2\}$, 
$\bar\gamma_1=15=\bar \gamma_2 =15 > \bar\gamma_3 =3$
.}\label{fig1:positive:oblate}
\end{subfigure}

~
\begin{subfigure}[b]{0.4\textwidth}
\centering
\includegraphics[width=2.0in]{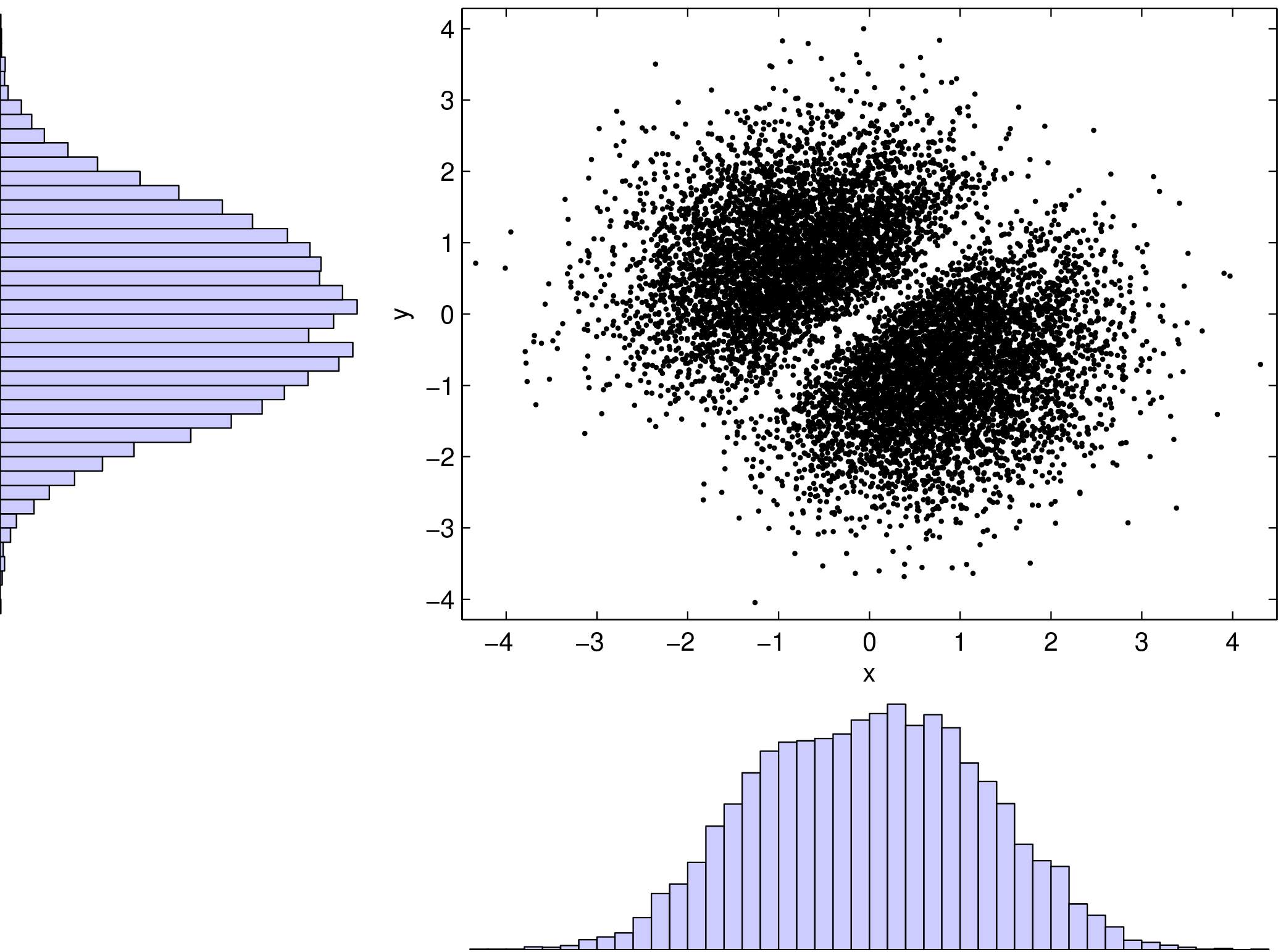}
\caption{  \footnotesize $2\times 2$-GOE eigenvalues}
\end{subfigure}
~
\begin{subfigure}[b]{0.4\textwidth}
\centering
\includegraphics[width=2.0in]{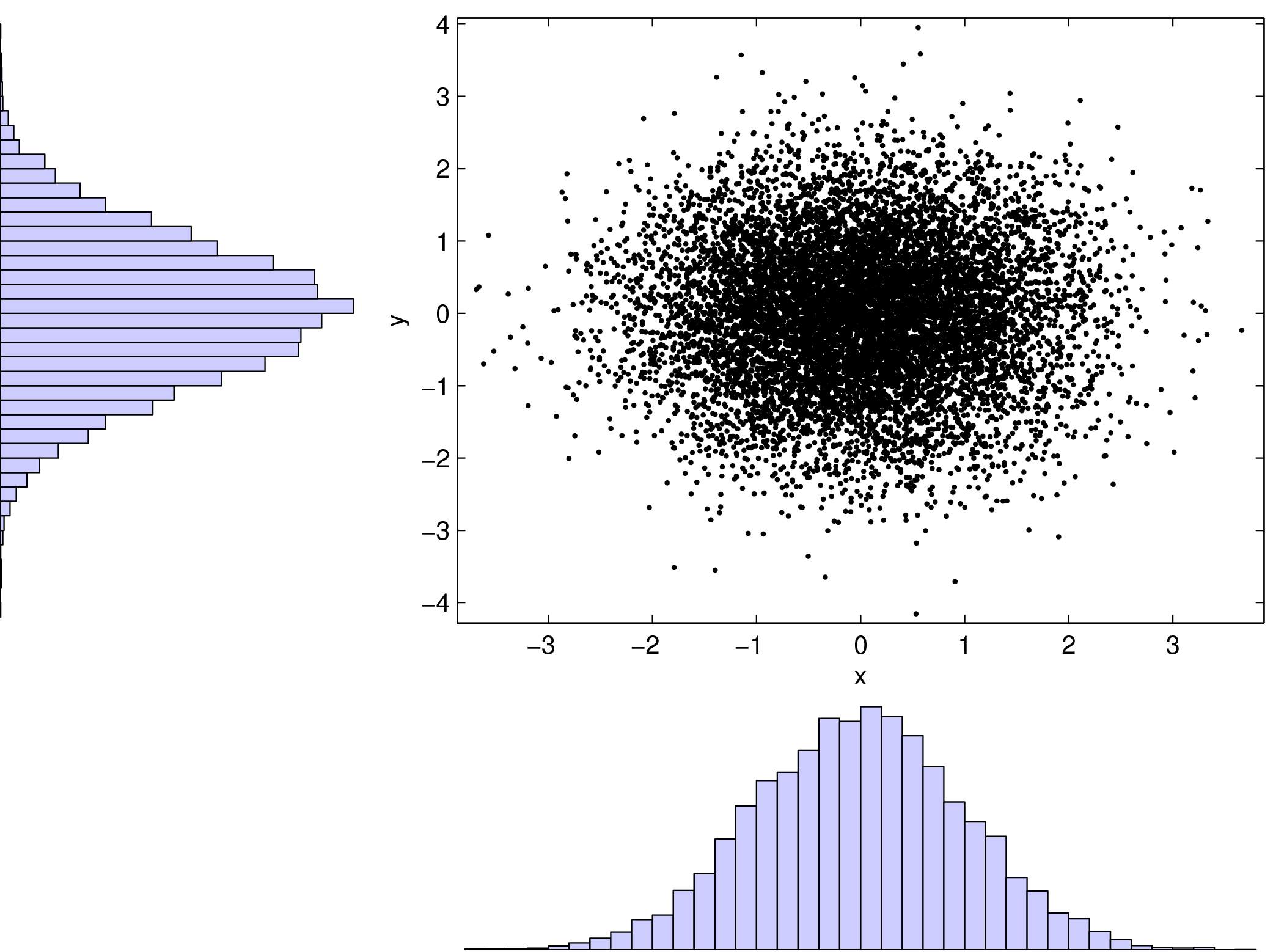}
\caption{  \footnotesize i.i.d. standard Gaussian pairs }
\end{subfigure}
\caption{ 
$10000$  pairs of distinct eigenvalues of i.i.d. symmetric random matrices
with isotropic Gaussian noise ($\mu=1/2$, $\lambda=0$) with various  mean:
zero, corresponding to the $3\times 3$-GOE (a),
 isotropic (b), 
 prolate (c), 
 oblate (d). 
 For comparison we show  i.i.d. $2\times 2$-GOE eigenvalue pairs (e), and  i.i.d. standard Gaussian pairs (f). 
 Within each pair the ordering is randomized, to emphasize
 the repulsion effect around the diagonal.
}\label{fig1} 
\end{figure} 
\begin{figure}[!t]
\centering
\includegraphics[width=2.5in]{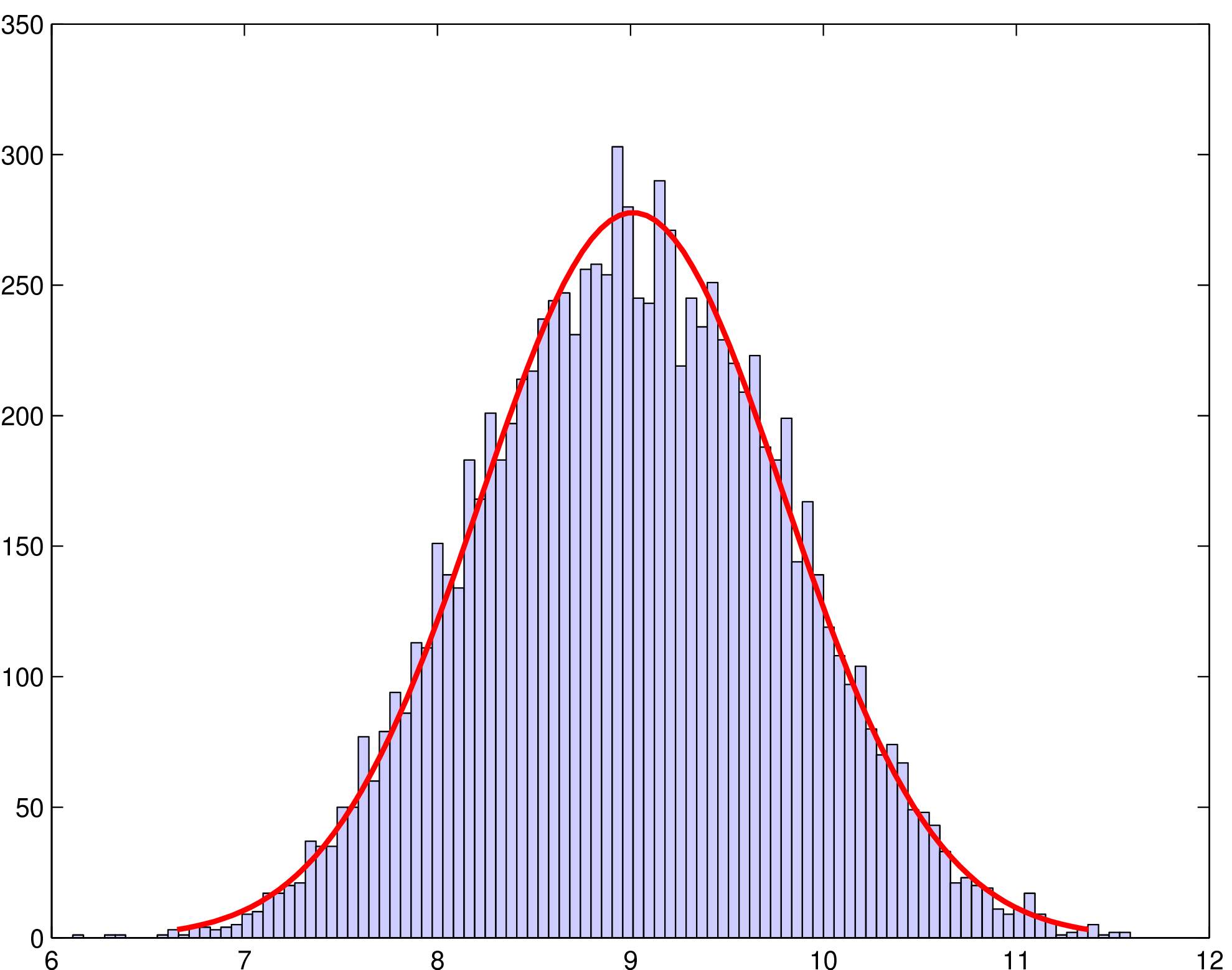}
\caption{Histogram and fitted Gaussian curve from 10000 i.i.d. realizations of the cluster barycenter 
$(\gamma_2+\gamma_3)/2$, in the  prolate  mean tensor case.}\label{fig3}
\end{figure}
\begin{figure}[!tbp]
 \includegraphics[width=1.35\textwidth,height=5.3in,left]{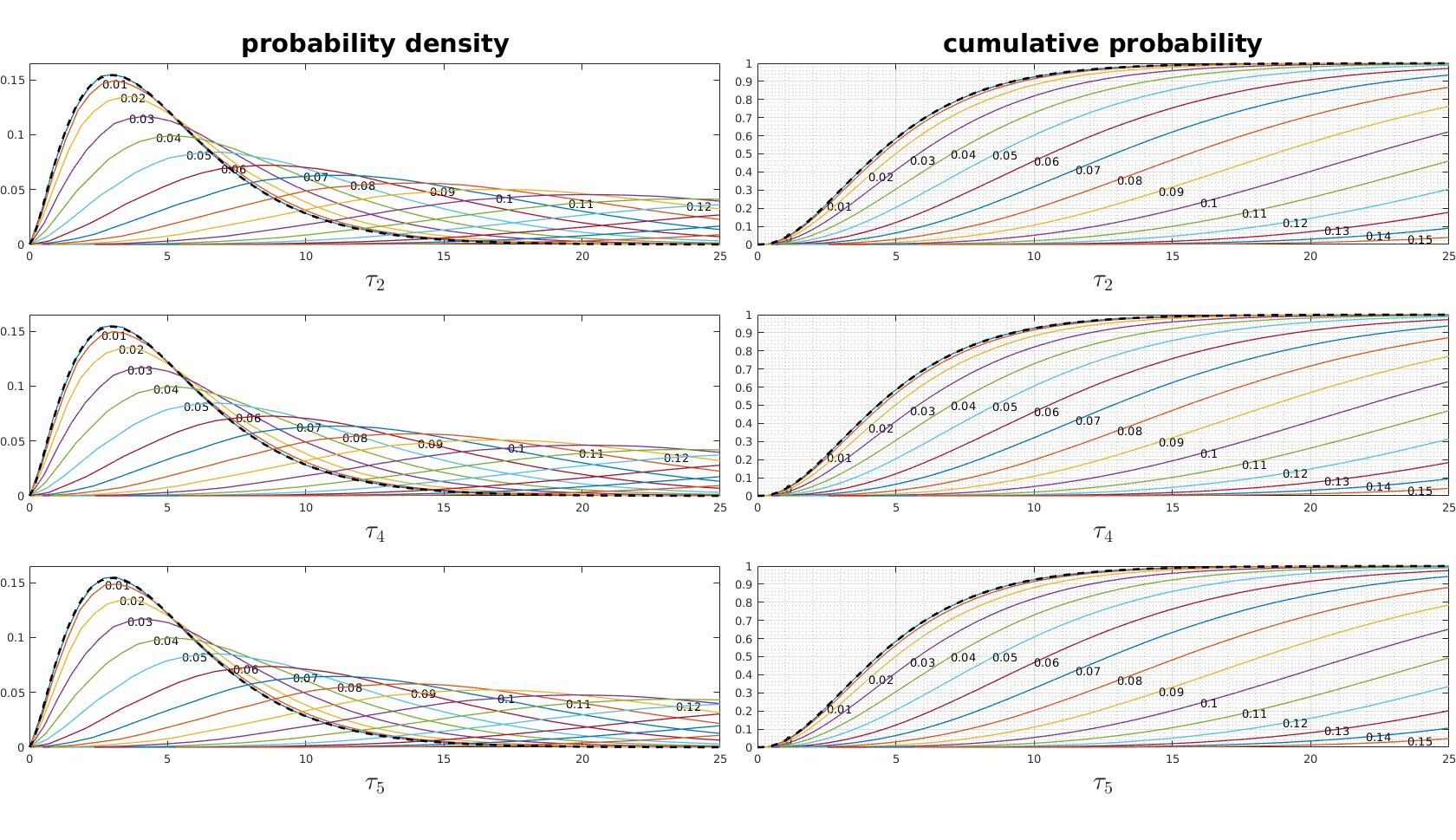}
 \caption{ Probability densities (left) and cumulative probabilities (right) of the sphericity test statistics $\tau_2(D),\tau_4(D),\tau_5(D)$,
 where the $3\times 3$ symmetric random matrix $D$ is  Gaussian with isotropic precision 
 $A(2,2)$,  and  there are $16$ alternative mean tensors $\bar D$, 
  with fixed mean diffusivity $\kappa_1(\bar D)=15$.
Under the null hypothesis $\bar D$ is spherical, while the alternatives correspond to prolate mean tensors with FA in $(0.0,0.15]$.
For each test statistics, the probability density and cumulative probability curves are 
labeled by the FA values of the corresponding mean tensors.
 The broken curves display the  $\chi^2_5$ limit distribution 
 under the null hypothesis. 
  }
 \label{fig:power:comparison}
\end{figure}
\begin{figure}[!tbp]
\begin{subfigure}[b]{0.4\textwidth}
\includegraphics[width=2in]{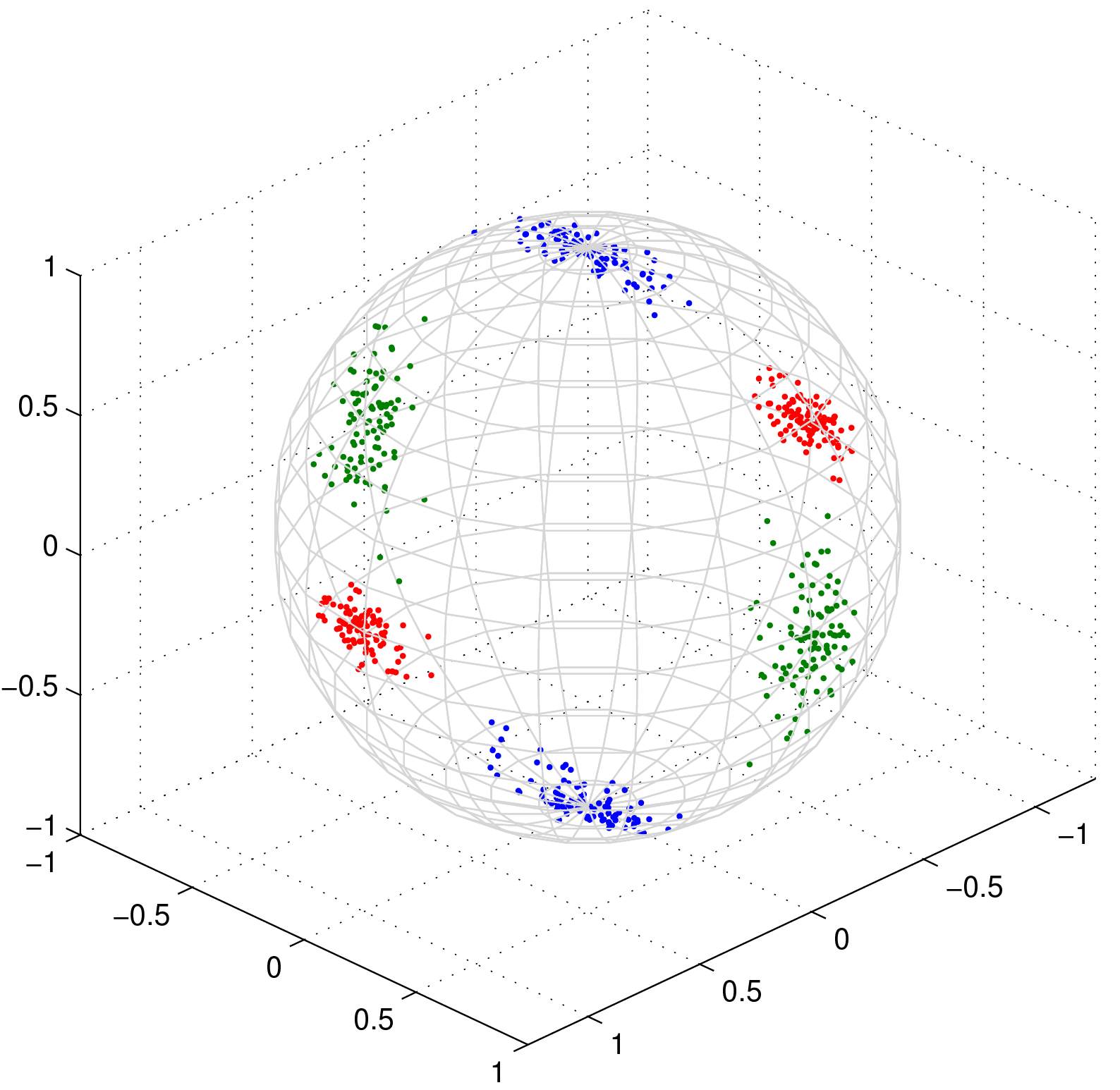}
\caption{  
$\bar\gamma_1=15,\bar\gamma_2=7.5, \bar\gamma_3=3$. }\label{fig:eigenvector:asymmetric}
\end{subfigure}
~
\begin{subfigure}[b]{0.4\textwidth}
\includegraphics[width=2in]{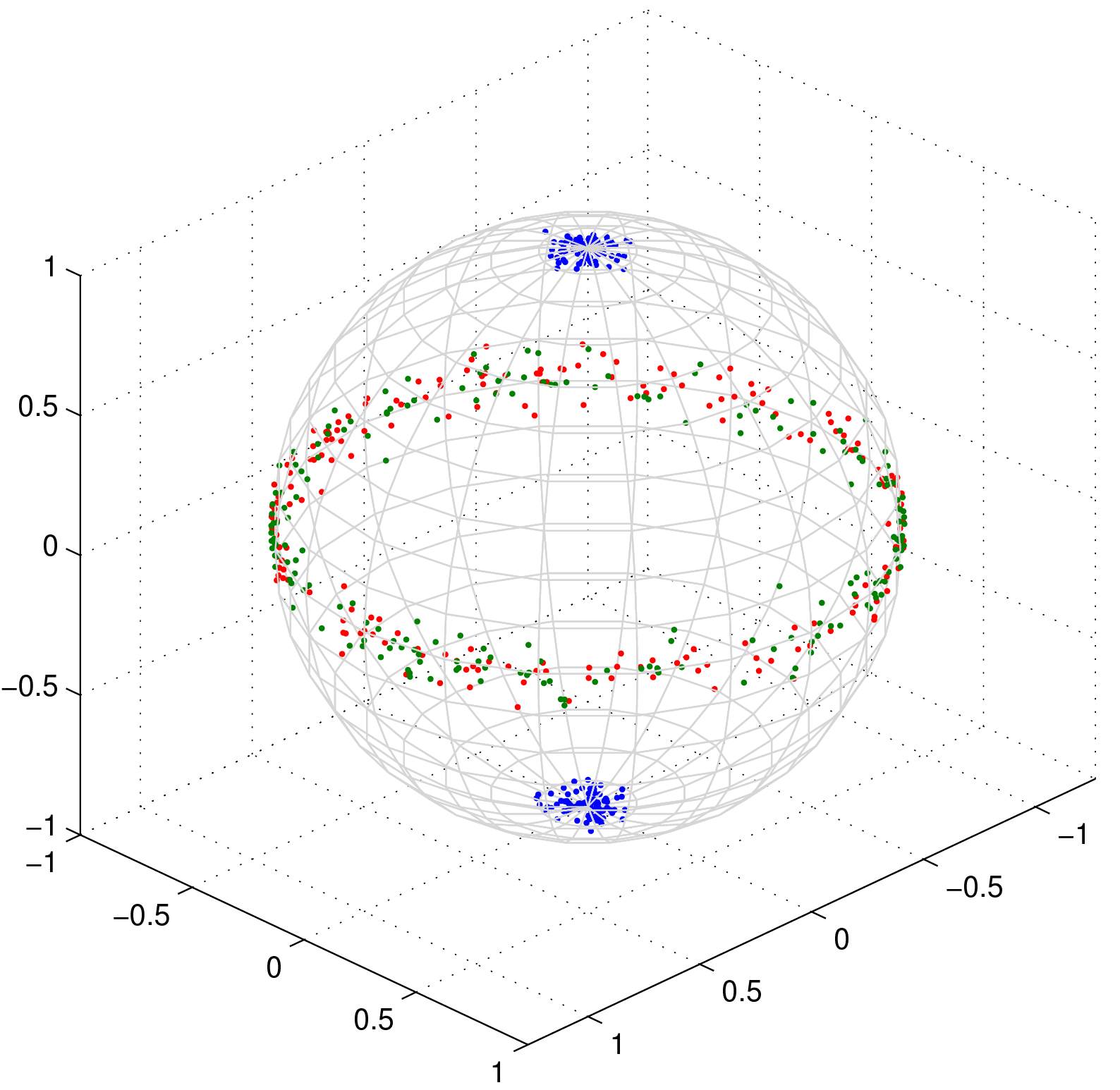}
\caption{ 
$\bar\gamma_1=\bar\gamma_2=15, \bar\gamma_3=3$. }\label{fig:eigenvector:oblate}
\end{subfigure}
\label{fig4}
\caption{200 i.i.d. orthonormal eigenvector 
triples from the Gaussian model with isotropic noise parameters $\mu=1/2,\lambda=0$,
with totally asymmetric (left) and  oblate (right) diagonal mean tensor,
using a similar graphical  construction as the one introduced in \cite{basser:artefact}. 
}
\end{figure}
%
%

\begin{figure}[!tbp]
\caption*{Design 1, sphericity   statistics. }
~
\begin{subfigure}[b]{0.9\textwidth}
\centering
\includegraphics[width=4.5in, height=2.5in]{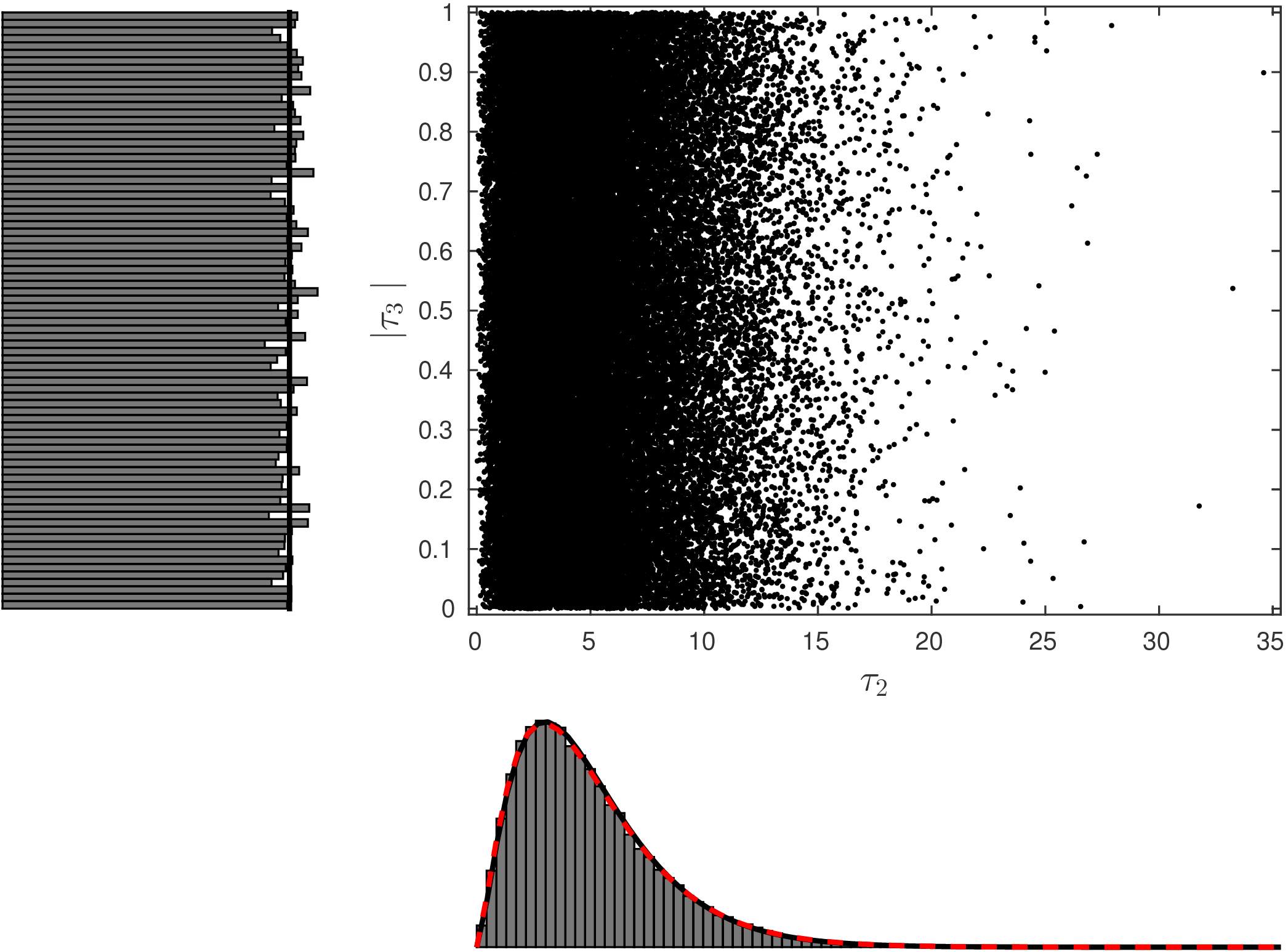}
\caption{}\label{fig:sphericity:test:14a}
\end{subfigure}
~
 
 \begin{subfigure}[b]{0.9\textwidth}
\centering
\includegraphics[width=4.5in, height=2.5in]{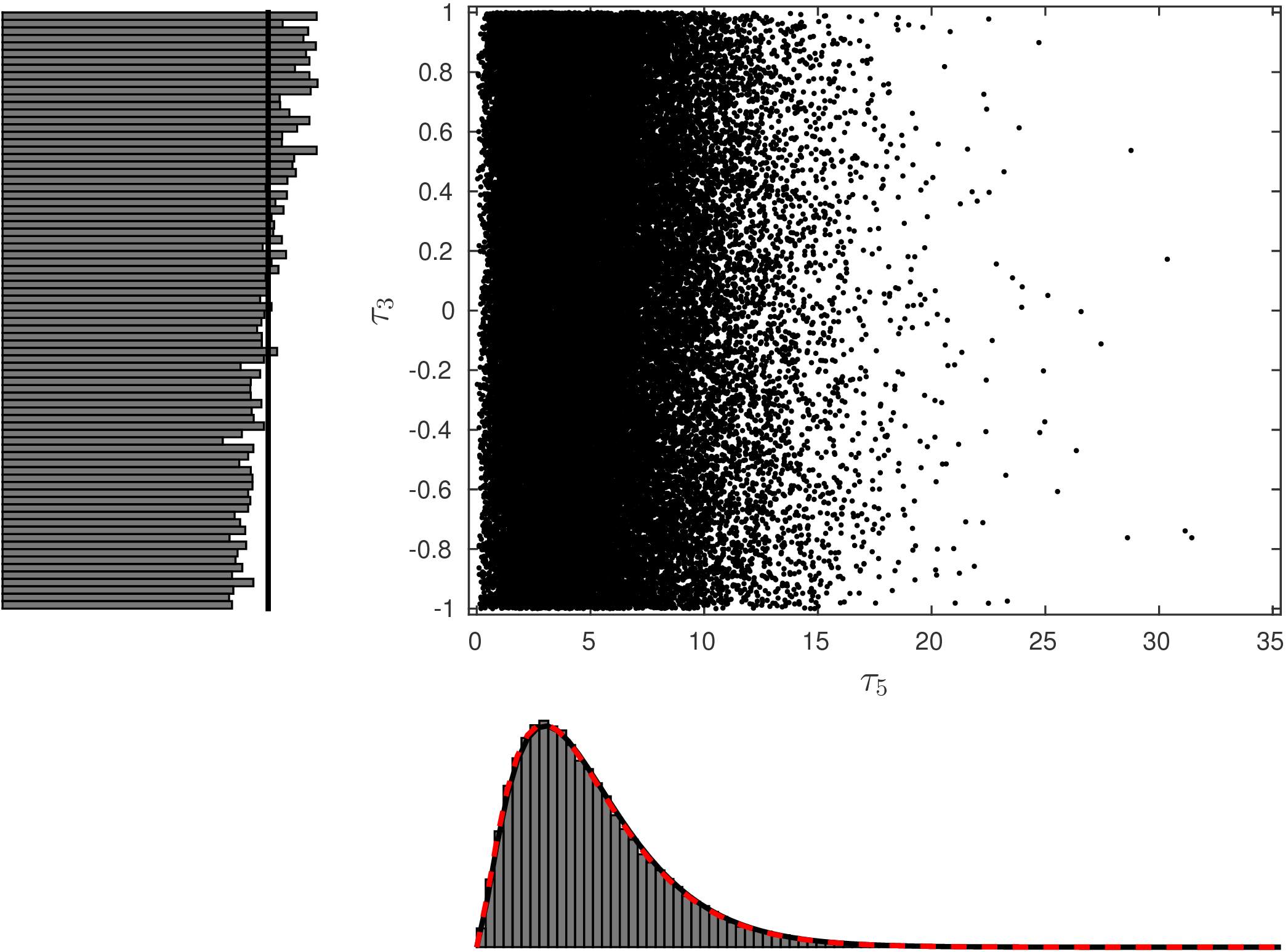}
\caption{ }\label{fig:sphericity:test:14b}
\end{subfigure}

\caption{Scatterplot of the eigenvalue statistics    $(\tau_2^{(n)},\vert\tau_3^{(n)}\vert)  $
in (a) and $(\tau_5^{(n)},\tau_3^{(n)})  $ in (b), from
 a Monte Carlo study
 based on $N=50000$  replications of  a dataset generated under Design 1, where the true tensor and the  Fisher information
 are isotropic. The histogram density estimators are compared with  theoretical limit  densities (black continuous curves), which are
 uniform on the vertical axes and  $\chi^2_5$ on the horizontal axes. The best fitting gamma densities
 (red broken curves) are also shown,  with shape parameter $2.4238$ and scale parameter $2.0627$
 in (a) and 
 with shape parameter  $2.4566$ and scale parameter $2.0137$ in (b).
}\label{fig:sphericity:test:14}
\end{figure}
   \begin{figure}[!tbp]
\caption*{ Design 2, sphericity   statistics. }
~
\begin{subfigure}[b]{0.9\textwidth}
\centering
\includegraphics[width=4.5in, height=2.5in]{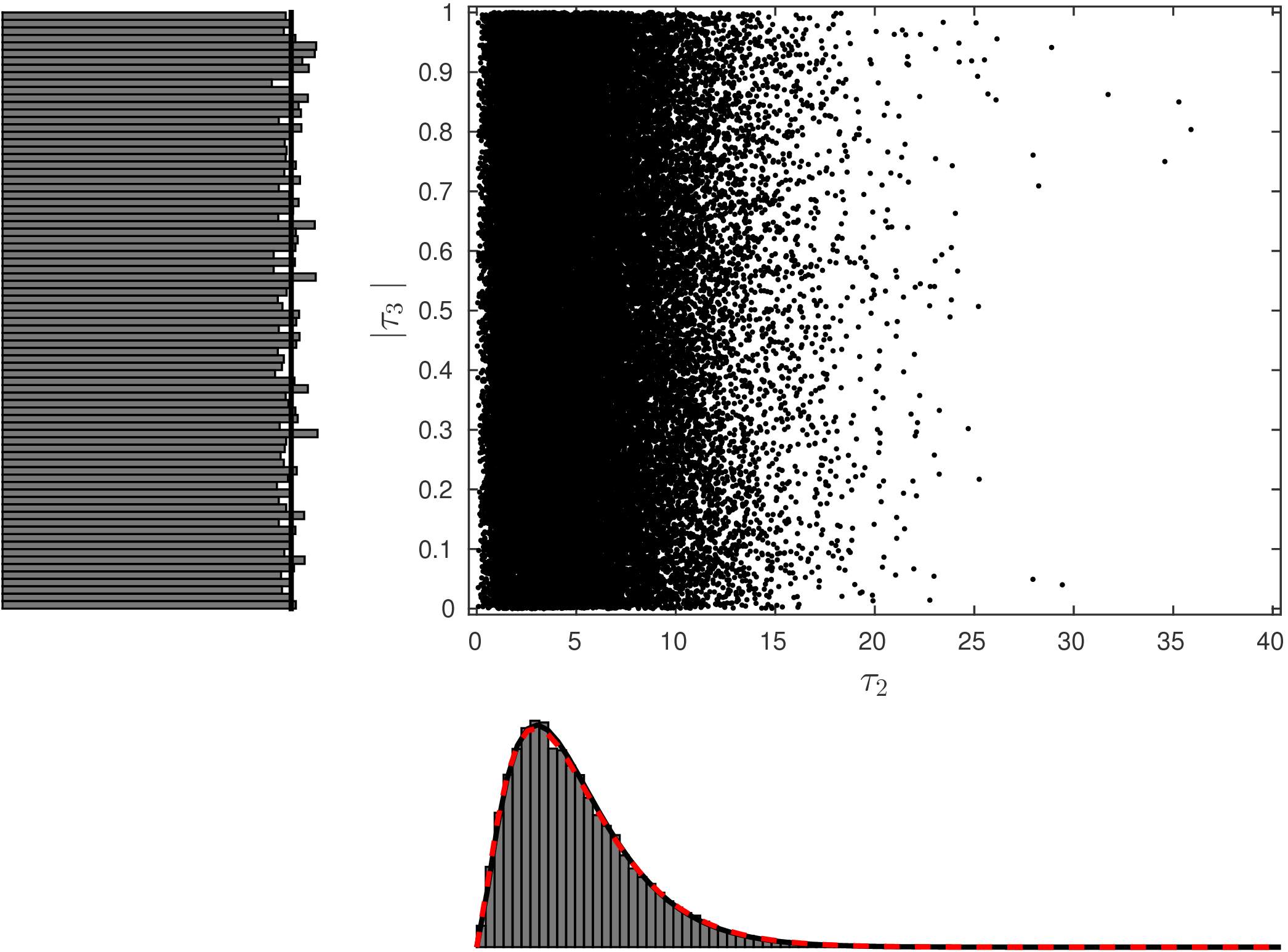}
\caption{}\label{fig:sphericity:test:6a}
\end{subfigure}
~
 
 \begin{subfigure}[b]{0.9\textwidth}
\centering
\includegraphics[width=4.5in, height=2.5in]{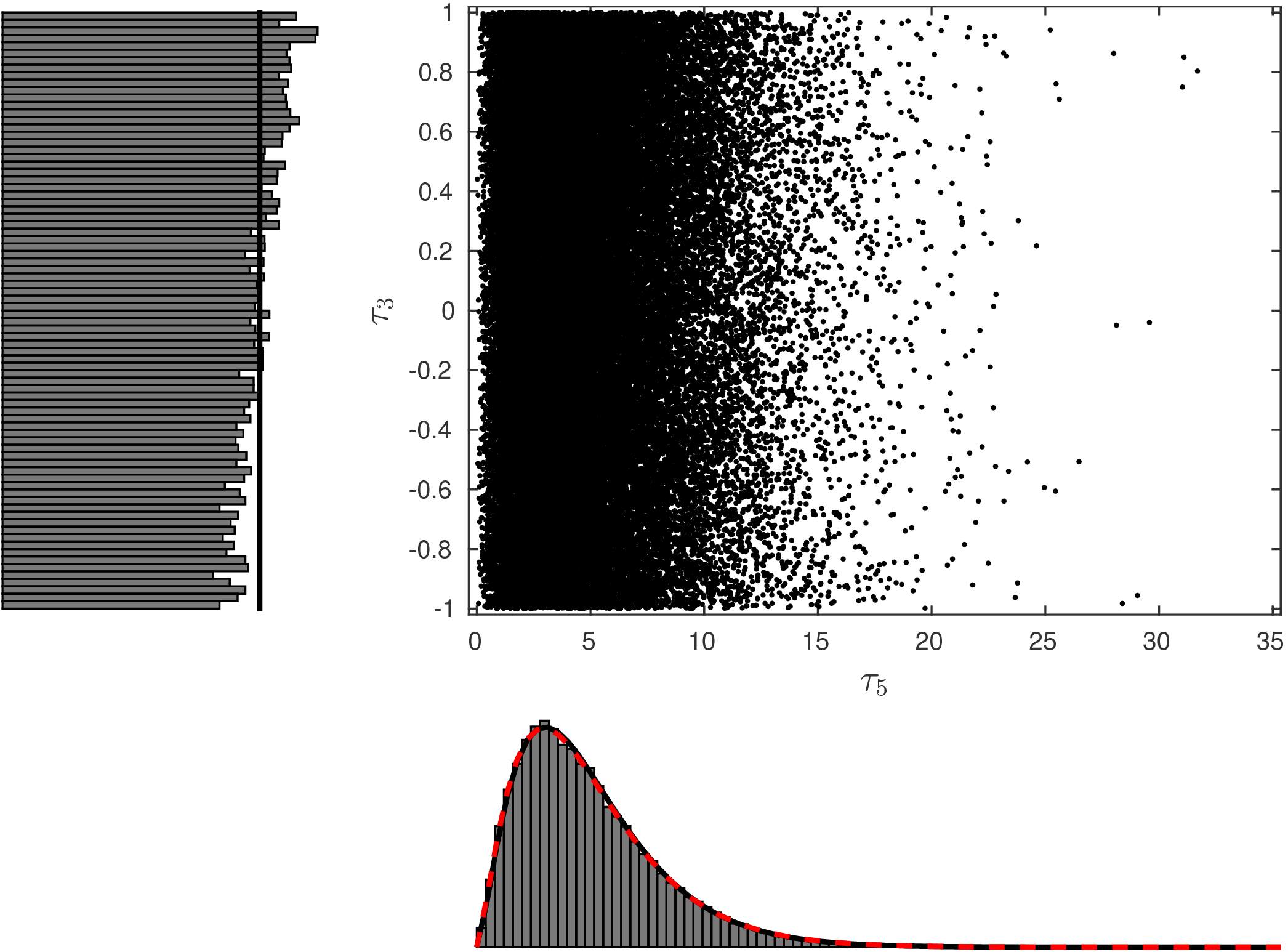}
\caption{ }\label{fig:sphericity:test:6b}
\end{subfigure}
\caption{Scatterplot of the eigenvalue statistics    $(\tau_2^{(n)},\vert\tau_3^{(n)}\vert)  $
in (a) and $(\tau_5^{(n)},\tau_3^{(n)})  $ in (b), from
 a Monte Carlo study
 based on $N=50000$  replications of  a dataset generated under Design 2, where the true tensor and the Fisher information are isotropic.
 The histogram density estimators are compared with  theoretical limit  densities (black continuous curves), which are
 uniform on the vertical axes and  $\chi^2_5$ on the horizontal axes. The best fitting gamma densities
 (red broken curves) are also shown,  with shape parameter
  $2.4103$ and scale parameter $2.0842$
 in (a) and 
 with shape parameter $2.4315$ and scale parameter $2.0542 $  in (b).
}\label{fig:sphericity:test:6}
\end{figure}

\begin{figure}[!tbp]
\caption*{  Design 3,  sphericity  statistics.  }
~
\begin{subfigure}[b]{0.9\textwidth}
\centering
\includegraphics[width=4.5in, height=2.5in]{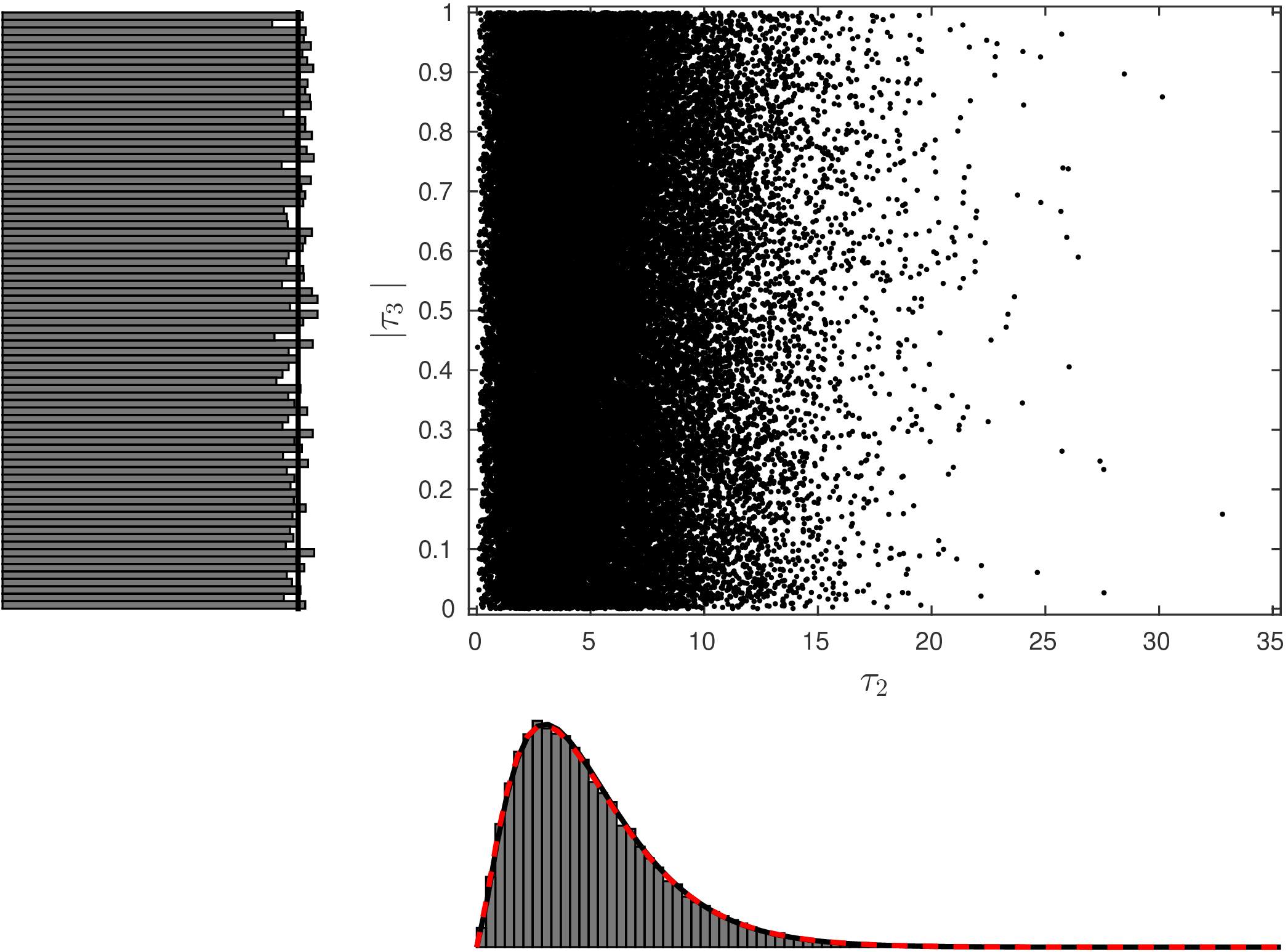}
\caption{}\label{fig:sphericity:test:10a}
\end{subfigure}
~
 
 \begin{subfigure}[b]{0.9\textwidth}
\centering
\includegraphics[width=4.5in, height=2.5in]{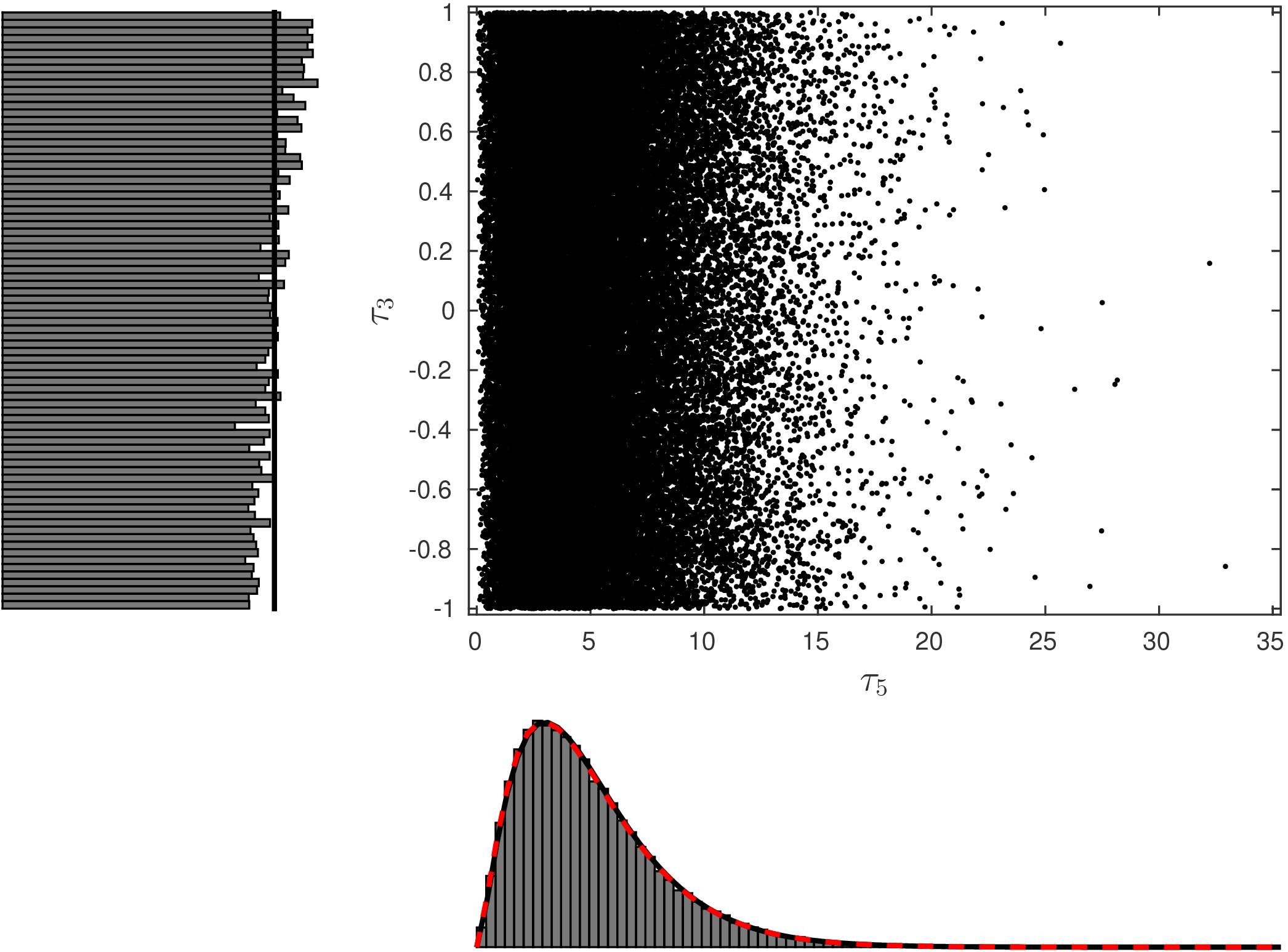}
\caption{ }\label{fig:sphericity:test:10b}
\end{subfigure}
\caption{Scatterplot of the eigenvalue statistics    $(\tau_2^{(n)},\vert\tau_3^{(n)}\vert)  $
in (a) and $(\tau_5^{(n)},\tau_3^{(n)})  $ in (b), from
 a Monte Carlo study
 based on $N=50000$  replications of  a dataset generated under Design 3, where the true tensor and the
 Fisher information are isotropic.
 The histogram density estimators are compared with  theoretical limit  densities (black continuous curves), which are
 uniform on the vertical axes and  $\chi^2_5$ on the horizontal axes. The best fitting gamma densities
 (red broken curves) are also shown,  with shape parameter $2.4405$ and scale parameter $2.0467$
 in (a) and 
 with shape parameter $2.4526$ and scale parameter $2.0298 $ in (b).
}\label{fig:sphericity:test:10}
\end{figure}
   
\begin{figure}[!tbp]
\caption*{Design 4, sphericity   statistics.  }
~
\begin{subfigure}[b]{0.9\textwidth}
\centering
\includegraphics[width=4.5in, height=2.5in]{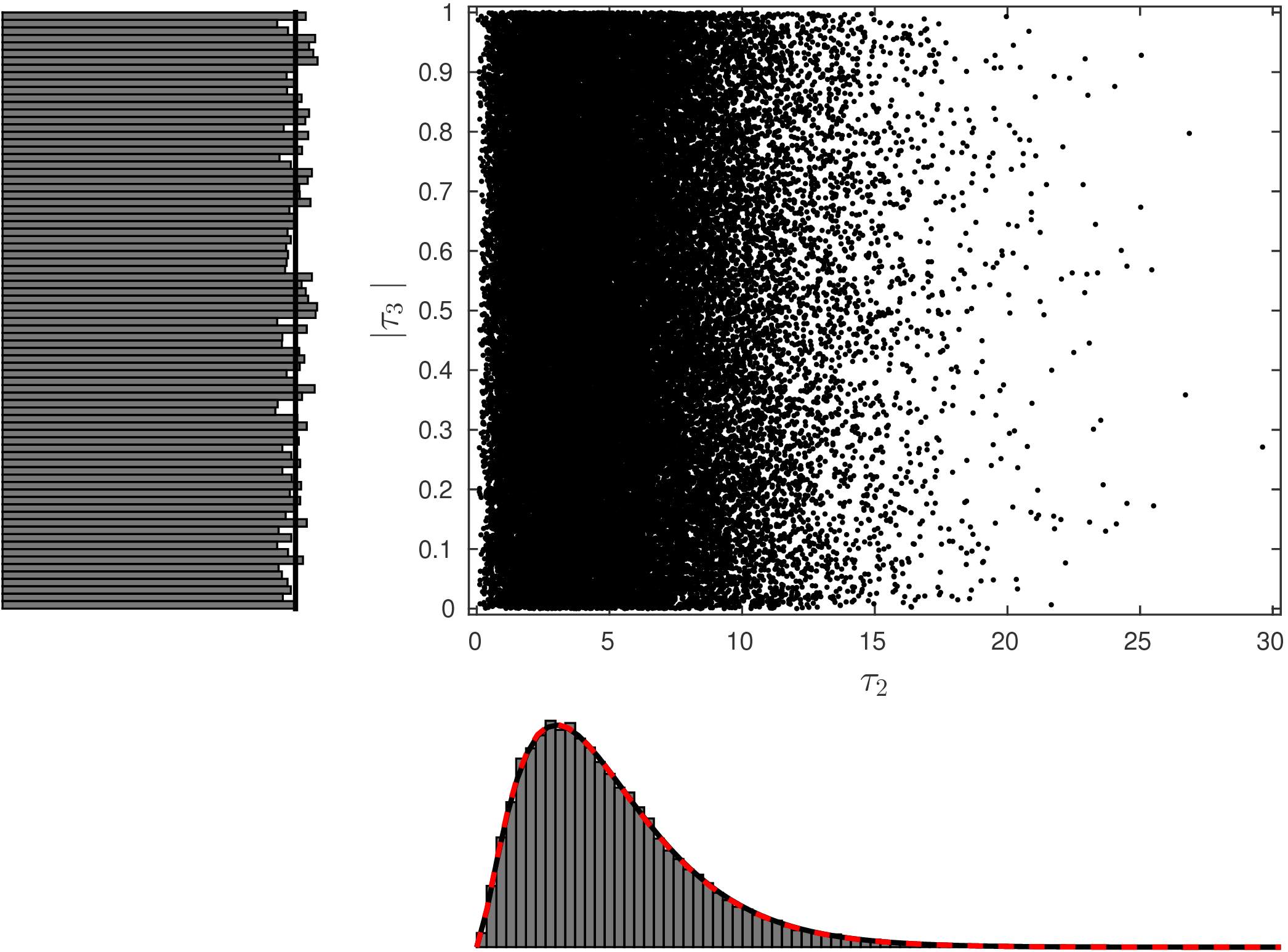}
\caption{}\label{fig:sphericity:test:162a}
\end{subfigure}
~
 
 \begin{subfigure}[b]{0.9\textwidth}
\centering
\includegraphics[width=4.5in, height=2.5in]{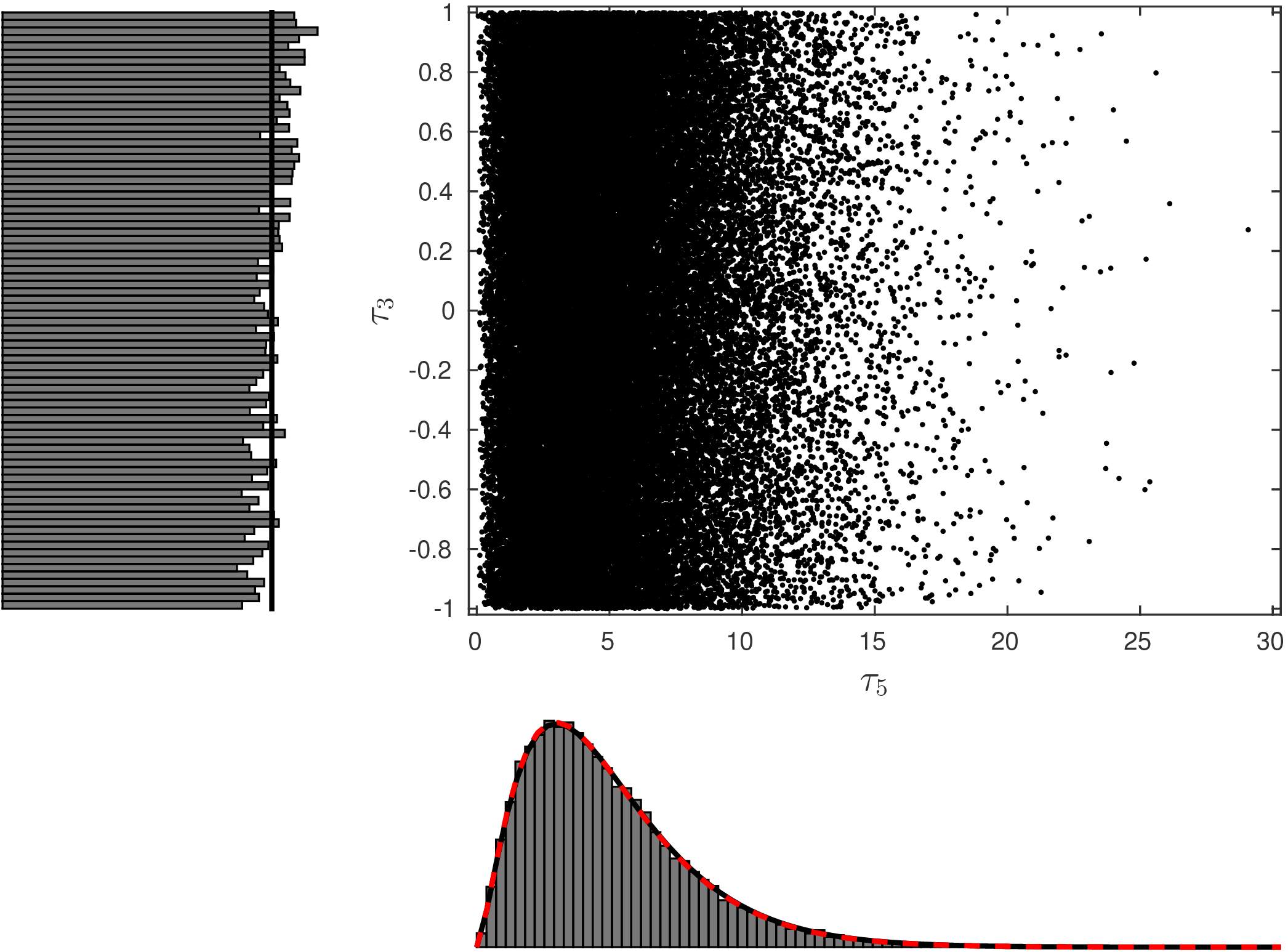}
\caption{ }\label{fig:sphericity:test:162b}
\end{subfigure}
\caption{Scatterplot of the eigenvalue statistics    $(\tau_2^{(n)},\vert\tau_3^{(n)}\vert)  $
in (a) and $(\tau_5^{(n)},\tau_3^{(n)})  $ in (b), from
 a Monte Carlo study
 based on $N=50000$  replications of  a dataset generated under Design 4, where the true tensor and the Fisher information are isotropic.
 The histogram density estimators are compared with  theoretical limit  densities (black continuous curves), which are
 uniform on the vertical axes and  $\chi^2_5$ on the horizontal axes. The best fitting gamma densities
 (red broken curves) are also shown,  with shape parameter $2.4924$ and scale parameter $1.9986 $
 in (a) and 
 with shape parameter  $2.4993$ and scale parameter $1.9896 $ in (b).
}\label{fig:sphericity:test:162}
\end{figure}
   
\begin{figure}[!tbp]
\caption*{ Design 5,  sphericity  statistics. }
~
\begin{subfigure}[b]{0.9\textwidth}
\centering
\includegraphics[width=4.5in, height=2.5in]{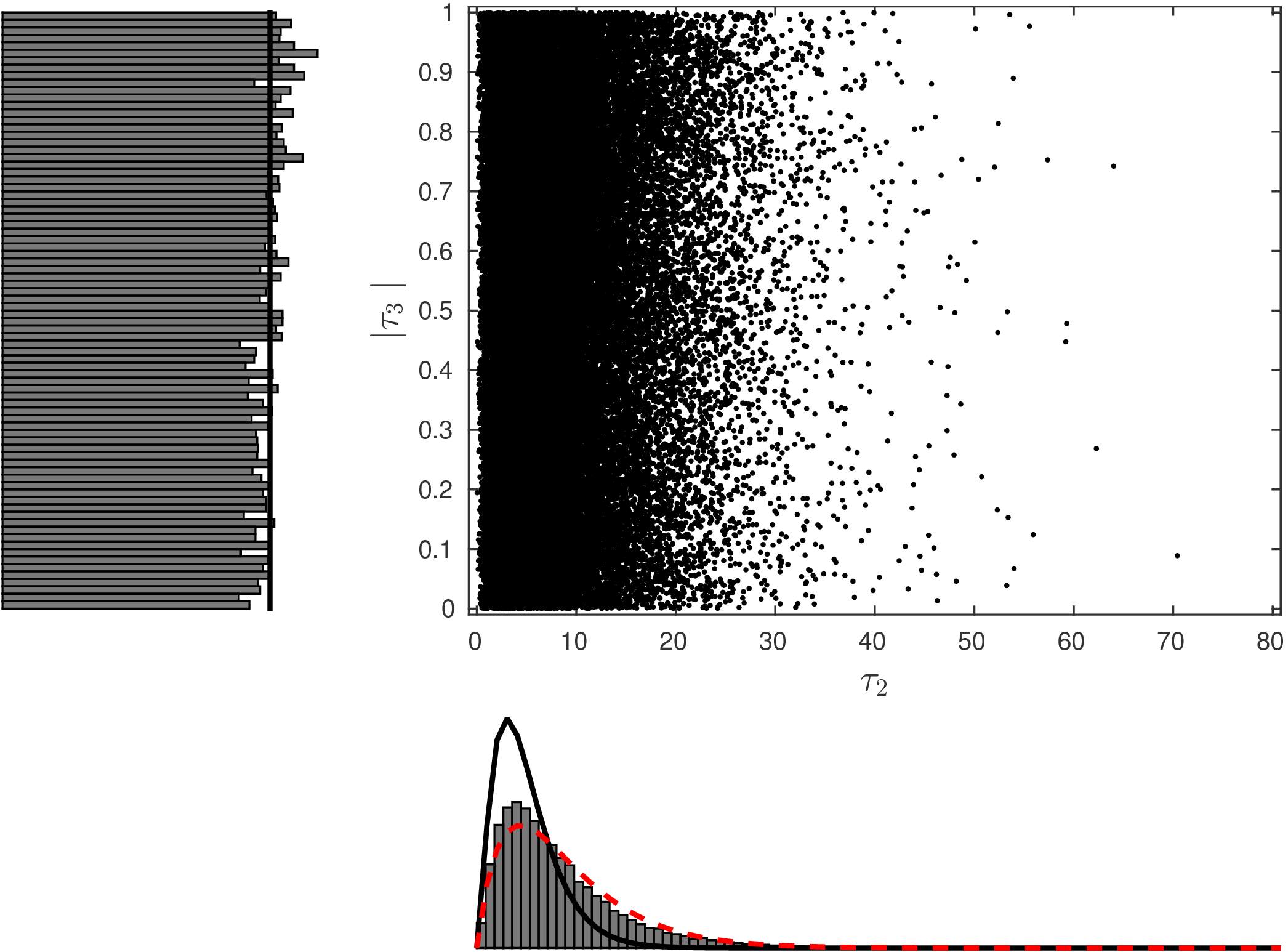}
\caption{}\label{fig:sphericity:test:32x3a}
\end{subfigure}
~
 
 \begin{subfigure}[b]{0.9\textwidth}
\centering
\includegraphics[width=4.5in, height=2.5in]{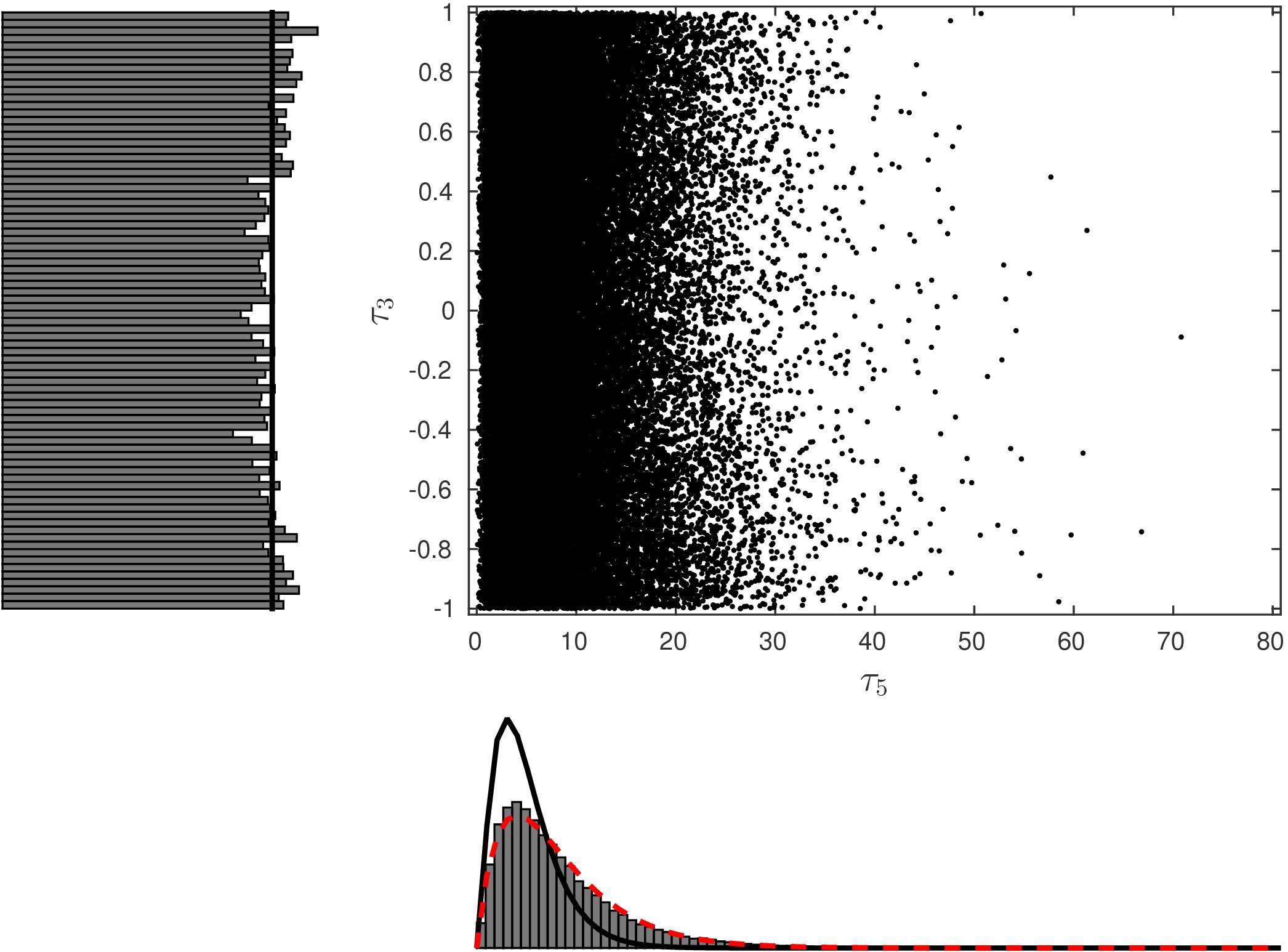}
\caption{ }\label{fig:sphericity:test:32x3b}
\end{subfigure}
\caption{Scatterplot of the eigenvalue statistics    $(\tau_2^{(n)},\vert\tau_3^{(n)}\vert)  $
in (a) and $(\tau_5^{(n)},\tau_3^{(n)})  $ in (b), from
 a Monte Carlo study
 based on $N=50000$  replications of  a dataset generated under Design 5, with isotropic true tensor and  anisotropic Fisher information.
 The histogram density estimators are compared with  theoretical limit  densities (black continuous curves), which are
 uniform on the vertical axes and  $\chi^2_5$ on the horizontal axes. The best fitting gamma densities
 (red broken curves) are also shown,  with shape parameter $1.9576$ and scale parameter $4.5494$ in (a) and 
 with shape parameter  $1.9565$ and scale parameter $4.2094$ in (b).
}\label{fig:sphericity:test:32x3}
\end{figure} 
   
\begin{figure}[]
\centering
\begin{minipage}[r]{\hsize}
\centering
\vspace{0pt}
{\footnotesize
\begin{tabular}[width=\textwidth]{ |c c c | c c c| } 
\hline
     $u_x$ &$u_y$ & $u_z$  &   $u_x$ &$u_y$ & $u_z$  \\ \hline 
   -0.5000  & -0.7071  &    0.5000     &      0.7071  &   -0.5261   &   0.4725    \\
   -0.5000  &  0.7071  &    0.5000     &     -0.7071  &   -0.0002   &   0.7071    \\
    0.7071  & -0.0000  &    0.7071     &     -0.7071  &    0.5261   &   0.4725    \\
   -0.6533  & -0.7071  &    0.2706     &      0.7071  &    0.5261   &   0.4725    \\
   -0.2087  & -0.7071  &    0.6756     &      0.4725  &    0.5261   &   0.7071    \\
    0.0197  & -0.7071  &    0.7068     &     -0.7071  &    0.0078   &   0.7071    \\
    0.4212  & -0.7071  &    0.5679     &     -0.6364  &    0.6436   &   0.4252    \\
    0.6899  & -0.7071  &    0.1549     &     -0.7060  &    0.0547   &   0.7060    \\
   -0.6535  & -0.7069  &    0.2707     &     -0.2929  &    0.6436   &   0.7071    \\
   -0.2929  & -0.6436  &    0.7071     &      0.2929  &    0.6436   &   0.7071    \\
    0.2945  & -0.6436  &    0.7064     &      0.7071  &    0.0078   &   0.7071    \\
    0.5150  & -0.7061  &    0.4861     &      0.7071  &    0.6436   &   0.2929    \\
    0.7071  & -0.6436  &    0.2929     &     -0.7063  &    0.0489   &   0.7063    \\
   -0.7071  & -0.5261  &    0.4725     &      0.0347  &    0.7071   &   0.7063    \\
   -0.4725  & -0.5261  &    0.7071     &      0.7071  &    0.0115   &   0.7071    \\
    0.5555  & -0.5261  &    0.6439     &      0.7071  &    0.7071   &   0.0000    \\
    \hline
\end{tabular}
}
\end{minipage}
\hfill
\begin{minipage}[l]{0.60\hsize}
\centering
\vspace{0pt}
\begin{figure}[H]
\centering
\includegraphics[width=\textwidth]
{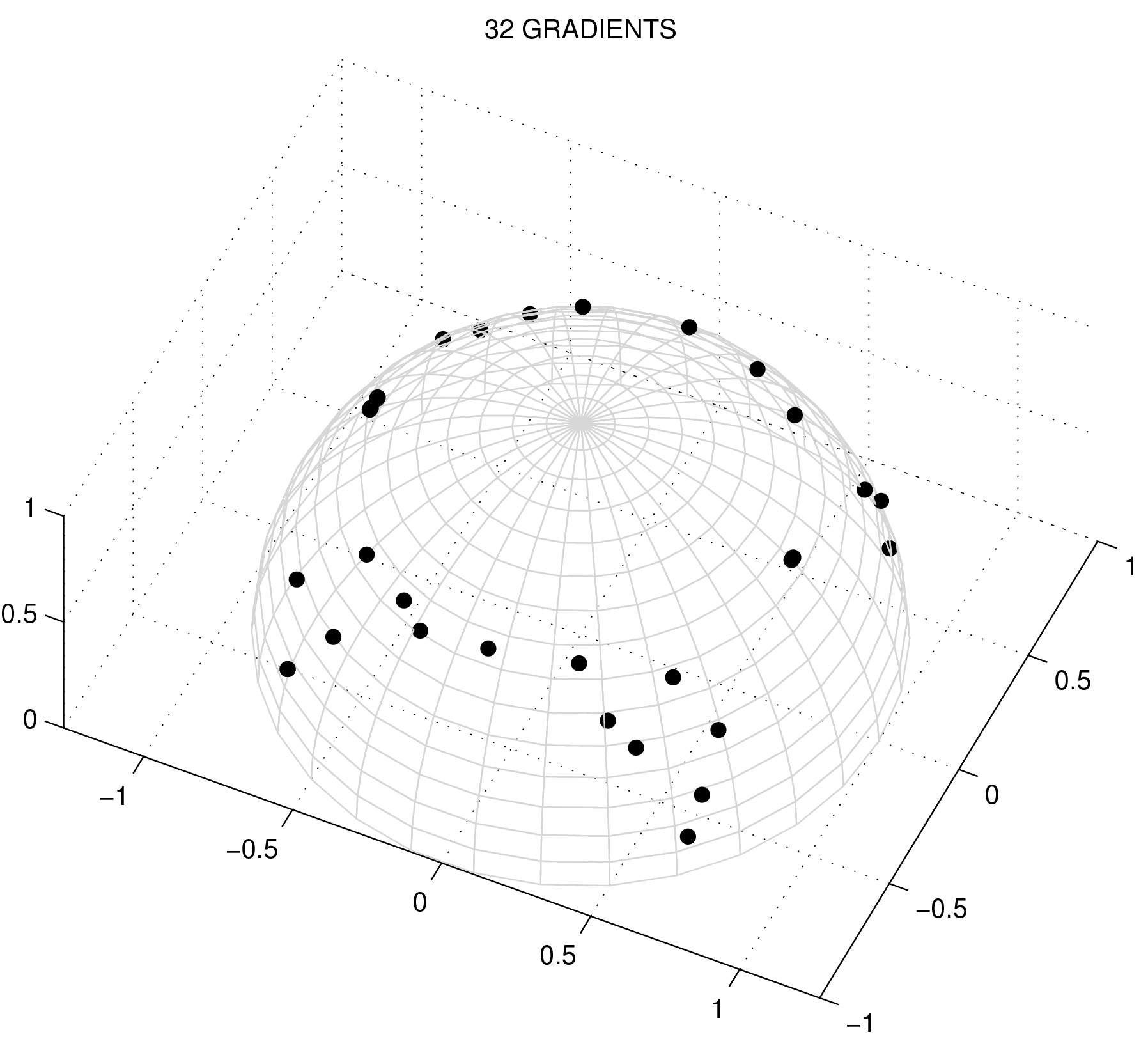}
\end{figure}
 \end{minipage}
 \caption{  The $32$ gradients table  used by default with the commercial 3T  Philips Achieva MR-scanner. }\label{gradient:table}
 \end{figure}
\FloatBarrier
\newpage
\beginsupplement
\begin{center}
{\bf Supplementary Materials } 

{\bf Title: Eigenvalues of random  matrices with isotropic Gaussian noise and the design of  Diffusion Tensor Imaging experiments. }

{ Authors: Dario Gasbarra, Sinisa Pajevic, Peter J. Basser.}

{ Contents: Appendix.}
\end{center}
\section{APPENDIX}\label{appendix:a}
 \subsection{Proof of Theorems \ref{spectral:clustering},\ref{eigenvector:asymptotics},\ref{HCIZ-asymptotics}}
 We follow the line of proof of Theorem 3.1 in \cite{chattopadhyay}, (see also \cite{muirhead}), which deals with the
 eigenvalues of rescaled  Wishart random matrices with growing degrees of freedom, and we   generalize it  
to the case of random matrices with asymptotically isotropic Gaussian noise, without the positivity assumption. 
By Sch\'effe's theorem (see \cite{van_der_vaart}), to prove convergence in distribution it is enough to show pointwise almost sure convergence
            of the densities to a probability density, and we achieve that
          by   using Laplace approximation.
Note first that when the mean matrix $\bar D$  is isotropic with 
 equal eigenvalues $\bar \gamma_1=\bar \gamma_2 = \dots =\bar  \gamma_m$, 
 \begin{align} \label{HCIZ:symmetric}
 {\mathcal I}_m( \gamma, \bar\gamma) = \exp( \gamma \cdot \bar\gamma) \;,
 \end{align}
and simply because on  a probability space the $L^{\mu}$ norm of a
  random variable converges to the $L^{\infty}$ norm as $\mu\to\infty$, it follows that
 \begin{align*}
    \lim_{\mu\to\infty} \frac 1 {\mu} \log {\mathcal I}_m(\mu \bar \gamma,\gamma)= \sup_{{\mathcal O}(m) } 
   \bigl  \{ \tr(O G O^{\top} \bar G )  \bigr \}= \gamma
    \cdot \bar \gamma  \; , 
  \end{align*}
  where the supremum is attained by every  orthogonal matrix $O\in {\mathcal K}_{\bar\gamma}$
  with  block-diagonal structure  \eqref{block:diagonal structure}
corresponding to the multiplicities of the $\bar D$-eigenvalues. 

We continue from  the joint density  \eqref{HCIZ:firststep} of  eigenvalues and eigenvectors, using
the representation
\begin{align}
R=\bar O^{\top} O = \check R  \hat R =
\check R \exp(S )
\in \bar O^{\top} {\mathcal O}(m)^+\;. \end{align}
In the new coordinates, since $\check R^{\top} \bar G \check R = \bar G$, the joint density of $(G,R)$ with respect to $d\gamma\times H_m(dR)$ is given by 
\begin{align*}
q_{\bar \gamma}( G, R)=q_{\bar \gamma}( G, \check R \hat R  )=q_{\bar \gamma}( G,  \hat R  ) 
\end{align*} 
which does not depend on $\check R$.
Moreover,
by using the properties of the wedge product,
\begin{align*} & 
\bigl( R^{\top} dR\bigr)^{\wedge} =  \bigl(  \hat R^{\top}  \check R^{\top} (  \check R \hat dR +d\check R \hat R ) \bigr)^{\wedge}
 =  \bigl( \hat R \hat dR+ \hat R^{\top}\check R^{\top}  d\check R \hat R    \bigr)^{\wedge} & \\ &
 =  \bigl( \hat R \hat dR  \bigr)^{\wedge}\bigl(  \hat R^{\top}\check R^{\top}  d\check R \hat R \bigr)^{\wedge}  
 =  \bigl( \hat R \hat dR  \bigr)^{\wedge} \bigl( \check R^{\top}  d\check R \bigr)^{\wedge}  
 =  \bigl( \hat R  d\hat R  \bigr)^{\wedge} \bigwedge_{i=1}^k \bigl( \check R^{\top}_{(i,i)} 
 d\check R_{(i,i)} \bigr)^{\wedge} \; . 
\end{align*}
Therefore, the blocks $(\check R_{(i,i)},i=1,\dots,k)$ of $\check R$ are  independent from the eigenvalues $\gamma$ and 
distributed  as the product of  Haar probability measures 
$H_{m_i}\bigl( d\check R_{(i,i)}\bigr)$
on the respective orthogonal groups ${\mathcal O}(m_i)$ corresponding
to the $\bar D$-eigenspaces, with the constraint $\bar O \check R \hat R\in {\mathcal O}(m)^+$.
 Since for every fixed $\hat R \in {\mathcal C}_{\gamma}$, by symmetry
 \begin{align*}
   \int_{ {\mathcal O}(m_1)\times \dots \times {\mathcal O}(m_k)  }  {\bf 1}\bigl( \bar O \check R \hat R  \in  {\mathcal O}(m)^+   \bigr) 
   \bigwedge_{i=1}^k \bigl( \check R^{\top}_{(i,i)}  d\check R_{(i,i)} \bigr)^{\wedge} = 2^{-m }
   \prod_{i=1}^k \mbox{Vol}\bigl( {\mathcal O}(m_i) \bigr) \;,
 \end{align*}
after integrating out $\check R$ we see that  $q_{\bar \gamma} (G,\hat R)$ in \eqref{HCIZ:firststep} 
is also the joint density of $(\gamma,\hat R)$ on 
$\{\gamma\in \R^m: \gamma_1 >\gamma_2 > \dots > \gamma_m\}\times {\mathcal C}_{\bar\gamma}$ with respect to the product measure
  $d\gamma_1 \times d\gamma_2 \times \dots \times d\gamma_m \times H_{\mathcal C_{\bar\gamma}} (d\hat R)$, where
\begin{align*} 
H_{\mathcal C_{\bar\gamma}} (d\hat R) &=  \mbox{Vol}({\mathcal C}_{\bar\gamma})^{-1}   ( \hat R^{\top} d\hat R \bigr)^{\wedge},
\mbox{ with }  & \\  \mbox{Vol}({\mathcal C}_{\bar\gamma})&=\int_{{\mathcal C}_{\bar\gamma} }
\bigl(\hat R^{\top}  d\hat R \bigr)^{\wedge} = 
\frac{  \mbox{ Vol} ( {\mathcal O}(m) ) } { \prod_{j=1}^k \mbox{Vol}\bigl(  {\mathcal O}(m_j) \bigr) } 
= \frac{ Z_m(1,0) }{ C_m(1,0) }\prod_{j=1}^k \frac{ C_{m_j}(1,0) }{ Z_{m_j}(1,0) }
&\\ &= \pi^{ \bigl( m^2-\sum\limits_{j=1}^k m_j^2\bigr)/4}
 \frac{  \prod\limits_{i=1}^k \prod\limits_{j=1}^{m_j} \Gamma( j/2) }{  \prod\limits_{\ell=1}^{m} \Gamma( \ell/2) }\; , &
&\end{align*}
is  the Haar probability measure of ${\mathcal C}_{\bar\gamma}$.
Since rows and columns of  $\hat R$ are normalized eigenvectors,
 $( R_{ij}^2: 1\le i,j \le m)$ is a {\it doubly stochastic} matrix, with
\begin{align*}
 \sum_{\ell=1}^m  \hat R_{i,\ell}^2 =\sum_{\ell=1}^m  \hat R_{\ell,j}^2 =1 \quad  \mbox{ for } 1\le i,j\le m,
\end{align*}
and by substitution
\begin{align} \label{trace:expression}
  \tr\bigl(  G\hat R^{\top} \bar G \hat R\bigr) & = \sum_{ij} \bar \gamma_i \gamma_j \hat R_{ij}^2
  = \sum_{j} \gamma_j \biggl\{  \sum_{i=1}^{\ell_{k-1}} \bar \gamma_i \hat R_{ij}^2 + \bar \gamma_k 
  \biggl( 1
  -\sum_{i=1}^{\ell_{k-1} } \hat R_{ij}^2 \biggr) \biggr\} & \\ &= 
   \sum_{j} \gamma_j \biggl\{ \bar \gamma_k+  \sum_{i=1}^{\ell_{k-1}} (\bar \gamma_i -\bar \gamma_k)\hat  R_{ij}^2 
 \biggr\} \; . &
\end{align}
For any fixed $\gamma_1> \gamma_2 > \dots > \gamma_m$,
the  maximum of \eqref{trace:expression} 
over ${\mathcal C}_{\bar \gamma}$
is attained 
 at $\hat R=\mbox{I}$ corresponding to  $\hat S=0$, and as $\mu\to \infty$
 the density of $\hat R$ will concentrate around this maximum. 
We apply the Laplace approximation method and take the second order expansion of the matrix exponential as $\hat S\to 0$,
\begin{align} \nonumber &
\hat R= \exp(\hat S)=  \mbox{I} + \hat S +  \hat S ^2 /2 + o( \parallel \hat S\parallel^2 )  , \mbox{ with }
 & \\ & \label{R:expansion}
  \hat R_{ii} =  1 - \frac 1 2 \sum_{j=1}^m \hat S_{ij}^2+ o( \parallel \hat S \parallel^2),
 \; \hat R_{ii}^2 =  1 -  \sum_{j=1}^m \hat S_{ij}^2+ o( \parallel \hat S \parallel^2),  & \\ \nonumber & \mbox{ and for $i\ne j$ }
   \hat R_{ij}^2= \hat S_{ij}^2 + o( \parallel \hat S^2 \parallel) .
&\end{align}
By substituting \eqref{R:expansion} in   \eqref{trace:expression}, 
\begin{align*}
   \tr\bigl(  G\hat R^{\top} \bar G \hat R\bigr) - \gamma \cdot \bar \gamma=
 - \sum_{i, j} ( \gamma_i - \gamma_j )( \bar \gamma_i- \bar\gamma_j) \hat S_{ij}^2 + o( \parallel \hat S \parallel^2)
  \; , \quad \mbox{  as $\hat S\to 0$, }
\end{align*}
where in the sum the terms indexed by   $(i,j)$ corresponding to identical $\bar D$-eigenvalues vanish.
We now take a  sequence of parameters $\mu_n=n$, and $\lambda_n$ such that  $-2/m < \lambda_n/n \to \lambda$,  as $n\to \infty$.
After changing of variables,  we approximate the density  of  $(\gamma,\hat S)$
with respect to the volume measure  $\bigl(d\hat S \bigr)^{\wedge}=\bigl(\hat R^{\top} d\hat R\bigr)^{\wedge}$, as
\begin{align}\label{gaussian:approximation:density} & 
  q_{\bar \gamma}^{(n)}(  \gamma,\hat S ) \sim     \sqrt{  1+ \lambda m/2 } \prod_{l=1}^k Z_{m_l}(1,0)
  V(\gamma) 
\exp\biggl( -  n \sum_{i,j=1}^m  \biggl(  \delta_{ij }+
\frac{\lambda} 2 \biggr) (\gamma_i -\bar \gamma_i)(\gamma_j-\bar \gamma_j) \biggr)
&\\ \nonumber&
 \times n^{m(m+1)/4} \frac{ C_m(1,0)  } 
 {  \prod_{l=1}^k  C_{m_l}(1,0) }
 \prod_{h=1}^{k-1} \prod_{i=\ell_{h-1}+1}^{\ell_h} \prod_{j=\ell_h+1}^m 
\exp\biggl( - 2n ( \gamma_i- \gamma_j)  ( \bar \gamma_i- \bar\gamma_j) \hat S_{ij}^2 \biggr)  , &
& \end{align}
where   $x_n \sim y_n$ when $\lim\limits_{n\to\infty} x_n/y_n =1$.

For fixed $\gamma_1> \gamma_2> \dots > \gamma_m$,  $(\gamma_i - \gamma_j)( \bar\gamma_i- \bar\gamma_j) \ge 0$,
and by integrating  $\hat S_{ij}$ over $\R$  for each $i<j$ with $\bar\gamma_i > \bar \gamma_j$
we obtain the  Laplace approximation
\begin{multline}\label{laplace:marginal} 
 q^{(n)}_{\bar\gamma}(\gamma) \sim  \\ n^{ \sum_{i=1}^k m_i(m_i+1)/4 }   
   \sqrt{  1+ \lambda m/2 } \prod_{l=1}^k Z_{m_l}(1,0)
  V(\gamma)\biggl( \frac   {\pi} 2 \biggr)^{m^2/4 -\sum_{i=1}^k m_i^2/4 }
   \frac{ C_m(1,0)  }  
{  \prod_{l=1}^k  C_{m_l}(1,0) } \times \\ 
\exp\biggl( -  n \sum_{i,j=1}^m  \biggl(  \delta_{ij }+
\frac{\lambda} 2 \biggr) (\gamma_i -\bar \gamma_i)(\gamma_j-\bar \gamma_j) \biggr)
 \prod_{h=1}^{k-1} \prod_{i=\ell_{h-1}+1}^{\ell_h} \prod_{j=\ell_h+1}^m \bigl\{ (\gamma_i-
 \gamma_j) (\bar\gamma_i -\bar \gamma_j) \bigr\}^{-1/2} , 
\end{multline}
and by comparing the right hand side of \eqref{laplace:marginal} with
the eigenvalue density \eqref{alt:eigdensity}  we obtain \eqref{HCIZ:asymptotic:eq}, proving Theorem \ref{HCIZ-asymptotics}.

By the further rescaling \eqref{gaussian:approximation:density} with
\begin{align*}\quad \xi_i = (\gamma_i- \bar \gamma_i) \sqrt{n}, i=1,\dots,m, \quad
\theta_{ij}= \widehat S_{ij} \sqrt n,  \;  1\le i < j \le m \mbox{ with } \; \bar\gamma_i > \bar \gamma_j, 
\end{align*} we obtain  
\begin{align*}& q_{\bar \gamma}^{(n)}(  \xi ,\theta )  \sim
q_{\bar \gamma}(  \xi ) \prod_{h=1}^{k} \prod_{i=\ell_{h-1}+1}^{\ell_h}
\prod_{j=\ell_h+1}^m  q_{\bar \gamma}( \theta_{ij}) 
\end{align*}
where the asymptotic density of the eigenvalue fluctuations is given by
\begin{multline*}  q_{\bar \gamma}(  \xi )= \\ 
 \sqrt{  1+ \lambda m/2 }
\prod_{r=1}^{k} \biggl\{ Z_{m_r}(1,0)  
\prod_{\ell_{r-1}+1\le v < w \le l_{r}} ( \xi_v-\xi_w) 
\exp\biggl( - \sum_{i,j=1}^m  \biggl(  \delta_{ij }+
\frac{\lambda} 2 \biggr) \xi_i \xi_j \biggr) \biggr\} 
 \end{multline*}
and for $i< j$ with $\bar\gamma_i > \bar\gamma_j$ the fluctuations $\theta_{ij}$ are asymptotically independent 
with respective Gaussian densities 
\begin{align} \label{angle:fluctuations}  q_{\bar \gamma}( \theta_{ij})=
\sqrt{ \frac 2 {\pi} }(\bar \gamma_i -\bar\gamma_j)\exp\biggl( - 2   ( \bar \gamma_i- \bar\gamma_j)^2  \theta_{ij}^2 
\biggr)\; ,
\end{align} 
 which completes the proof of  Theorem \ref{eigenvector:asymptotics}.
 
In order to study the fluctuations of the cluster barycenters and  eigenvalue distribution
within clusters, we  change variables again  by using the linear maps
\begin{align*} 
 T_i\bigr(\xi_{\ell_{i-1}+1}, \dots \xi_{\ell_{i}-1},\xi_{\ell_i} \bigr )=
\bigl( \zeta_{\ell_{i-1}+1}, \dots , \zeta_{\ell_i-1}, \widetilde \xi_i \bigr ) , \quad 1\le i  \le k
\end{align*}
with Jacobian determinants $\det( \nabla T_i)=1/m_i$,
where, we have denoted the cluster barycenters as 
\begin{align*}
   \widetilde \xi_i 
   =   \frac{1}{ m_i}  \sum_{j=\ell_{i-1}+1}^{\ell_i}  \xi_j
   ,\quad  1\le i \le k  \; .
   \end{align*} 
   and  
   \begin{align*}
     \zeta_j=  \xi_j - \widetilde \xi_i, \quad  1\le i \le k, \quad  \ell_{i-1}< j \le \ell_i   
   \end{align*}
 are  the differences between the eigenvalue  and their  cluster barycenters.
 In these new random variables the asymptotic eigenvalue fluctuation density factorizes as
as   \begin{align*} q_{\bar \gamma}(  \xi )=     
     q( \widetilde \xi_1,\dots \widetilde \xi_k) \prod_{i=1}^k  q_{m_i}( 
     \zeta_{\ell_{i-1}+1},\dots ,\zeta_{\ell} ) \;,
     \end{align*}
     where the cluster barycenters have Gaussian density \eqref{center:density:asymptotic},
which is also the density of $\widetilde X$ in 
   \eqref{cluster:center:distribution}, 
and the differences 
$(\zeta_{\ell_{i-1}+1},\dots ,\zeta_{\ell})$
between 
   the $m_i\times m_i$-GOE eigenvalues and their barycenter
   have degenerate densities \eqref{cluster:density:asymptotic}.


\begin{thebibliography}{99}

  
  
  \bibitem
  {an} An C., Chen X., Sloan I.H. and Womersley R.S. (2010).
  Well conditioned spherical designs for integration and interpolation
   on the two-sphere. SIAM J. Numer. Anal. 48 (6) 2315-2157.
  
\bibitem
{anderson}Anderson G.A.(1965).
An asymptotic expansion for the distribution of the 
latent roots of the estimated covariance matrix. 
Ann. Math. Statist. 36 1153-1173. 


\bibitem
{AGZ} Anderson G.W., Alice Guionnet A., Ofer Zeitouni (2010).
An Introduction to Random Matrices. Cambridge University Press.

\bibitem
{bannai}
Bannai E., Bannai E. (2009).
A survey on spherical designs and algebraic combinatorics
on spheres. European Journal of Combinatorics 30 1392-1425.

\bibitem
{basser1994a} Basser, P.J., Mattiello, J., LeBihan, D. (1994). MR
diffusion tensor spectroscopy and imaging. Biophysical journal, 66(1),
259.

\bibitem
{basser1994b} Basser, P.J., Mattiello, J., LeBihan, D. (1994). Estimation of the effective self-diffusion tensor from the NMR spin echo. Journal of Magnetic Resonance, Series B, 103(3), 247-254.

\bibitem
{basser1995} Basser P. J. (1995)
Inferring Microstructural 		Features and the Physiological State
of Tissues from Diffusion Weighted Images.
NMR in Biomedicine 8 333-344.

\bibitem
{basser1997} Basser P. J. (1997)
New Histological and Physiological Stains Derived from Diffusion-Tensor
MR-Images.  Imaging Brain Structure and Function,
Annals of the New York Academy of Sciences Vol. 820 123-138.

\bibitem
{basser:artefact} Basser P.J. Pajevic S. (2000).
Statistical Artifacts in Diffusion Tensor MRI (DT-MRI)
Caused by Background Noise. Magnetic Resonance in Medicine 44:41-50.

\bibitem
{basser-pajevic03}Basser P.J., Pajevic S. (2003).
A normal distribution for tensor-valued random
variables: applications to diffusion tensor MRI.  
\textit{ IEEE Trans. Med. Imag.} 22 (7) 785-794.  

\bibitem
{basser:pajevic2003}Basser P.J., Pajevic S. (2003).
Dealing with uncertainity in Diffusion Tensor MR Data.
Israel Journal of Chemistry 43 129-144.

\bibitem
{basser-pajevic07}
Basser, P. J., Pajevic, S. (2007). Spectral decomposition of a 4th-order covariance tensor: Applications to
diffusion tensor MRI. Signal Processing, 87(2), 220-236.

\bibitem
{batchelor} Batchelor P.G., Atkinson D.,Hill D.L.G.,Calamante F.,Connelly A. (2003).
Anisotropic Noise Propagation in Diffusion Tensor MRI
Sampling Schemes. Magnetic Resonance in Medicine 49:1143-1151.


\bibitem
{chattopadhyay}Chattopadhyay A.K., Pillai K.C. (1973). Asymptotic expansions for the distributions
of characteristic roots when the parameter matrix has several multiple roots. 
Multivariate analysis, III (Proc. Third Internat. Sympos., Wright State Univ., Dayton, Ohio, 1972), 117-127. Academic Press, New York. 

\bibitem{chiani} 
Chiani M. (2014).
Distribution of the largest eigenvalue for real Wishart and
Gaussian random matrices and a simple approximation for
the Tracy–Widom distribution Journal of Multivariate Analysis 129 69-81.


\bibitem
{chikuse} Chikuse Y. (2003).
Statistics on Special Manifolds. Springer Lecture Notes in Statistics 174. 


\bibitem
{clement_et_al}Clement-Spychala M.E., Couper D., Zhu H., Muller K.E. (2010).
Approximating the Geisser-Greenhouse sphericity estimator and its applications to
diffusion tensor imaging. Stat Interface 3 (1) 81-90.


\bibitem
{delsarte}
Delsarte P.,  Goethals J.M., and  Seidel J.J. (1977). Spherical Codes and Designs.
Geometriae Dedicata 6 363-388. 



\bibitem
{drton2009} Drton M. (2009). Likelihood ratio tests and singularities.
Annals of Statistics 37 (2) 979-1012.


\bibitem
{drton} Drton M, Xiao H. (2016). Wald tests of singular hypothesis. Bernoulli 22 (1) 38-59.

\bibitem
{dyson} Dyson F.J. (1962). A Brownian-motion model for the eigenvalues of a random matrix. J. Mathematical Phys. 3 1191-1198.


\bibitem
{edelman}Edelman A. (1989) Eigenvalue and Condition Numbers of Random Matrices. Ph.D. Thesis, MIT.



\bibitem
{farrell}Farrell R.H. (1985). Multivariate Calculation
Use of the Continuous Groups. Springer.

\bibitem
{forrester}Forrester P.J. (2010). Log-Gases and Random Matrices. 
London Mathematical Society Monographs, Princeton University Press.

  
   
  \bibitem
  {guionnet}Guionnet A. (2012). Large Random Matrices: Lectures on  Macroscopic Asymptotics.  In
  Noncommutative Probability and Random Matrices at Saint-Flour,
 Biane, Philippe, Guionnet, Alice, Voiculescu, Dan-Virgil (auth.), 170-466.
  
  \bibitem
  {hikami_brezin} Hikami S, Br\'ezin E (2006).
  WKB-Expansion of the Harish-Chandra-Itzykson-Zuber Integral
for Arbitrary $\beta$.  
Progress of Theoretical Physics 116 (3) 441-502


\bibitem
{idier} Idier J.,Collewet G. (2014).Properties of Fisher information for Rician
distributions and consequences in MRI. Preprint hal-01072813.


\bibitem
{james}James A.T. (1964). Distributions of matrix variates and latent roots derived from normal samples. 
Ann. Math. Statist. 35 475-501. 

\bibitem
{liu} Liu J., Gasbarra D., Railavo J. (2016). Fast Estimation of Diffusion Tensors under Rician noise by the EM algorithm. Journal of Neuroscience Methods 257 147-158.

\bibitem
{magnus}
Magnus J.R. \&  Neudecker H. (1999).  
Matrix Differential Calculus with applications  in Statistics and Econometrics,  Wiley.
 


\bibitem
{mallows} Mallows, C.L. (1961). Latent vectors of random symmetric matrices. Biometrika 48 133-149.

 \bibitem
 {metha}Mehta M.L. (2004). Random Matrices, 3rd Edition.  Elsevier.

\bibitem
{muirhead}Muirhead R.J. (1984). Aspects of Multivariate Statistics. Wiley.



\bibitem
{pajevic2003}
Pajevic, S., Basser, P. J. (2003). Parametric and non-parametric statistical analysis of DT-MRI data. Journal of magnetic resonance, 161(1), 1-14.



\bibitem
{pajevicbasser2010}  Pajevic S., Basser P.J. (2010). A joint PDF for the Eigenvalues and Eigenvectors of a Diffusion Tensor.
Proc. Intl. Soc. Mag. Reson. Med. 18 303.

\bibitem
{pierpaoli_et_al:1994} Pierpaoli C., Infante I., Mattiello J., Di Chiro G., Le Bihan D., Basser P.J. (1994).
Diffusion Tensor Imaging of Brain White Matter Anisotropy.
Proceedings of the Society of Magnetic Resonance, Vol. 2 1038.

\bibitem
{schwartzman}Schwartzman A., Mascarenhas W.F.,
 Taylor J.E. (2008). 
 Inference for eigenvalues and eigenvectors
 of Gaussian symmetric matrices.
Annals of Statistics 36 (6) 2886-2919.

\bibitem{schwartzman2010}
Schwartzman A., Dougherty R.F., Taylor J.E. (2010). Group Comparison of
Eigenvalues and Eigenvectors of Diffusion Tensors, Journal of the American Statistical Association 105:490 588-598.

\bibitem
{takemura}Takemura A.(1984) Zonal Polynomials. Institute of Mathematical Statistics  Lecture notes, Vol. 4.

\bibitem
{tao}Tao T. (2012). Topics in Random Matrix Theory. American Mathematical Society.

\bibitem
{tao_blog}Tao T. (2013). The Harish-Chandra-Itzykson-Zuber integral formula.
 
\noindent What's new (blog) 
\url{https://terrytao.wordpress.com/2013/02/08/the-harish-chandra-itzykson-zuber-integral-formula/}.

\bibitem
{van_der_vaart}  Van Der Vaart A.W. (2000). Asymptotic Statistics, Cambridge University Press.


\bibitem
{watanabe}Watanabe S.(2009). Algebraic geometry and statistical learning theory. 


\bibitem
{ibrahim} Zhu H., Zhang H., Ibrahim J.G., Peterson B.S., (2007). 	
Statistical analysis of diffusion tensors in diffusion-weighted magnetic resonance imaging data. JASA  102 (480) 1085-1102.  
  
\end{thebibliography}
   \end{document}